\documentclass[10pt]{article}
\usepackage[english]{babel}
\usepackage[utf8]{inputenc}
\usepackage[T1]{fontenc}
\usepackage{lmodern}
\usepackage{amssymb,amsmath,amsfonts}
\usepackage{xspace}
\usepackage{graphicx}
\usepackage{hyperref}
\usepackage{url}
\usepackage{stmaryrd}
\usepackage[babel=true]{csquotes}
\usepackage{amsthm}
\usepackage{subfigure}
\usepackage{mathrsfs}
\usepackage{float}
\usepackage[inner=2.5cm,outer=3.5cm]{geometry}

\newcommand{\strong}{\emph}

\theoremstyle{plain}
\newtheorem{thm}{Theorem}[section]
\newtheorem{lm}[thm]{Lemma}
\newtheorem{prop}[thm]{Proposition}
\newtheorem{cor}[thm]{Corollary}
\newtheorem{conj}[thm]{Conjecture}

\theoremstyle{definition}
\newtheorem{dfn}[thm]{Definition}
\newtheorem{rmk}[thm]{Remark}

\def\dpar#1#2{\frac{\partial #1}{\partial #2}}

\def\Hil{\mathcal{H}\xspace}

\def\Z{\mathbb{Z}\xspace}
\def\R{\mathbb{R}\xspace}
\def\C{\mathbb{C}\xspace}

\def\P{\mathbb{P}\xspace}
\def\T{\mathbb{T}\xspace}
\def\S{\mathbb{S}\xspace}

\def\scal#1#2{\left\langle #1,#2\right\rangle}

\def\classe#1#2{\mathcal{C}^{#1}#2}

\def\bigO#1{\mathcal{O}\mathopen{}\left(#1\right)\mathclose{}}

\hypersetup{pdfborder={0000}, colorlinks=true, linkcolor=blue, citecolor=red}

\newcommand\blfootnote[1]{
  \begingroup
  \renewcommand\thefootnote{}\footnote{#1}
  \addtocounter{footnote}{-1}
  \endgroup
}

\begin{document}

\title{\sc{Symplectic geometry and spectral properties of classical and quantum coupled angular momenta}}

\author{Yohann Le Floch\hspace{2cm} \'Alvaro Pelayo}

\date{}

\maketitle

\begin{abstract}
We give a detailed study of the symplectic geometry of a family of integrable systems obtained by coupling two angular momenta in a non trivial way. These systems depend on a parameter $t \in [0,1]$ and exhibit different behaviors according to its value. For a certain range of values, the system is semitoric, and we compute some of its symplectic invariants. Even though these invariants have been known for almost a decade, this is to our knowledge the first example of their computation in the case of a non\--toric semitoric system on a compact manifold (the only invariant of toric systems is the image of the momentum map). In the second part of the paper we quantize this system, compute its joint spectrum, and describe how to use this joint spectrum to recover information about the symplectic invariants.
\end{abstract}



\section{Introduction}

\subsection{A fundamental model in physics: coupled angular momenta}

One of the most mathematically interesting and physically relevant finite dimensional integrable systems with two degrees of freedom is produced by coupling two angular momenta in a non trivial fashion, depending on a parameter $t \in [0,1]$ encoding the different non trivial ways in which the coupling can occur. To describe the system precisely, let $R_2 > R_1 > 0$ be two positive real numbers, which represent the norms of the two angular momenta. Endow $M = \S^2 \times \S^2$ with the coordinates $(x_1,y_1,z_1,x_2,y_2,z_2)$ and the symplectic form
$\omega = -(R_1 \omega_{\S^2} \oplus R_2 \omega_{\S^2}),$ where $\omega_{\S^2}$ is the standard symplectic form on the sphere $\S^2$ (the one giving area $4\pi$). The \strong{coupled angular momenta system} is given by the map $F = (J,H)$ where $J$ and $H$ are defined as
\begin{equation} \begin{cases} J = R_1 z_1 + R_2 z_2, \\ H = (1-t) z_1 + t (x_1 x_2 + y_1 y_2 + z_1 z_2), \end{cases} \label{eq:moment_map}\end{equation}
with $t$ a parameter in $[0,1]$. This system has been introduced by Sadovski\'{i} and Zhilinski\'{i} \cite{SadZhi}.

 It is of great relevance in physics, since it can be used to model a variety of phenomena. Even though it is given by relatively simple formulae, its symplectic and analytic properties are very rich; the characteristics of the system can change drastically as the parameter varies, and it displays many interesting features for which the available  literature is minimal: passing from non\--degenerate to degenerate singularities, degeneration of elliptic fibers into nodal fibers, and so on. One of its most interesting features, and motivation for its study in \cite{SadZhi}, is that it exhibits non trivial monodromy for certain values of $t$; for these values of the parameter, we prove that it is a so-called \strong{semitoric integrable system} with one focus-focus critical value. These systems seem to be common in the physics literature, but from the mathematical point of view, only a few examples are known. To our knowledge, the only (non toric) system which was rigorously proven to be semitoric prior to the present paper was the so-called Jaynes-Cummings system \cite{PelVN}, whose phase space is $\S^2 \times \R^2$. Hence, our paper constitutes the first rigorous study of a semitoric system with compact phase space. Note that, in the time interval between the first and current arXiv versions of this manuscript, Hohloch and Palmer \cite{HohPal} generalized this system to obtain family of semitoric systems with two focus-focus points on $\S^2 \times \S^2$. The symplectic classification of semitoric systems has been understood only recently \cite{PelVuSemi,PelVuConstr}; it relies on five symplectic invariants. See \cite{PelVNAMS,PelAMS} for recent surveys on integrable systems and Hamiltonian symmetries.

\subsection{Past works}

To our knowledge, the computation of these invariants has been performed only in the case of the coupling of a spin and an oscillator, the so-called Jaynes\--Cummings system, so far, in which case some computations are already heavy, in spite of the presences of nice symmetries. Part of this computation was performed in \cite{PelVN}, and this study was recently completed with the computation of the full Taylor series invariant and of the twisting-index invariant by Alonso, Dullin and Hohloch \cite{AloDulHoh}. We should also mention that the computation of one of the invariants has been achieved in the case of the spherical pendulum \cite{Dul}, but the latter is not a semitoric system but a generalized semitoric system, see \cite{PelRatVu}.  

Our first goal is to completely prove that the system at hand is semitoric for some values of the parameter $t$. In addition, we also obtain a complete parameterization for the boundary of the image of the moment map for all $t$, which is quite remarkable. Our second goal is to compute the five symplectic invariants for the coupled angular momentum system, thus providing another example, the first one on a compact phase space; it turns out that it is a quite complicated task, since the system does not exhibit the same kind of symmetries as the spin-oscillator does. In the process of achieving these two goals, we try to give as many details as possible, so that this paper could serve as a starting point when one wants to compute the symplectic invariants for other semitoric systems. This choice sometimes leads to the presence of a lot of technical details, but we think that this is a necessary evil.

Our third goal is to quantize this system and compute the associated joint spectrum, with the help of Berezin-Toeplitz operators; besides being interesting in itself, this also constitutes a good example for a future general study of the description of the joint spectrum of commuting self-adjoint Berezin-Toeplitz operators near a focus-focus value of the underlying integrable system, in the spirit of the work of V{\~u} Ng{\d{o}}c \cite{VN} on pseudodifferential operators on cotangent bundles. This example can also give some insight on the spectral behaviour during the transition in which an elliptic-elliptic point becomes focus-focus.

\subsection{Main results}

The paper emphasizes the interplay between the symplectic geometry of a classical integrable system and the spectral theory of the associated semiclassical integrable system. Our main results concern both the classical coupled angular momentum system and its quantum counterpart.

\subsubsection{Symplectic geometry of classical coupled angular momenta}

We will prove later that the map $F:M \to \R^2$ determined by (\ref{eq:moment_map}) is the momentum map for an integrable system, which means that its components $J$ and $H$ Poisson-commute and that the associated Hamiltonian vector fields $X_J, X_H$ are almost everywhere linearly independent.

\begin{dfn}
An integrable system $F= (J,H): M \to \R^2$ on a connected four-dimensional symplectic manifold $(M,\omega)$ is said to be \strong{semitoric} if $J$ is proper and is the momentum map for an effective Hamiltonian circle action and $F$ has only non-degenerate singularities with no hyperbolic component (if only this last property is satisfied, the system is said to be \strong{almost toric}). A semitoric integrable system is said to be \strong{simple} if there is at most one focus-focus point in each fiber of $J$.
\end{dfn}

\begin{rmk}
This definition implies that in semitoric or almost toric systems only singularities of elliptic-elliptic, elliptic-transverse and focus-focus type can occur (see Section \ref{sect:classical} for more details about singularities of integrable systems).
 \end{rmk}

The symplectic classification of simple semitoric systems has been achieved by the second author and V{\~u} Ng{\d{o}}c \cite{PelVuSemi,PelVuConstr}, and relies on five invariants:\begin{enumerate}
\item the \strong{number} of focus-focus critical values of the system,
\item a family of \strong{convex polygons} obtained by unwinding the singular affine structure of the system,
\item a number $h > 0$ for each focus-focus singularity, the \strong{height invariant}, corresponding to the height of the image of the focus-focus critical value in any of these polygons, and measuring the volume of some reduced space,
\item a \strong{Taylor series} of the form $S^{\infty} = a_1 X + a_2 Y + \sum_{i+j > 1} b_{ij} X^i Y^j$ for each focus-focus singularity,
\item roughly speaking, an integer associated with each focus-focus singularity and polygon in item 2, called the \strong{twisting index}, reflecting the fact that there exists a privileged toric momentum map in a neighborhood of the singularity. When $m_f = 1$, one can always find a polygon in item 2 whose associated twisting index vanishes.
\end{enumerate}

We will describe these invariants in more details in Section \ref{sect:symp_invariant}; for a complete discussion, we refer the reader to \cite{PelVuSemi} and to the recent notes by Sepe and V{\~u} Ng{\d{o}}c \cite{SepeVN}. As we already said before, we will see that the system of coupled angular momenta is of toric type for certain values of the parameter $t$ and semitoric with exactly one focus-focus value for a range of values of $t$ always including $t=1/2$. Our main result is the computation of some of the symplectic invariants in this case $t=1/2$; this is to our knowledge the first time those are computed for a semitoric system on a compact manifold.

\begin{figure}[H]
\begin{center}
\subfigure[Polygon $\Delta_1$ with $\epsilon=1$.]{\includegraphics[scale=0.22]{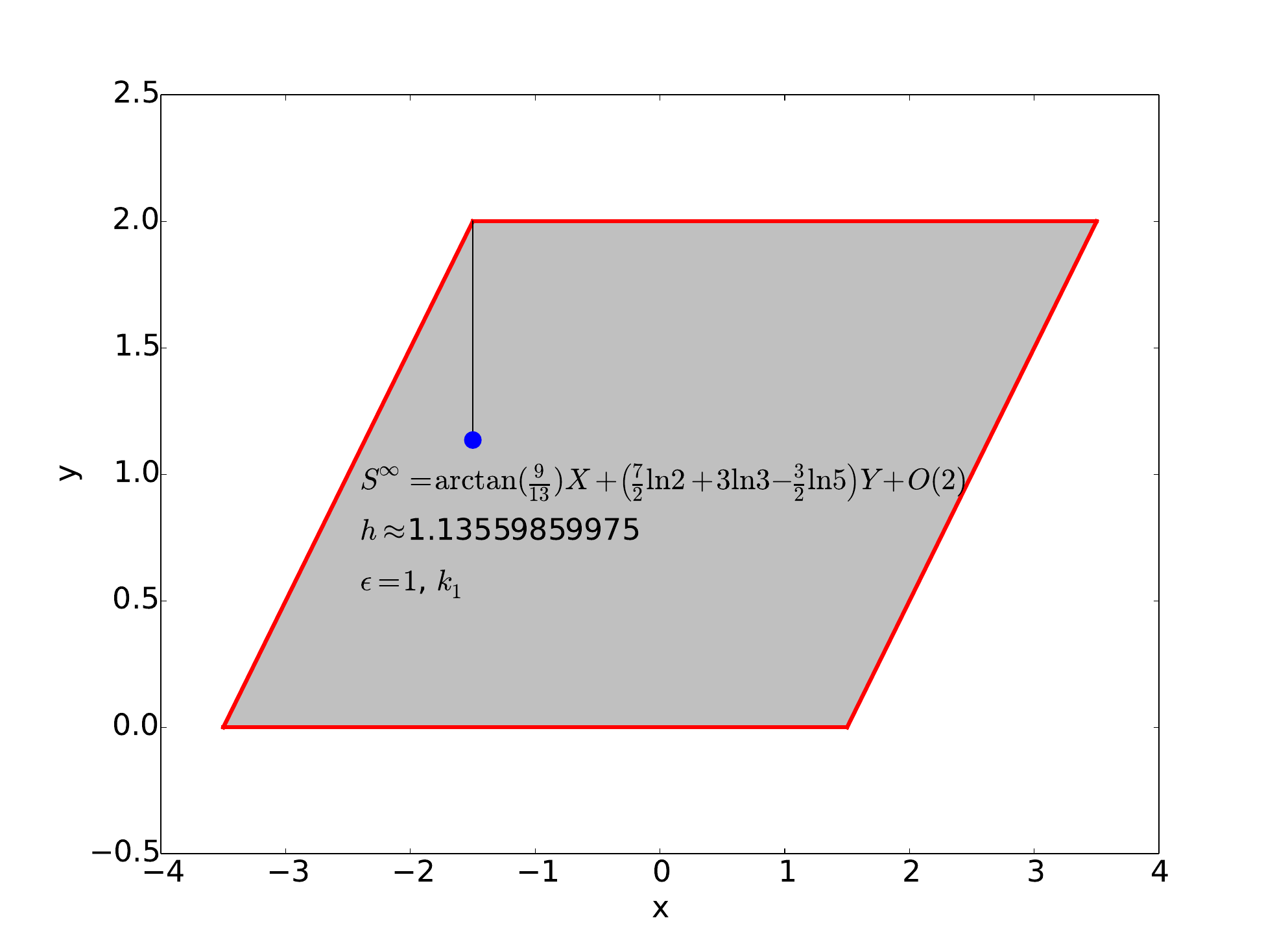} }
\subfigure[Polygon $\Delta_2$ with $\epsilon=-1$.]{\includegraphics[scale=0.22]{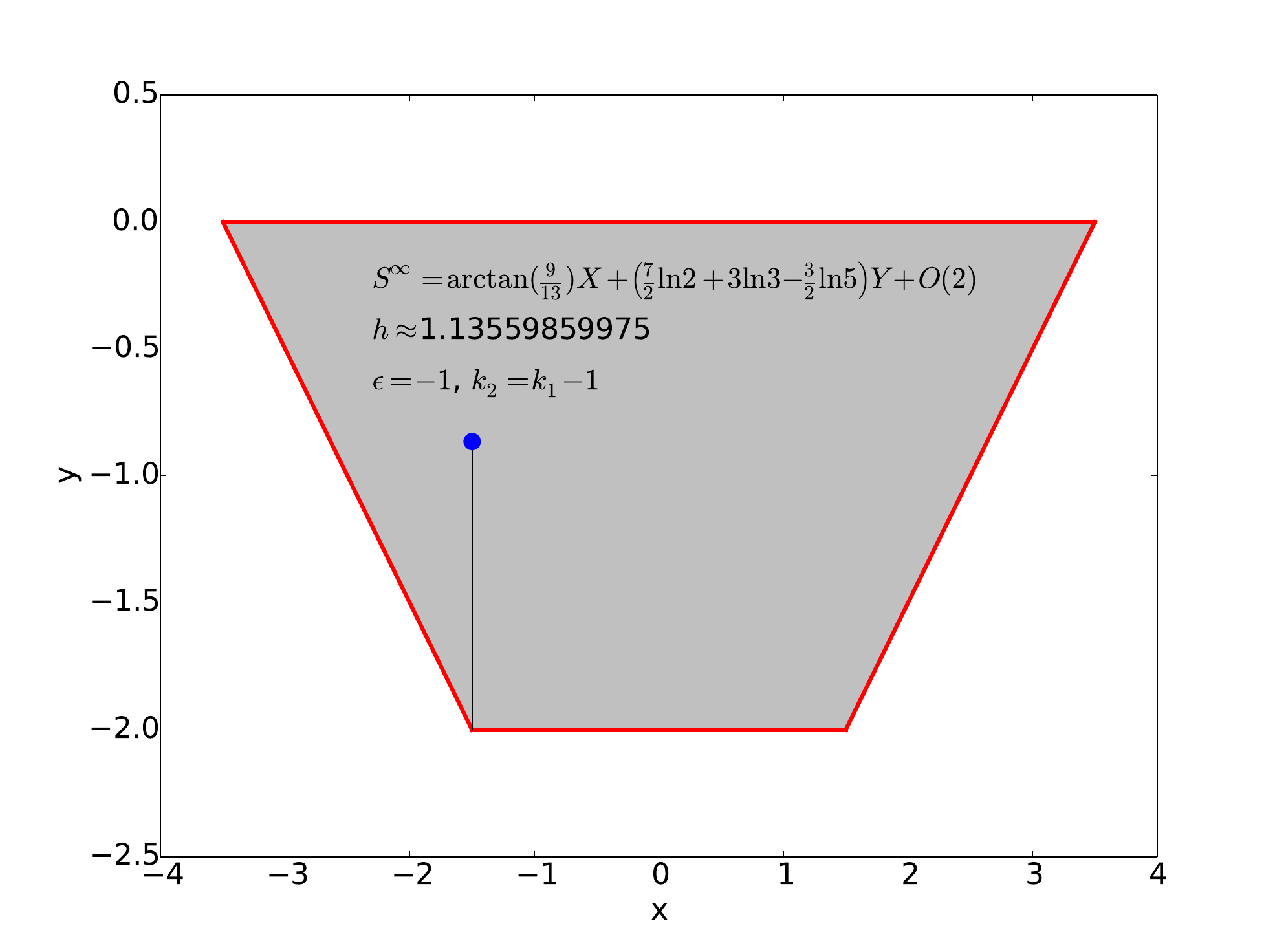} }
\end{center}
\caption{The symplectic invariants of~\cite{PelVuSemi,PelVuConstr} for the coupled angular momentum system $F$ in formula
(\ref{eq:moment_map}) for the values of $R_1=1,R_2=5/2$. The blue dot indicates the image of the focus-focus critical value in the polygon. The number $\epsilon \in\{ -1,1 \}$ corresponds to a choice of vertical half-line at the focus-focus value, as represented by the thin black segment.}
\label{fig:invariants}
\end{figure}

\begin{thm}
For $t = 1/2$, the coupled angular momentum system $(J,H)$ is a simple semitoric integrable system. Its symplectic invariants satisfy the following properties:
\begin{enumerate}
\item the number of focus-focus critical values is equal to one,
\item the polygonal invariant is represented by the two polygons $\Delta_1$ and $\Delta_2$ drawn in Figure \ref{fig:polygons},
\item if we set $\Theta = \frac{R_2}{R_1}$, the height invariant is equal to
\[ h = \frac{R_1}{\pi} \left(2 \arccos\left(\frac{1}{2\sqrt{\Theta}}\right) + \sqrt{4\Theta - 1} - 2(\Theta-1) \arctan\left( \frac{4 \Theta - 1 - \sqrt{4\Theta-1}}{(2 \Theta - 1)(\sqrt{4\Theta - 1} - 1)} \right) \right), \]
\item for $R_1 = 1, R_2 = 5/2$, the two first terms of the Taylor series invariant $S$ are $a_1 = \arctan(\frac{9}{13})$ and $a_2 = \frac{7}{2} \ln 2 + 3 \ln 3 - \frac{3}{2} \ln 5$, which means that
\[ S^{\infty}(X,Y) = \arctan\left(\frac{9}{13}\right) X + \left(\frac{7}{2} \ln 2 + 3 \ln 3 - \frac{3}{2} \ln 5 \right) Y + O(2),\]
 where $O(2)$ stands for terms of total degree strictly greater than one,
\item the twisting indices $k_1,k_2$ of $\Delta_1$, $\Delta_2$ satisfy $k_2 = k_1-1$.

\end{enumerate}
\end{thm}

This theorem follows from Propositions \ref{prop:focus}, \ref{prop:polygon}, \ref{prop:height}, \ref{prop:a1} and \ref{prop:a2} and Lemma \ref{lm:twist}.

\begin{rmk}
We do not compute $k_1, k_2$ in the theorem.
We do not compute the terms $a_1$ and $a_2$ in the Taylor series invariants for general values of $R_1$ and $R_2$ because this would lead to overly complicated computations, but we give detailed explanations so that the interested reader can compute them for other fixed values of $R_1, R_2$. These results are summed up in Figure \ref{fig:invariants} for the case $R_1=1,R_2=5/2$.
\end{rmk}

\subsubsection{Spectral properties of quantum coupled angular momenta}

Our second result is the construction of a quantum integrable system $(\hat{J}_k,\hat{H}_k)$ quantizing $(J,H)$, that is the data of two commuting self-adjoint operators acting on some Hilbert space, with underlying classical system $(J,H)$. Note that such a quantization may not exist for a general integrable system, see \cite{GvS} for the description of some obstructions. To be more precise, since the coupled angular momentum system is defined on a compact phase space, the relevant quantization involves the so-called geometric quantization \cite{Kos,Sou} and $\hat{J}_k,\hat{H}_k$ are Berezin-Toeplitz operators, see for instance \cite{BouGui,Cha1,Mama,Schli}.

These operators act on finite dimensional Hilbert spaces $\Hil_k$ that we describe explicitly in Section \ref{sect:quantum} and which depend on a semiclassical parameter $k$, a positive integer which tends to infinity. This integer plays the part of the inverse of the Planck constant $\hbar$. The rigourous way to express that $(\hat{J}_k,\hat{H}_k)$ quantizes $(J,H)$ is to say that the principal symbols of $\hat{J}_k,\hat{H}_k$ are equal to $J$ and $H$ respectively.

We sum up some of the results of Section \ref{sect:quantum} in the following theorem.

\begin{thm}
The Hilbert space $\Hil_k$ has dimension $4k^2 R_1 R_2$. There exists a basis $(g_{\ell,m})_{\substack{0 \leq \ell \leq 2kR_1-1 \\ 0 \leq m \leq 2kR_2-1}}$ of $\Hil_k$ such that
\[ \hat{J}_k g_{\ell,m} = \left( R_1 + R_2 - \frac{\ell + m + 1}{k} \right) g_{\ell,m}. \]
Furthermore, using the convention $g_{-1,m} = g_{2kR_1,m} = g_{\ell,-1} = g_{\ell,2kR_2} =  0$, we have that
\[ \begin{split} \hat{H}_k g_{\ell,m} & =  \frac{1}{4k^2R_1R_2} \left( 2t \sqrt{\ell(2kR_1-\ell)(m+1)(2kR_2-1-m)} \ g_{\ell-1,m+1} \right. \\
& \\
& + \left(2(kR_1-\ell)-1\right)(2kR_2 - (2m+1)t) g_{\ell,m} \left. + 2t \sqrt{(\ell+1)(2kR_1-1-\ell)m(2kR_2-m)} \ g_{\ell+1,m-1} \right). \end{split} \]
\end{thm}

We compute numerically the joint spectrum of $(\hat{J}_k,\hat{H}_k)$ using these formulas. Furthermore, we state a conjecture regarding the description of this joint spectrum near the focus-focus critical value and explain how to check (assuming this conjecture is true) on the joint spectrum that our computations of the symplectic invariants of the classical system are correct.

\subsection{Structure of the article}

Our study is divided into three parts, and the structure of this article respects them: the general study of the system, the computation of the semitoric symplectic invariants, and the quantization problem. We start, in the next section, by investigating the classical properties of the system (e.g. singularities, image of the momentum map). Ultimately, we obtain that the system is semitoric but not of toric type for some values of the deformation parameter $t$, always including $t=1/2$ no matter which values $R_1$ and $R_2$ take; we study this particular case in the third section. As explained earlier, the symplectic classification of semitoric systems is achieved by five symplectic invariants, and  we compute four of them for this particular example. In the last section, we explain how to quantize the coupled angular momentum system; more precisely, we construct a pair of commuting self-adjoint Berezin-Toeplitz operators whose principal symbols are equal to $J$ and $H$ respectively. We also describe the computation of the joint spectrum of these operators, and display numerical simulations of the latter. We conclude by stating a conjecture about the description of the joint spectrum of commuting self-adjoint Berezin-Toeplitz operators near a focus-focus value of the underlying integrable system, and by providing numerical evidence in favor of this conjecture.

\paragraph{Acknowledgements} YLF was supported by the European Research Council Advanced Grant 338809. Part of this work was done during a visit of YLF to AP at University of California San Diego in April 2016, and he would like to thank the members of this university for their hospitality. AP is supported by NSF CAREER grant DMS-1518420. He also received support from Severo Ochoa Program at ICMAT in Spain. Part of this work was carried out at ICMAT and Universidad Complutense de Madrid. We thank Joseph Palmer for useful comments, and Jaume Alonso and Holger Dullin for pointing out a mistake in the computation of the coefficient $a_1$ in the Taylor series invariant (Proposition \ref{prop:a1}) in an earlier version of this manuscript.

\section{Classical properties of $F$}
\label{sect:classical}

Let us study the classical system $(J,H)$ of coupled angular momenta introduced in Equation (\ref{eq:moment_map}). The first observation that we make is that these two functions indeed Poisson commute. This is a simple exercise, but it is a good way to introduce the convention that we will be using throughout the paper. If $f: M \to \R$ is a smooth function, its \strong{Hamiltonian vector field} $X_f$ is defined by the formula $df + \omega(X_f,\cdot) = 0$. The Poisson bracket of two smooth functions $f,g: M \to \R$ is given by $\{f,g\} = \mathcal{L}_{X_f} g = \omega(X_f,X_g)$, where $\mathcal{L}$ stands for the Lie derivative.

\begin{lm}
We have that $\{J,H\} = 0$.
\end{lm}

\begin{proof}
Since the coordinates on the first factor Poisson commute with the ones on the second factor, and since $\{z_1,z_1\} = 0$, we obtain that
\[ \{J,H\} = t \left( R_1 \{z_1,x_1 x_2\} + R_1 \{z_1,y_1 y_2\} + R_2 \{z_2,x_1 x_2\} + R_2 \{z_2,y_1 y_2\} \right), \]
and by the Leibniz rule (still using the remark above):
\[ \{J,H\} = t \left( R_1 x_2 \{z_1,x_1\} + R_1 y_2 \{z_1,y_1\} + R_2 x_1 \{z_2,x_2\} + R_2 y_1 \{z_2, y_2\} \right). \]
With the convention above, one readily checks that the relation $\{x_i,y_i\} = -\frac{1}{R_i} z_i$ holds, as well as its cyclic permutations, for $i=1,2$. Therefore, the previous equality yields $\{J,H\} = 0$.
\end{proof}

In order to check that $F = (J,H)$ forms an integrable system, we still need to understand its singularities, which are the points where $X_J, X_H$ are not linearly independent, or equivalently where $dF$ has non zero corank. Since we are in dimension four, there are two cases: critical points of corank one (i.e. where $dF$ has corank one), and critical points of corank two, for which $dF=0$, which are usually called \strong{fixed points}.

\begin{dfn}
A fixed point $m \in M$ is said to be \strong{non-degenerate} if the Hessians $d^2 J(m)$ and $d^2 H(m)$ generate a Cartan subalgebra of the Lie algebra of quadratic forms on $T_m M$, equipped with the linearization of the Poisson bracket, that is a maximal Abelian subalgebra $\mathfrak{c}$ such that for every element $q \in \mathfrak{c}$, $\mathrm{ad}_q$ is a semisimple endomorphism.
\end{dfn}

The problem with this definition is that one may ask how to check this condition in practice; this is the purpose of the following lemma, which can be found in \cite{BolFom} for instance, and uses the identification of the Lie algebra of quadratic forms and the symplectic Lie algebra $\mathfrak{sp}(4,\R)$.

\begin{lm}
\label{lm:crit_nondeg}
A critical point $m \in M$ of corank two of $F$ is non-degenerate if and only if there exists a linear combination $A$ of $\Omega^{-1} d^2 J(m)$ and $\Omega^{-1} d^2 H(m)$ with four distinct eigenvalues (such a $A$ will be called \emph{regular}). Here we slightly abuse notation by fixing a basis $\mathcal{B}$ of $T_m M$ and by identifying the Hessians of $J$ and $H$ with their matrices in $\mathcal{B}$, and $\Omega$ is the matrix of the symplectic form $\omega$ in $\mathcal{B}$.
\end{lm}

If such a linear combination $A$ exists, the Cartan subalgebra $\mathfrak{c}$ of $\mathfrak{sp}(4,\R)$ generated by $\Omega^{-1} d^2 J(m)$ and $\Omega^{-1} d^2 H(m)$ is equal to the commutant of $A$. The eigenvalues of $A$ come in pairs $(\lambda,-\lambda)$, and the spaces $E_{\pm \lambda} = \ker (A-\lambda \mathrm{Id}) \oplus \ker (A+\lambda \mathrm{Id})$ associated with distinct pairs are symplectically orthogonal. There are three distinct possibilities:
\begin{enumerate}
\item if $\lambda \in \R$, there exists a symplectic basis of $E_{\pm \lambda}$ in which the restriction of $A$ to $E_{\pm \lambda}$ has matrix $\begin{pmatrix} \lambda & 0 \\ 0 & -\lambda \end{pmatrix}$,
\item if $\lambda = i \alpha$, $\alpha \in \R$, there exists a symplectic basis of $E_{\pm \lambda}$ in which the restriction of $A$ to $E_{\pm \lambda}$ has matrix $\begin{pmatrix} 0 & \alpha \\ -\alpha & 0 \end{pmatrix}$,
\item if $\lambda = \alpha + i \beta$, $\alpha,\beta \neq 0$, then $\pm \bar{\lambda}$ are also eigenvalues of $A$, and there exists a symplectic basis of $T_m M$ in which $A$ has matrix $\begin{pmatrix} -\alpha & \beta & 0 & 0 \\ -\beta & -\alpha & 0 & 0 \\ 0 & 0 & \alpha & \beta \\ 0 & 0 & -\beta & \alpha \end{pmatrix}$.
\end{enumerate}
This leads to the following classification, due to Williamson \cite{Wil}: there exist linear symplectic coordinates $(u_1,u_2,\xi_1,\xi_2)$ of $T_m M$ and a basis $q_1, q_2$ of $\mathfrak{c}$ such that each $q_i$ is one of the following:
\begin{enumerate}
\item $q_i = x_i \xi_i$ (\strong{hyperbolic component}),
\item $q_i = \frac{1}{2} (x_i^2 + \xi_i^2)$ (\strong{elliptic component}),
\item $q_1 = u_1 \xi_2 - u_2 \xi_1$ and $q_2 = u_1 \xi_1 + u_2 \xi_2$ (\strong{focus-focus component}).
\end{enumerate}

This notion of non-degeneracy can be extended to the case of critical points of corank one, see for instance \cite[Definition $1.21$]{BolFom}.

\subsection{Critical points of maximal corank}

We start by looking for the critical points of maximal corank (fixed points).

\begin{lm}
\label{lm:crit_points}
The map $F$ has four critical points $m_i$, $i=0,\ldots,3$, of maximal corank:
\[ \begin{cases} m_0 = (0,0,1,0,0,-1), & c_0 = F(m_0) = (R_1 - R_2,1-2t), \\ m_1 = (0,0,-1,0,0,-1), & c_1 = F(m_1) = (-(R_1 + R_2),2t - 1), \\ m_2 = (0,0,-1,0,0,1), & c_2 = F(m_2) = (R_2 - R_1,- 1), \\ m_3 = (0,0,1,0,0,1), & c_3 = F(m_3) = (R_1 + R_2,1). \end{cases}\]
\end{lm}

\begin{proof}
At a critical point $m = (x_1,y_1,z_1,x_2,y_2,z_2)$ of maximal corank for $F$, the restriction of $dJ = R_1 dz_1 + R_2 dz_2$ to the tangent space $T_m(\S^2 \times \S^2)$ must vanish. Therefore, we must have $z_1, z_2 = \pm 1$. But then the restriction of $dH = (1-t)dz_1 + t (z_1 dz_2 + z_2 dz_1)$ to $T_m(\S^2 \times \S^2)$ also vanishes.
\end{proof}

We want to understand if these critical points are non-degenerate, and if it is the case, what their Williamson types are. Let us define $t^-, t^+ \in (0,1)$ as
\begin{equation} t^{\pm} = \frac{R_2}{2 R_2 + R_1 \mp 2\sqrt{R_1 R_2}}; \label{eq:tpm} \end{equation}
observe that we always have $t^- < \frac{1}{2} < t^+$.

\begin{prop}
\label{prop:focus}
The critical point $m_0$ is non-degenerate of focus-focus type when $t^- < t < t^+$, degenerate when $t \in \{ t^-, t^+ \}$, and non-degenerate of elliptic-elliptic type otherwise.
\end{prop}

These results are summed up in Figure \ref{fig:type}. The proof will rely on the criterion stated in Lemma \ref{lm:crit_nondeg}, and most of the time the relevant linear combination will simply be $\Omega^{-1} d^2 H(m_0)$.

\begin{figure}[H]
\begin{center}
\includegraphics[scale=0.7]{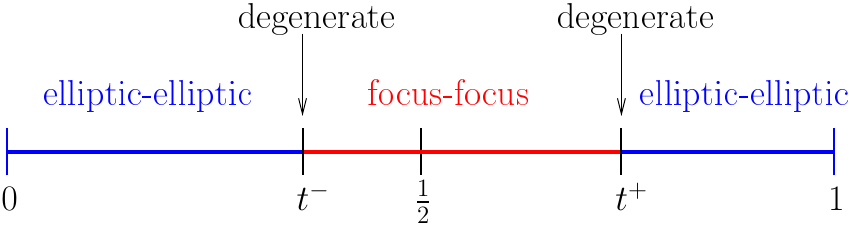}
\end{center}
\caption{The singularity type of $m_0$ with respect to the parameter $t$.}
\label{fig:type}
\end{figure}

\begin{proof}
We start by computing the Hessian matrices of $J$ and $H$ at $m_0$. Near this point, $(x_1,y_1,x_2,y_2)$ form a local coordinate system and $z_1 = \sqrt{1-(x_1^2 + y_1^2)}, \ z_2 = - \sqrt{1 - (x_2^2 + y_2^2)}$. The point $m_0$ corresponds to $x_1=y_1=x_2=y_2=0$. Therefore, by using the Taylor series expansion of $\sqrt{1+u}$ when $u \to 0$, we get
\[ J = R_1 - R_2 - R_1 \left( \frac{x_1^2 + y_1^2}{2} \right) + R_2 \left( \frac{x_2^2 + y_2^2}{2} \right) + O(3) \]
near $m_0$, where $O(3)$ stands for $O(\|(x_1,y_1,x_2,y_2)\|^3)$ to simplify notation. Similarly,
\[ H = 1 - 2t + (2t-1) \left( \frac{x_1^2 + y_1^2}{2} \right) + t \left( \frac{x_2^2 + y_2^2}{2} \right) + t(x_1x_2 + y_1 y_2) + O(3) \]
near $m_0$. Consequently, the Hessians of $J$ and $H$ at $m_0$ in the basis $\mathcal{B}$ of $T_{m_0} M$ associated with $(x_1,y_1,x_2,y_2)$ read
\[ d^2 J(m_0) = \begin{pmatrix} -R_1 & 0 & 0 & 0 \\ 0 & -R_1 & 0 & 0 \\ 0 & 0 & R_2 & 0 \\ 0 & 0 & 0 & R_2  \end{pmatrix}, \quad d^2 H(m_0) = \begin{pmatrix} 2t-1 & 0 & t & 0 \\ 0 & 2t - 1  & 0 & t \\  t & 0 & t & 0 \\ 0 & t & 0 & t  \end{pmatrix}. \]
Moreover, $\omega_{m_0} = -R_1 dx_1 \wedge dy_1 + R_2 dx_2 \wedge dy_2$ on $T_{m_0} M = T_{(0,0,1)} \S^2 \times T_{(0,0,-1)} \S^2$, so the matrix of $\omega_{m_0}$ in $\mathcal{B}$ is given by
\[ \Omega = \begin{pmatrix} 0 & -R_1 & 0 & 0 \\ R_1 & 0 & 0 & 0 \\ 0 & 0 & 0 & R_2 \\ 0 & 0 & -R_2 & 0  \end{pmatrix}. \]
Consequently, we obtain that
\begin{equation} \Omega^{-1} d^2 J(m_0) = \begin{pmatrix} 0 & -1 & 0 & 0 \\ 1 & 0 & 0 & 0 \\ 0 & 0 & 0 & -1 \\ 0 & 0 & 1 & 0  \end{pmatrix}, \quad  \Omega^{-1} d^2 H(m_0) = \begin{pmatrix} 0 & \frac{2t-1}{R_1} & 0 & \frac{t}{R_1} \\ \frac{1-2t}{R_1} & 0 & -\frac{t}{R_1} & 0 \\ 0 & -\frac{t}{R_2} & 0 & -\frac{t}{R_2} \\ \frac{t}{R_2} & 0 & \frac{t}{R_2} & 0  \end{pmatrix}. \label{eq:hessians}\end{equation}

The characteristic polynomial of $A = \Omega^{-1} d^2 H(m_0)$ is
\[ P = X^4 + \frac{ (R_1^2 + 4 R_2^2 - 2R_1 R_2) t^2 - 4R_2^2 t + R_2^2}{R_1^2 R_2^2} X^2 + \frac{t^2 (t-1)^2}{R_1^2 R_2^2} = Q(X^2), \]
where $Q$ is a degree two polynomial. A straightforward computation shows that the discriminant $\Delta$ of $Q$ satisfies
\[ R_1^4 R_2^4 \Delta = \left( (R_1^2 + 4 R_2^2) t^2 - 2R_2 (R_1 + 2 R_2) t + R_2^2 \right) \left( (2 R_2 - R_1) t - R_2 \right)^2.  \]
The polynomial $\left( (2 R_2 - R_1) t - R_2 \right)^2$ is positive, except at $t = R_2/(2R_2 - R_1)$ where it vanishes. So the sign of $\Delta$ is given by the sign of $ T = (R_1^2 + 4 R_2^2) t^2 - 2R_2 (R_1 + 2 R_2) t + R_2^2$. The discriminant $\widetilde{\Delta}$ of $T$ satisfies $ \widetilde{\Delta} = 4R_2^2 (R_1 + 2R_2)^2 - 4R_2^2 (R_1^2 + 4 R_2^2) = 16 R_1 R_2^3 > 0$, therefore $T$ has two real roots
\[ t^{\pm} = \frac{R_2(R_1 + 2R_2) - 2 R_2 \sqrt{R_1R_2}}{R_1^2 + 4 R_2^2}. \]
We claim that these are the same as in Equation (\ref{eq:tpm}); in order to see this, we multiply both the numerator and denominator by the quantity $R_2(R_1 + 2R_2) + 2 R_2 \sqrt{R_1R_2}$, and we use the fact that
\[ \left(R_2(R_1 + 2R_2) - 2 R_2 \sqrt{R_1R_2}\right)\left(R_2(R_1 + 2R_2) + 2 R_2 \sqrt{R_1R_2}\right) = R_2^2 (R_1^2 + 4 R_2^2). \]

\paragraph{When $t^- < t < t^+$ and $t \neq R_2/(2R_2 - R_1)$,} we get that $\Delta < 0$, which means that for such $t$, $Q$ has two complex conjugate roots, and $A$ has four eigenvalues of the form $\pm \alpha \pm i \beta$, $\alpha, \beta \in \R$. Hence, for such $t$, $m_0$ is non-degenerate of focus-focus type.

\paragraph{When $0 < t < t^-$ or $t^+ < t < 1$,} $\Delta > 0$, so $Q$ has two real roots $\lambda^{\pm} = (-b \pm \sqrt{\Delta})/2$ where
\[ b = \frac{ (R_1^2 + 4 R_2^2 - 2R_1 R_2) t^2 - 4R_2^2 t + R_2^2}{R_1^2 R_2^2} = \frac{T(t) + 2 R_1 R_2 t(1-t)}{R_1^2 R_2^2} > 0. \]
Let $C = R_1^4 R_2^4 (\Delta - b^2) = T \left( (2 R_2 - R_1) t - R_2 \right)^2 - \left( T + 2 R_1 R_2 t(1-t) \right)^2$; then
\[ C = T \left( (2 R_2 - R_1)^2 t^2 - 2R_2(2R_2 - R_1)t + R_2^2 - 4 R_1 R_2 t(1-t) - T \right) - 4 R_1^2 R_2^2 t^2 (1-t)^2. \]
By expanding, one can check that $(2 R_2 - R_1)^2 t^2 - 2R_2(2R_2 - R_1)t + R_2^2 - 4 R_1 R_2 t(1-t) = T$, thus we finally obtain that $C = - 4 R_1^2 R_2^2 t^2 (1-t)^2 \leq 0$. Since $0 < t < t^-$ or $t^+ < t < 1$, $C < 0$, hence $\lambda^{\pm} < 0$, and $A$ has four distinct eigenvalues of the form $\pm i \alpha$ and $\pm i \beta$, $\alpha \neq \beta \in \R$. Therefore $m_0$ is non-degenerate of elliptic-elliptic type.

When $t$ is equal to $R_2/(2R_2 - R_1), 0, 1, t^-$ or $t^+$, $A$ has multiple eigenvalues (i.e. is not regular), so its study is not sufficient to conclude anything. We need to investigate these cases separately.

\paragraph{When $t =  R_2/(2R_2 - R_1)$.} Observe that $t^- < 1/2 < R_2/(2R_2 - R_1) < t^+$, so we expect $m_0$ to be of focus-focus type for this value of $t$. One way to check this is to consider the linear combination
\[ B =  \Omega^{-1} d^2 J(m_0) + (2R_2 - R_1) \Omega^{-1} d^2 H(m_0) =  \begin{pmatrix} 0 & 0 & 0 & \frac{R_2}{R_1} \\ 0 & 0 & -\frac{R_2}{R_1} & 0 \\ 0 & -1 & 0 &  -2 \\  1 & 0 &  2 & 0  \end{pmatrix}. \]
Its characteristic polynomial is equal to $P(X^2)$ where $ P = X^2 + \frac{2(2R_1 - R_2)}{R_1} X + \frac{R_2^2}{R_1^2}$, and the discriminant of $P$ is equal to $16(R_1 - R_2)/R_1 < 0$. Therefore, $B$ has four distinct eigenvalues of the form $\pm \alpha \pm i \beta$, $\alpha, \beta \in \R$, hence $m_0$ is non-degenerate of focus-focus type.

\paragraph{When $t =  0$.} The linear combination
\[ \Omega^{-1} d^2 J(m_0) + \Omega^{-1} d^2 H(m_0) =  \begin{pmatrix} 0 & -\left( 1 + \frac{1}{R_1} \right)  & 0 & 0 \\  1 + \frac{1}{R_1}   & 0 & 0 & 0 \\ 0 & 0 & 0 &  -1 \\  0 & 0 & 1 & 0  \end{pmatrix} \]
has four distinct eigenvalues $\pm i \left( 1 + \frac{1}{R_1} \right), \pm i$, so $m_0$ is non-degenerate of elliptic-elliptic type.

\paragraph{When $t =  1$.} One can check that the linear combination
\[ \Omega^{-1} d^2 J(m_0) + R_1 \Omega^{-1} d^2 H(m_0) =  \begin{pmatrix} 0 & 0  & 0 & 1 \\  0 & 0 & -1 & 0 \\ 0 & -\frac{R_1}{R_2} & 0 &  -\left( 1 + \frac{R_1}{R_2} \right) \\  \frac{R_1}{R_2} & 0 & 1 + \frac{R_1}{R_2}  & 0  \end{pmatrix} \]
has eigenvalues $\pm i, \pm i\frac{R_1}{R_2}$, so $m_0$ is non-degenerate of elliptic-elliptic type.

\paragraph{When $t =  t^-$.} For $\lambda \in \R$, let $B(\lambda) = \left( 2R_2 + R_1 + 2\sqrt{R_1 R_2} \right) \lambda \Omega^{-1} d^2 J(m_0) + \Omega^{-1} d^2 H(m_0)$; then
\[ B(\lambda) =   \begin{pmatrix} 0 & - \left(a \lambda + b \right)    & 0 & \frac{R_2}{R_1} \\  a \lambda + b  & 0 & -\frac{R_2}{R_1} & 0 \\ 0 & -1 & 0 &  -\left(  a \lambda + 1 \right) \\  1 & 0 &  a \lambda +1 & 0  \end{pmatrix} \]
with $a = 2R_2 + R_1 + 2\sqrt{R_1 R_2} $ and $b =  1 + 2 \sqrt{\frac{R_2}{R_1}}$. One can check that the characteristic polynomial of $B(\lambda)$ is $P(X^2)$ with
\[ P = X^2 + \left( \frac{2 a \lambda R_1 (a \lambda + b + 1) + b^2 R_1 + R_1 - 2 R_2 }{R_1} \right) X + \frac{\left(a\lambda R_1(a \lambda R_1 + b + 1) + b R_1 + R_2 \right)^2}{R_1^2}. \]
The discriminant $\Delta$ of $P$ satisfies $R_1 \Delta = (2 a \lambda + b + 1)^2 ( R_1 + b^2 R_1 - 4R_2 - 2bR_1 )$, and
\[ R_1 + b^2 R_1 - 4R_2 - 2bR_1 = R_1 + R_1 \left( 1 + 4 \sqrt{\frac{R_2}{R_1}} + 4 \frac{R_2}{R_1} \right) - 4 R_2 - 2 R_1 - 4 \sqrt{R_1 R_2} = 0, \]
so $\Delta = 0$. Consequently, $B(\lambda)$ has double eigenvalues for every $\lambda$, hence $m_0$ is degenerate.

\paragraph{When $t =  t^+$.} 
For $\lambda \in \R$, let $B(\lambda) = \left( 2R_2 + R_1 - 2\sqrt{R_1 R_2} \right) \lambda \Omega^{-1} d^2 J(m_0) + \Omega^{-1} d^2 H(m_0)$; then
\[ B(\lambda) =   \begin{pmatrix} 0 & - \left(a \lambda + b \right)    & 0 & \frac{R_2}{R_1} \\  a \lambda + b  & 0 & -\frac{R_2}{R_1} & 0 \\ 0 & -1 & 0 &  -\left(  a \lambda + 1 \right) \\  1 & 0 &  a \lambda + 1 & 0  \end{pmatrix} \]
with $a = 2R_2 + R_1 - 2\sqrt{R_1 R_2} $ and $b =  1 - 2 \sqrt{\frac{R_2}{R_1}}$. So as above, the characteristic polynomial is of the form $P(X^2)$ with $P$ a polynomial of degree two with discriminant
\[ \Delta = \frac{(2 a \lambda + b + 1)^2 ( R_1 + b^2 R_1 - 4R_2 - 2bR_1 )}{R_1}. \]
Again, one can check that $R_1 + b^2 R_1 - 4R_2 - 2bR_1 = 0$, so $\Delta = 0$ and $m_0$ is degenerate.
\end{proof}

It turns out that the Williamson types of the three other fixed points of $F$ do not depend on the value of the parameter $t \in [0,1]$.

\begin{prop}
For every $t \in [0,1]$, the critical points $m_1, m_2$ and $m_3$ are non-degenerate of elliptic-elliptic type.
\end{prop}

The proof follows the same lines as the study of $m_0$, hence we leave it to the reader.

\subsection{Critical points of corank one}

We now look for the points where $dJ, dH$ are linearly dependent but $dF \neq 0$.

\begin{prop}
\label{prop:crit_corank_one}
When $t \notin \{0,1\}$, the critical points of corank one of $F$ are the points $(x_1,y_1,z_1,x_2,y_2,z_2) \in \S^2 \times \S^2$ for which there exists $\lambda \in \R \setminus \{ 0, \frac{1-t}{R_1} \}$ such that $(x_2,y_2) = \frac{1-t-R_1 \lambda}{R_2 \lambda} (x_1,y_1)$ and
\[ z_1 = \frac{(t^2 + R_2^2 \lambda^2)(1-t-R_1 \lambda)^2 - t^2 R_2^2 \lambda^2}{2 t R_2 \lambda (1-t-R_1 \lambda)^2}, \quad z_2 = \frac{(t^2 - R_2^2 \lambda^2)(1-t-R_1 \lambda)^2 - t^2 R_2^2 \lambda^2}{2 t R_2^2 \lambda^2 (1-t-R_1\lambda)} \]
and which are different from the points $m_0, m_1, m_2, m_3$ introduced earlier. When $t = 1$, the critical points of corank one are either the points $(x_1,y_1,z_1,x_2,y_2,z_2) \in \S^2 \times \S^2 \setminus \{m_0, m_1, m_2, m_3\}$ such that $(x_2,y_2,z_2) = \pm (x_1,y_1,z_1)$ or those for which there exists $\lambda \neq 0$ such that
\[ (x_2,y_2) = -\frac{R_1}{R_2} (x_1,y_1), \quad z_1 = \frac{R_1^2 - R_2^2 + R_1^2 R_2^2 \lambda^2}{2 R_1^2 R_2 \lambda}, \quad z_2 = \frac{R_2^2-R_1^2 + R_1^2 R_2^2 \lambda^2}{2 R_1 R_2^2 \lambda}. \]
When $t = 0$, the critical points of corank one of $F$ are the points of the form $(0,0,\pm 1,x_2,y_2,z_2)$ with $(x_2,y_2,z_2) \neq (0,0,\pm 1)$ or of the form $(x_1,y_1,z_1,0,0,\pm 1)$ with $(x_1,y_1,z_1) \neq (0,0,\pm 1)$.
\end{prop}

\begin{proof}
If $m$ is a critical point of corank one, there exist $\lambda, \mu_1, \mu_2 \in \R$ such that
\[ \nabla H (m) = \lambda \nabla J (m) + \mu_1 \nabla f_1 (m) + \mu_2 \nabla f_2 (m) \]
with $f_i(x_1,y_1,z_1,x_2,y_2,z_2) = x_i^2 + y_i^2 + z_i^2$, $i=1,2$. We obtain the following equations:
\[ \begin{cases} t x_2 = 2 \mu_1 x_1, \\ t y_2 = 2 \mu_1 y_1, \\ 1-t + t z_2 = \lambda R_1 + 2 \mu_1 z_1, \\ t x_1 = 2 \mu_2 x_2, \\ t y_1 = 2 \mu_2 y_2, \\ t z_1 = \lambda R_2 + 2 \mu_2 z_2. \end{cases} \]
We start with the case $t \notin \{ 0,1 \}$. If $\mu_1 = 0$, we immediately get that $x_2 = y_2 =0$, which implies that $x_1 = y_1 =0$. Therefore $z_1, z_2 \in \{-1,1\}$, and we find the critical points of corank two, see the previous section. So we may assume that $\mu_1 \neq 0$, and we may also assume that $x_1 \neq 0$ or $y_1 \neq 0$, because otherwise we find the critical points of corank two again. Then, comparing either the first and fourth equations or the second and fifth equation, we get $ \mu_2 = \frac{t^2}{4 \mu_1}$. Thus, combining the third and last equations, we obtain that
\[ 1 - t + t z_2 = \lambda R_1 + \frac{2 \lambda R_2 \mu_1 }{t} + t z_2 \]
which yields $\lambda \neq 0$ (since $t \neq 1$) and $\mu_1 = \frac{t(1-t-R_1 \lambda)}{2 R_2 \lambda}$. In particular, $1-t-R_1 \lambda \neq 0$. Now, a straightforward computation gives
\[ z_2 = \frac{(1-t-R_1 \lambda)(t z_1 - R_2 \lambda)}{t R_2 \lambda}. \]
Thanks to the above results, the equality $1 = x_2^2 + y_2^2 + z_2^2$ reads $ 1 = \frac{(1-t-R_1 \lambda)^2}{t^2 R_2^2 \lambda^2} \left( t^2(x_1^2 + y_1^2) + (tz_1 - R_2 \lambda)^2 \right)$. Using the fact that $x_1^2 + y_1^2 + z_1^2 = 1$, this allows us to derive the equality
\[ z_1 = \frac{(t^2 + R_2^2 \lambda^2)(1-t-R_1 \lambda)^2 - t^2 R_2^2 \lambda^2}{2 t R_2 \lambda (1-t-R_1 \lambda)^2}. \]

When $t=1$, we still find the same points when $\lambda \neq 0$, but the difference is that we can now have $\lambda = 0$. In this case we obtain $(x_2,y_2,z_2) = 2 \mu_1 (x_1,y_1,z_1)$, hence $(x_2,y_2,z_2) = \pm (x_1,y_1,z_1)$.

Now, in the case where $t=0$, we obtain that
\[ \begin{cases} 0 = 2 \mu_1 x_1, \\ 0 = 2 \mu_1 y_1, \\ 1 = \lambda R_1 + 2 \mu_1 z_1, \\ 0 = 2 \mu_2 x_2, \\ 0 = 2 \mu_2 y_2, \\ 0 = \lambda R_2 + 2 \mu_2 z_2. \end{cases} \]
The last three equations imply that $\lambda^2 R_2^2 = 4 \mu_2^2$. Hence, if $\mu_2 = 0$, we also have that $\lambda = 0$ and the first three equations imply that $\mu_1 \neq 0$, so $x_1 = 0 = y_1$, and $z_1 = \pm 1$. Now, if $\mu_2 \neq 0$, we immediately get that $x_2 = 0 = y_2$ and $z_2 = \pm 1$. But we already know that when $z_1 = \pm 1$ and $z_2 = \pm 1$, we get the critical points of maximal corank.
\end{proof}

\subsection{Image of the momentum map}

We can now use the previous results to describe the image of $F$; more precisely, we will obtain a complete parameterization of the boundary of this image. For $0 < t \leq 1$, we define two functions $f,g : \R \setminus \{0, \frac{1-t}{R_1} \} \to \R$ by the formulas
\[ f(\lambda) = \frac{(t^2 + R_2^2 \lambda^2)(1-t-R_1 \lambda)^2 - t^2 R_2^2 \lambda^2}{2 t R_2 \lambda (1-t-R_1 \lambda)^2} , \quad g(\lambda) = \frac{(t^2 - R_2^2 \lambda^2)(1-t-R_1 \lambda)^2 - t^2 R_2^2 \lambda^2}{2 t R_2^2 \lambda^2 (1-t-R_1\lambda)}. \]
We saw that the critical points of corank one of $F$ are those for which there exists $\lambda$ such that $z_1 = f(\lambda)$ and $z_2 = g(\lambda)$. Therefore, we want to know for which values of $\lambda$ the numbers $f(\lambda)$ and $g(\lambda)$ both belong to $[-1,1]$; in other words, we want to describe $f^{-1}([-1,1]) \cap g^{-1}([-1,1])$. This is the purpose of the following technical lemma.

\begin{lm}
For $0 < t < 1$, we define the numbers
\[ \lambda_1^{\pm} = \frac{(1-2t)R_2 - t R_1 \pm \sqrt{((1-2t)R_2 - t R_1)^2 + 4 R_1 R_2 t (1-t)}}{2 R_1 R_2},\]
and
\[\lambda_2^{\pm} = \frac{ R_2 - t R_1 \pm \sqrt{(R_2 - t R_1)^2 + 4 R_1 R_2 t(1-t)} }{2 R_1 R_2}, \quad \lambda_3^{\pm} = \frac{ t R_1 + R_2 \pm \sqrt{(R_2 - t R_1)^2 + 4t^2 R_1 R_2} }{2 R_1 R_2}. \]
Then, for such $t$, one has $\lambda_1^- < \lambda_2^- < 0 < \lambda_3^- < \lambda_1^+ < \lambda_2^+ < \lambda_3^+ $. Moreover, for $t^- \leq t \leq t^+$, $ f^{-1}([-1,1]) = [\lambda_1^-,\lambda_2^-] \cup [\lambda_3^-,\lambda_1^+] \cup [\lambda_2^+,\lambda_3^+]$. Now, for  $0 < t < t^-$ or $t^+ < t < 1$, let
\[ \lambda_0^{\pm} = \frac{ (1-2t)R_2  + t R_1 \pm \sqrt{(R_1^2 + 4 R_2^2) t^2 - 2 R_2 (R_1 + 2 R_2) t + R_2^2}}{2 R_1 R_2}. \]
Then 
\begin{itemize} \item for $0 < t < t^-$, the inequalities $\lambda_1^- < \lambda_2^- < 0 < \lambda_3^- < \lambda_0^- < \lambda_0^+ < \lambda_1^+ < \lambda_2^+ < \lambda_3^+$ hold and
$ f^{-1}([-1,1]) = [\lambda_1^-,\lambda_2^-] \cup [\lambda_3^-,\lambda_0^-] \cup [\lambda_0^+,\lambda_1^+] \cup [\lambda_2^+,\lambda_3^+]$,
\item for $t^+ < t < 1$, one has that $\lambda_1^- < \lambda_0^- < \lambda_0^+ < \lambda_2^- < 0 < \lambda_3^-  < \lambda_1^+ < \lambda_2^+ < \lambda_3^+$ and $f^{-1}([-1,1]) = [\lambda_1^-,\lambda_0^-] \cup [\lambda_0^+,\lambda_2^-] \cup [\lambda_3^-,\lambda_1^+] \cup [\lambda_2^+,\lambda_3^+]$,
\item for $t = 1$, $f^{-1}([-1,1]) = \left[ -\frac{1}{R_2} - \frac{1}{R_1}, \frac{1}{R_2} - \frac{1}{R_1}\right] \cup \left[-\frac{1}{R_2} + \frac{1}{R_1} ,\frac{1}{R_2} + \frac{1}{R_1} \right]$.
\end{itemize}
Furthermore, in all the cases above, $g^{-1}([-1,1]) = f^{-1}([-1,1])$.
\end{lm}

\begin{proof}
We start with the case $0 < t < 1$. Observe that
\[ f(\lambda ) + 1 = \frac{(t + R_2 \lambda)^2(1-t-R_1 \lambda)^2 - t^2 R_2^2 \lambda^2}{2 t R_2 \lambda (1-t-R_1 \lambda)^2} = \frac{P(\lambda)Q(\lambda)}{2 t R_2 \lambda (1-t-R_1 \lambda)^2}, \]
where $P$ and $Q$ are polynomials defined as
\[ P(\lambda) = (t + R_2 \lambda)(1-t-R_1 \lambda) - t R_2 \lambda, \quad Q(\lambda) = (t + R_2 \lambda)(1-t-R_1 \lambda) + t R_2 \lambda. \]
By expanding $P$, we obtain $P(\lambda) = -R_1 R_2 \lambda^2 + ((1-2t)R_2 - t R_1) \lambda + t(1-t)$; its discriminant is equal to $((1-2t)R_2 - t R_1)^2 + 4 R_1 R_2 t (1-t) > 0$. Therefore, $P$ has two real roots
\[ \lambda_1^{\pm} = \frac{(1-2t)R_2 - t R_1 \pm \sqrt{((1-2t)R_2 - t R_1)^2 + 4 R_1 R_2 t (1-t)}}{2 R_1 R_2}. \]
Note that $\lambda_1^- <  0 < \lambda_1^+$, and that $P(\lambda) \geq 0$ when $\lambda_1^- \leq \lambda \leq \lambda_1^+$ and $P(\lambda) < 0$ otherwise. Similarly, one finds that $Q$ also has two real roots
\[ \lambda_2^{\pm} = \frac{ R_2 - t R_1 \pm \sqrt{(R_2 - t R_1)^2 + 4 R_1 R_2 t(1-t)} }{2 R_1 R_2} \]
with $\lambda_2^- < 0 < \lambda_2^+$, $Q(\lambda) \geq 0$ when $\lambda_2^- \leq \lambda \leq \lambda_2^+$ and $Q(\lambda) < 0$ otherwise. A similar computation shows that
\[ f(\lambda) - 1 = \frac{S(\lambda)T(\lambda)}{2 t R_2 \lambda (1-t-R_1 \lambda)^2},  \]
where $S(\lambda) = (t - R_2 \lambda)(1-t-R_1 \lambda) - t R_2 \lambda$ and $T(\lambda) = (t - R_2 \lambda)(1-t-R_1 \lambda) + t R_2 \lambda$. One finds that $S$ has two real roots
\[ \lambda_3^{\pm} = \frac{ t R_1 + R_2 \pm \sqrt{(R_2 - t R_1)^2 + 4t^2 R_1 R_2} }{2 R_1 R_2} \]
with $0 < \lambda_3^- < \lambda_3^+$, that $S(\lambda) \leq 0$ when $ \lambda_3^- \leq \lambda \leq \lambda_3^+$ and that $S(\lambda) > 0$ otherwise. The case of $T$ is more interesting; its discriminant is equal to $ \Delta = (R_1^2 + 4 R_2^2) t^2 - 2 R_2 (R_1 + 2 R_2) t + R_2^2$; we already saw in the proof of Proposition \ref{prop:focus} that $\Delta \leq 0$ for $t^- \leq t \leq t^+$ and $\Delta > 0$ otherwise. Therefore, $T$ has no real root when $t^- < t < t^+$, has two real roots
\[ \lambda_0^{\pm} = \frac{ (1-2t)R_2  + t R_1 \pm \sqrt{\Delta}}{2 R_1 R_2} \]
when $t < t^-$ or $t > t^+$, and one real root $\lambda_0$ when $t \in \{t^-,t^+\}$. Obviously $\lambda_0^- \leq \lambda_0^+$ with equality when $t \in \{ t^-,t^+ \}$. Moreover, $T(\lambda) \geq 0$ for $t^- \leq t \leq t^+$, and for other values of $t$ we have that $T(\lambda) \leq 0$ when $\lambda_0^- \leq \lambda \leq \lambda_0^+$, $T(\lambda) > 0$ otherwise. When $t < t^-$, $(1-2t)R_2  + t R_1 > 0$ since $t^- \leq 1/2$, and one readily checks that $ \left((1-2t)R_2  + t R_1 \right)^2 - \Delta = 4R_1R_2t(1-t) > 0$, thus $\lambda_0^- > 0$. When $t > t^+$, we have that $(2t-1)R_2  - t R_1 > 2t \left( 2 \sqrt{R_1 R_2} - R_1 \right) > 0$. Since $\left( (2t-1)R_2  - t R_1 \right)^2 - \Delta = 4R_1R_2t(1-t) > 0$, this yields $(2t-1)R_2  - t R_1  > \sqrt{\Delta}$, thus $\lambda_0^+ < 0$.

In order to be able to compute the signs of $f(\lambda) + 1$ and $f(\lambda) - 1$ everywhere, we still need to compare all the $\lambda_i^{\pm}$. The claim follows from careful computations; let us show for instance that $\lambda_2^- > \lambda_1^-$, the other cases involving similar methods. First, observe that we have that $2R_1R_2(\lambda_2^- - \lambda_1^-) = \sqrt{a} - \sqrt{b} + 2tR_2$ with $a = ((1-2t)R_2 - t R_1)^2 + 4 R_1 R_2 t (1-t), \quad b = (R_2 - t R_1)^2 + 4 R_1 R_2 t(1-t)$. One readily checks that
\[ a = b +  4 t^2 R_2^2 - 4 t R_2 (R_2 - t R_1) \quad \text{so} \ \left( \sqrt{a} + 2tR_2 \right)^2 - b = 4tR_2 \left( 2t R_2 + \sqrt{a} - R_2 + tR_1  \right). \]
We will prove that the right hand side of this equality is positive, which will imply that $\lambda_2^- - \lambda_1^->0$. If $R_2 - tR_1 - 2tR_2 \leq 0$, this is obvious. Otherwise, we write
\[ a - \left( R_2 - tR_1 - 2tR_2 \right)^2 = a - (R_2 - tR_1)^2 + 4tR_2(R_2-tR_1) - 4t^2R_2^2 = b - (R_2 - tR_1)^2 = 4 R_1 R_2 t(1-t) > 0. \]
Consequently, $\sqrt{a} >  R_2 - tR_1 - 2tR_2$, and the result follows.

When $t = 1$, we compute
\[ f(\lambda) + 1 = \frac{R_1^2 - R_2^2 + R_1^2 R_2^2 \lambda^2 + 2R_1^2 R_2 \lambda}{2 R_1^2 R_2 \lambda}, \quad f(\lambda) - 1 = \frac{R_1^2 - R_2^2 + R_1^2 R_2^2 \lambda^2 - 2R_1^2 R_2 \lambda}{2 R_1^2 R_2 \lambda} \]
and we conclude by checking the signs of both numerators.

We leave the verification of the last statement to the reader. Actually, the study of $g$ is not too difficult now because we have that
\[ g(\lambda) + 1 = \frac{P(\lambda) T(\lambda)}{2 t R_2^2 \lambda^2 (1-t-R_1\lambda)}, \quad g(\lambda) - 1 = \frac{Q(\lambda) S(\lambda)}{2 t R_2^2 \lambda^2 (1-t-R_1\lambda)} \]
where $P,Q,S,T$ are the polynomials introduced above. Another useful observation is that $\lambda_1^+ < \frac{1-t}{R_1} < \lambda_2^+$ when $0 < t < 1$.
\end{proof}

\begin{prop}
The image of $F$ can be described as follows:
\begin{itemize}
\item when $t = 0$, $F(M)$ is the compact domain enclosed by the parallelogram with vertices at $(-(R_1 + R_2),-1), (R_1-R_2,1), (R_1 + R_2,1)$ and $(R_2 - R_1,-1)$,
\item when $t= 1$, $F(M)$ is the compact domain enclosed by the closed curve obtained as the union of the four following curves:
\begin{enumerate}
\item the horizontal segment $H=1$, $-(R_1 + R_2) \leq J \leq R_1 + R_2$,
\item the horizontal segment $H=-1$, $R_1 - R_2 \leq J \leq R_2 - R_1$,
\item the parametrized curve $J = R_1 R_2 \lambda, H = \frac{R_1^2 R_2^2 \lambda^2 - R_1^2 - R_2^2}{2 R_1 R_2}$, $-\frac{1}{R_2} - \frac{1}{R_1} \leq \lambda \leq \frac{1}{R_2} - \frac{1}{R_1}$,
\item  the parametrized curve $J = R_1 R_2 \lambda, H = \frac{R_1^2 R_2^2 \lambda^2 - R_1^2 - R_2^2}{2 R_1 R_2}$, $-\frac{1}{R_2} + \frac{1}{R_1} \leq \lambda \leq \frac{1}{R_2} + \frac{1}{R_1}$,
\end{enumerate}
\item when $0 < t < 1$, the boundary of $F(M)$ consists of the points $(J_{\lambda},H_{\lambda}) \in \R^2$ with
\[ J_{\lambda} = \frac{t(1-t)f(\lambda) - R_2 (1-t) \lambda + R_1 R_2 \lambda^2}{t \lambda}, \quad H_{\lambda} = \frac{t(1-t) - t R_1 \lambda + R_1 R_2 \lambda^2 f(\lambda)}{R_2 \lambda} \]
for $\lambda \in f^{-1}([-1,1])$. It is a closed continuous curve $\mathcal{C}$ in $\R^2$, and $F(M)$ is the compact domain enclosed by $\mathcal{C}$.
\end{itemize}
\end{prop}

One can see what this image looks like in Figure \ref{fig:spectrum}. Of course, the first part of the proposition is not surprising, since it is easy to check that the system is toric if $R_1 = 1$, and toric up to a ``vertical scaling'' (that is by modifying the second factor of the symplectic form by a multiplicative constant) otherwise when $t = 0$, hence the image of the momentum map is a convex polygon \cite{Ati,GuiSte}. Observe also that the results of this section are consistent with what we found when studying the critical points of maximal corank of $F$. Indeed, when $t^- < t < t^+$, there are only three elliptic-elliptic points (corresponding to corners on the boundary of $F(M)$) and the boundary of $F(M)$ is the union of three parametrized curves, while for $0 \leq t < t^-$ and $t^+ < t \leq 1$ there are four elliptic-elliptic points and the boundary of $F(M)$ is the union of four parametrized curves.

\begin{proof}
We leave the cases $t=0$ and $t=1$ to the reader and assume that $t \notin \{ 0,1 \}$. As the image of a compact, connected manifold by a continuous function, $F(M)$ is compact and connected. We saw in the proof of Lemma \ref{lm:crit_points} that the only critical points of $J$ are $m_0, m_1, m_2$ and $m_3$. Hence, for $E \neq J(m_i)$, the level set $J^{-1}(E)$ is a smooth compact manifold, therefore $H$ admits a minimum and a maximum on $J^{-1}(E)$. The critical points of $H$ on $J^{-1}(E)$ are the critical points of corank one of $F$, thus they are given by Proposition \ref{prop:crit_corank_one}. A straightforward computation shows that their images by $J$ and $H$ are the expressions $J_{\lambda}$ and $H_\lambda$ written in the statement of the proposition, and for a given $E \neq J(m_i)$, there are exactly two values of $H_{\lambda}$ such that $J_{\lambda} = E$, the minimum and the maximum mentioned above. 
\end{proof}

\begin{dfn}
An integrable system $G = (g_1,g_2)$ on a compact connected four dimensional symplectic manifold $N$ is said to be of \strong{toric type} if there exists an effective Hamiltonian $\T^2$-action on $N$ whose momentum map is of the form $f \circ G$, where $f$ is a local diffeomorphism from $G(N)$ to its image.
\end{dfn}

\begin{cor}
\label{cor:semitor}
The system $F = (J,H)$ of coupled angular momenta forms an integrable system of toric type if $t < t^-$ or $t > t^+$, and a semitoric system with one focus-focus singularity if $t^- < t < t^+$. It is even toric when $t = 0$ and $R_1 = 1$.
\end{cor}

\begin{proof}
We already saw that $J$ and $H$ Poisson commute, and that the critical points of corank two of $F$ are all non-degenerate. One readily checks that this is also true for its critical points of corank one, which are therefore of elliptic-transverse type by the above considerations (see Appendix B for more details). Hence $(J,H)$ is an integrable system; since no singularity has hyperbolic components, it is almost-toric. It is easy to check that the Hamiltonian flow of $J$ at time $s$ corresponds to the rotation of angle $s$ around the $z_1$-axis in the first factor and the $z_2$-axis in the second factor, thus $J$ generates an effective circle action, so the system is semitoric. By Corollary $3.5$ in \cite{VNpoly}, it is of toric type for $0 \leq t < t^-$ and $t^+ < t \leq 1$ because it has no focus-focus singularity.

\end{proof}

\section{Symplectic invariants in the $t = 1/2$ case}

It follows from the previous study that $F = (J,H)$ is always a simple semitoric system with one focus-focus value when $t=1/2$, since we always have $t^- < 1/2 < t^+$. As already mentioned earlier, the symplectic classification of these systems has been achieved by the second author and V\~u Ng\d{o}c \cite{PelVuSemi,PelVuConstr}; it involves five invariants that we quickly describe here for the sake of completeness.

\subsection{Description of the invariants}
\label{sect:symp_invariant}

As it is the case in our example, we will assume that $(M,\omega)$ is a compact connected symplectic manifold; in the non-compact case, the polygonal invariant is not exactly a polygon in the usual sense. Let $F = (J,H)$ be a simple semitoric system on $(M,\omega)$. The first invariant is extremely simple; it is the number $m_f$ of focus-focus critical values of $F$, which in our case is equal to one. Consequently, and since the notation becomes heavy when there is more than one focus-focus critical value, we will only state the definitions of the other invariants in the case $m_f = 1$. We only explain here the main ingredients appearing in these invariants; for a precise account on these, we refer the reader to \cite{PelVuSemi,SepeVN}.

\paragraph{The Taylor series invariant.} Let $m_0$ be the unique critical point of $F$ of focus-focus type and let $c_0 = F(m_0)$ be the corresponding critical value. Endowing $\R^4$ with coordinates $(\hat{u}_1,\hat{u}_2,\hat{\xi}_1,\hat{\xi}_2)$ and symplectic form $d\hat{u}_1 \wedge d\hat{\xi}_1 + d\hat{u}_2 \wedge d\hat{\xi}_2$, it follows from Eliasson's normal form theorem \cite{Elia} that there exist neighborhoods $U$ of $m_0$ in $M$ and $V$ of the origin in $\R^4$, a local symplectomorphism $\phi: V \to U$ sending the origin to $m_0$, and a local diffeomorphism $g = (g_1, g_2): (\R^2,0) \to (\R^2,0)$ with $\dpar{g_2}{y} > 0$ (note that this sign is important and was forgotten in \cite{PelVuSemi}) such that $F \circ \phi = g \circ q$, where the components of $q$ satisfy
\begin{equation} q_1 = \hat{u}_1 \hat{\xi}_2 - \hat{u}_2 \hat{\xi}_1, \quad q_2 = \hat{u}_1 \hat{\xi}_1 + \hat{u}_2 \hat{\xi}_2.  \label{eq:Eliasson_nf} \end{equation}
Hence there exists a global momentum map $G$ for the singular foliation defined by $F$ which agrees with $q \circ \phi^{-1}$ on $U$. Let us write $G = (G_1,G_2)$ and for $z \in \R^2 \simeq \C$, $\Lambda_z = G^{-1}(z)$. It follows from the above normal form that near $m_0$, the trajectories of the Hamiltonian flow of $g_1$ must be periodic, with primitive period $2\pi$. For $z \neq 0$, let $A$ be a point on $\Lambda_z$, and define the quantity $\tau_2(z) > 0$ as the smallest positive time it takes the Hamiltonian flow of $G_2$ to meet the trajectory of the Hamiltonian flow of $G_1$ passing through $A$. Let $\tau_1(z) \in \R / 2\pi \Z$ be the time that it takes to go back to $A$ from this meeting point following the flow of $G_1$. Observe that the two numbers $\tau_1(z), \tau_2(z)$ do not depend on the choice of $A \in \Lambda_z$. Now, let $\log$ be some determination of the complex logarithm. It was proved in \cite[Proposition $3.1$]{VNsemi} that $\sigma_1, \sigma_2$ defined as
\[ \sigma_1(z) = \tau_1(z) - \Im(\log z), \quad \sigma_2(z) = \tau_2(z) + \Re(\log z) \]
extend to smooth single-valued functions in a neighborhood of $z=0$ and that the differential form $\sigma = \sigma_1 dz_1 + \sigma_2 dz_2$ is closed. In fact, one must be very careful when translating the results of \cite{VNsemi}, because in this paper another convention was adopted, namely $q_1$ and $q_2$ where inverted. We may, and will, choose the lift of $\tau_2$ to $\R$ such that $\sigma_2(0)$ belongs to $[0,2\pi)$. Let $S$ be the unique smooth function defined near the origin in $\R^2$ such that $dS = \sigma$ and $S(0,0) = 0$; the \strong{Taylor series invariant} $S^{\infty} \in \R[[X,Y]]$ is the Taylor expansion of $S$ at the origin. It is of the form $S^{\infty} = a_1 X + a_2 Y + \sum_{i+j > 1} b_{ij} X^i Y^j$.

\paragraph{The polygonal invariant.} We consider the plane $\R^2$ with its standard affine structure and orientation. Let $\mathcal{T} \subset GL(2,\Z) \ltimes \R^2$ be the subgroup of integral-affine transformations leaving a vertical line invariant; in other words, $\mathcal{T}$ consists of integral-affine transformations obtained by composing a vertical translation with a transformation of the form $T^k$, $k \in \Z$, with
\[ T = \begin{pmatrix} 1 & 0 \\ 1 & 1 \end{pmatrix} \in GL(2,\Z). \]
Let $\mathrm{Vert}(\R^2)$ be the set of vertical lines in $\R^2$, and choose $\ell \in \mathrm{Vert}(\R^2)$; it divides the plane into two half-planes, $H^{\text{left}}$ on the left and $H^{\text{right}}$ on the right. Let $n \in \Z$. Fix an origin in $\ell$, and define the piecewise integral-affine transformation $t_{\ell}^n: \R^2 \to \R^2$ as the identity on $H^{\text{left}}$ and as $T^n$ on $H^{\text{right}}$. Now, let $\ell_0$ be the vertical line passing through $c_0$. Let $\mathcal{G} = \{-1,1\}$; for $\epsilon \in \mathcal{G}$, let $\ell_0^{\epsilon} \subset \ell_0$ be the vertical half-line starting at $c_0$ and extending upwards if $\epsilon = 1$ and downwards if $\epsilon = -1$. From this data, one can construct a rational convex polygon $\Delta$, that is a convex polygon whose edges are directed along vectors with rational coefficients, associated with $F$, as follows. Let $B = F(M)$ and let $B_r$ be the set of regular values of $F$, which is endowed with an integral-affine structure coming from action variables.

\begin{thm}[{\cite[Theorem $3.8$]{VNpoly}}]
For every $\epsilon \in \mathcal{G}$, there exists a homeomorphism $f = f_{\epsilon}: B \to \Delta$, where $\Delta = f(B) \subset \R^2$, such that
\begin{itemize}
\item $f_{|B \setminus \ell_0^{\epsilon}}$ is a diffeomorphism into its image,
\item $f_{|B_r \setminus \ell_0^{\epsilon}}$ sends the integral affine structure of $B_r$ to the standard integral affine structure of $\R^2$,
\item $f$ preserves $J$, i.e. is of the form $f(x,y) = (x,f_2(x,y))$,
\item $\Delta$ is a rational convex polygon.
\end{itemize}
\end{thm}

Such a $\Delta$ is called a \strong{generalized toric moment polygon} for $(M,\omega,F)$, and $\mu = f \circ F$ is called a \strong{generalized toric momentum map} for $(M,\omega,F)$. This polygon $\Delta$, however, is not yet the invariant that we are trying to define since it is highly non unique. It depends on the choice of
\begin{itemize}
\item an initial set of action variables $f_0$ near a regular Liouville torus; if we choose a different one, $f$ will be composed on the left with an element $\tau$ of $\mathcal{T}$, and $\Delta$ will become $\tau(\Delta)$,
\item the choice of $\epsilon \in \mathcal{G}$; if we choose $\delta$ instead of $\epsilon$, $f$ will be composed on the left by $t_{\ell_0}^n$ with $n = (\epsilon - \delta)/2$, and $\Delta$ will become $t_{\ell_0}^n(\Delta)$.
\end{itemize}
A \emph{weighted polygon} is a triple of the form $\Delta_{\text{weight}} = \left( \Delta, \ell_{\lambda}, \epsilon \right)$ where $\Delta$ is a rational convex polygon, $\ell_{\lambda}$ is the vertical line $\{x=\lambda\} \subset \R^2$ and $\epsilon \in \mathcal{G}$. The group $\mathcal{G} \times \mathcal{T}$ acts on the set of weighted polygons via the formula
\[ (\delta,\tau) \cdot \left( \Delta, \ell_{\lambda}, \epsilon \right) = \left( t_{\ell_{\lambda}}^n(\tau(\Delta)), \delta \epsilon \right) \]
with $n = (\epsilon - \delta)/2$. The $\mathcal{G}$ part of this action may not preserve the convexity of $\Delta$, but when $\Delta$ is a generalized toric moment polygon for a semitoric system, it does. Hence we say that a weighted polygon $\left( \Delta, \ell_{\lambda}, \epsilon \right)$ is \strong{admissible} when the convexity of $\Delta$ is preserved by the $\mathcal{G}$-action, and we define $\mathcal{W}\text{Polyg}(\R^2)$ as the set of all admissible weighted polygons. Let $(\Delta,\ell_0,\epsilon)$ be an admissible weighted polygon obtained as in the above theorem; then the \strong{polygonal invariant} of $(M,\omega,F)$ is the $(\mathcal{G} \times \mathcal{T})$-orbit
\[ \left( \mathcal{G} \times \mathcal{T} \right) \cdot \left( \Delta, \ell_0, \epsilon \right) \in \mathcal{W}\text{Polyg}(\R^2) / (\mathcal{G} \times \mathcal{T}).  \]

\paragraph{The height invariant.} Let $\mu$ be a generalized toric momentum map for $(M,\omega,F)$, and let $\Delta = \mu(M)$ be the associated generalized toric moment polygon. Then $\mu(m_0)$ belongs to the intersection of $\ell_0$ with the interior of $\Delta$. The vertical distance
\[ h = \pi_2(\mu(m_0)) - \min_{p \in \Delta \cap \ell_0} \pi_2(p), \]
where $\pi_2: \R^2 \to \R$ is the projection to the second factor, does not depend on the choice of $\mu$, and is called the \strong{height invariant} of $(M,\omega,F)$. In fact, this height invariant has the following geometric interpretation, which we will use to compute it. Let $M^{\text{red}} = J^{-1}(J(m_0))/S^1$ be the reduced manifold with respect to the $S^1$-action generated by $J$; it is endowed with a canonical symplectic form $\omega^{\text{red}}$. The height invariant $h$ is equal to the volume of $\{H < H(m_0)\}$ in $M^{\text{red}}$ (this makes sense because $H$ is invariant under the $S^1$-action), with respect to $|\omega^{\text{red}}|/2\pi$.

\paragraph{The twisting index invariant.} We only sketch the description of the twisting index invariant, and refer the reader to \cite[Section $5.2$]{PelVuSemi} for more details. The key point is the existence of a privileged toric momentum map $\nu$ in a neighborhood of $m_0$. Now, let $\mu$ be a generalized toric momentum map for $(M,\omega,F)$, and let $(\Delta,\ell_0,\epsilon)$ be the corresponding weighted polygon; there exists an integer $k \in \Z$ such that $\mu = T^k \nu$ near $m_0$. This integer $k$ is called the twisting index of $(\Delta,\ell_0,\epsilon)$. If we compose $\mu$ on the left by an affine transformation $\tau \in \mathcal{T}$ with linear part $T^r$, the twisting index becomes $k+r$. So we consider the following action of $\mathcal{G} \times \mathcal{T}$ on $\mathcal{W}\text{Polyg}(\R^2) \times \Z$:
\[ (\delta,\tau) \cdot \left( \Delta,\ell_{\lambda},\epsilon,k \right) = \left( t_n(\tau(\Delta)), \ell_{\lambda}, \delta \epsilon, k + r \right) \]
where $n = (\epsilon - \delta)/2$ and $T_r$ is the linear part of $\tau$. Now let $(\Delta,\ell_0,\epsilon)$ be a weighted polygon for $(M,\omega,F)$ and let $k$ be its twisting index; the \strong{twisting index invariant} of $(M,\omega,F)$ is the $(\mathcal{G} \times \mathcal{T})$-orbit
\[ (\mathcal{G} \times \mathcal{T}) \cdot \left( \Delta,\ell_0,\epsilon,k \right) \in \left( \mathcal{W}\text{Polyg}(\R^2) \times \Z \right) / (\mathcal{G} \times \mathcal{T}) \]
Consequently, one can always find a weighted polygon for which the twisting index is zero; nevertheless, one should keep in mind that fixing the representative $\Delta$ fixes the twisting index.

Coming back to our problem, our goal is to compute some of these invariants for the system $(J,H)$ of coupled angular momenta when $t=1/2$:
\begin{equation} J = R_1 z_1 + R_2 z_2, \quad H = \frac{1}{2} \left( z_1 + x_1x_2 + y_1y_2 + z_1z_2 \right). \label{eq:fiber}\end{equation}
Actually, we will only compute the first two terms of the Taylor series invariants, and for some fixed value of the pair $(R_1,R_2)$; however, we will describe the method carefully so that one can compute these for other values of $(R_1,R_2)$.

\subsection{Parametrization of the singular fiber}

We start by parametrizing $F^{-1}(c_0)$, which is given by the points $(x_1,y_1,z_1,x_2,y_2,z_2) \in M$ such that
\[ R_1 z_1 + R_2 z_2 = R_1 - R_2, \quad  z_1 + x_1x_2 + y_1y_2 + z_1z_2 = 0. \]
Observe that neither $\{ (0,0,-1) \} \times \S^2 $ nor $ \S^2 \times \{ (0,0,1) \}$ intersects $F^{-1}(c_0)$. Indeed, the quantity $H(0,0,-1,x_2,y_2,z_2) = -(1+z_2)/2$ vanishes if and only if $z_2 = -1$, but in this case $J(0,0,-1,x_2,y_2,z_2) = -(R_1+R_2) \neq R_1 - R_2$; similarly, if $H(x_1,y_1,z_1,0,0,1) = z_1$ vanishes, then $J(x_1,y_1,z_1,0,0,1) = R_2 \neq R_1-R_2$. Therefore, we can identify $F^{-1}(c_0)$ with a subset of $\C^2$ by means of stereographic projections, from the south pole to the equatorial plane on the first factor, and from the north pole to the equatorial plane on the second factor. In other words, we consider the diffeomorphisms
\[ \pi_S: \S^2 \setminus \{ (0,0,-1) \} \to \C, \quad (x,y,z) \mapsto \frac{x-iy}{1+z}, \qquad \pi_N: \S^2 \setminus \{ (0,0,1) \} \to \C, \quad (x,y,z) \mapsto \frac{x+iy}{1-z}. \]
Then we get a diffeomorphism
\begin{align*} \Psi: & \left( \S^2 \setminus \{(0,0,-1)\} \right) \times \left( \S^2 \setminus \{(0,0,1)\} \right) \subset M & \to & \hspace{25mm} \C^2 \\
 & \hspace{25mm} (x_1,y_1,z_1,x_2,y_2,z_2) & \mapsto & \ \left( \pi_S(x_1,y_1,z_1), \pi_N(x_2,y_2,z_2) \right). \end{align*}
We want to describe the image $\Lambda_0:=\Psi(F^{-1}(c_0))$ of the singular fiber; note that one has $\Psi(m_0) = (0,0)$. Now, let $(z,w) \in \C^2$; it is standard that
\begin{equation} \pi_S^{-1}(z) = \frac{1}{1+|z|^2} \left( 2 \Re z, -2\Im z, 1 - |z|^2 \right), \quad \pi_N^{-1}(w) = \frac{1}{1+|w|^2} \left( 2 \Re w, 2\Im w, |w|^2 -1 \right). \label{eq:invstereo} \end{equation}
In view of Equation (\ref{eq:fiber}), in these coordinates, $J, H$ read 
\begin{equation} J =  \frac{R_1(1-|z|^2)}{1+|z|^2} + \frac{R_2(|w|^2-1)}{1+|w|^2}, \qquad H = \frac{(1-|z|^2)|w|^2 + 2 \Re(zw)}{(1+|z|^2)(1+|w|^2)}, \label{eq:JH_comp}\end{equation}
and a straightforward computation shows that $(z,w)$ belongs to $\Lambda_0$ if and only if
\begin{equation} \begin{cases}  R_2(1+|z|^2 )|w|^2 = R_1|z|^2(1+|w|^2),  \\  \vspace{-3mm} \\  (1-|z|^2)|w|^2 + 2 \Re(zw) = 0. \end{cases} \label{eq:lambda_0} \end{equation}

We will parametrize $\Lambda_0$ with the help of polar coordinates; if $z = \rho \exp(i\theta)$ and $w = \eta \exp(i\varphi)$, the system (\ref{eq:lambda_0}) becomes
\[ \begin{cases}  R_2(1+\rho^2)\eta^2 = R_1 \rho^2 (1+\eta^2), \\  \vspace{-3mm} \\  (1-\rho^2)\eta^2 + 2 \rho \eta \cos(\theta + \varphi) = 0. \end{cases} \]
By using the first equation and substituting $\eta$ into the second equation, we obtain
\begin{equation} \eta = \rho \sqrt{\frac{R_1}{R_2 + (R_2-R_1)\rho^2}}, \quad (1-\rho^2) \rho^2 \sqrt{\frac{R_1}{R_2 + (R_2-R_1)\rho^2}} + 2 \rho^2 \cos(\theta + \varphi) = 0. \label{eq:eta}\end{equation}
When $\rho \neq 0$, this becomes
\begin{equation} \cos(\theta + \varphi) = \frac{\rho^2 - 1}{2} \sqrt{\frac{R_1}{R_2 + (R_2-R_1)\rho^2}}. \label{eq:cos}\end{equation}
So necessarily, the right hand side of this equality belongs to $[-1,1]$, which is equivalent to the fact that $P(\rho^2) \leq 0$ where
\[ P(\tau) = R_1(\tau-1)^2 - 4 (R_2 + (R_2 - R_1) \tau)  = R_1 (\tau+1)\left( \tau + 1 - 4\frac{R_2}{R_1} \right). \]
So Equation (\ref{eq:cos}) can be satisfied only if $\rho$ belongs to $I = [0,\zeta]$, where $\zeta = \sqrt{\frac{4 R_2}{R_1} - 1}$, and when this is the case, we get
\[ \varphi = \varepsilon \arccos\left( \frac{\rho^2-1}{2} \sqrt{\frac{R_1}{R_2 + (R_2-R_1)\rho^2}} \right) - \theta \]
for $\varepsilon = \pm 1$. Consequently, we define two maps $S_{\varepsilon}: I \times \R / 2\pi \Z \to \C^2$, $\varepsilon = \pm 1$, by the formula 
\[ S_{\varepsilon}(\rho,\theta) = \left( \rho \exp(i\theta), \rho f(\rho)  \exp\left( i \left( \varepsilon \arccos\left( \frac{\rho^2-1}{2} f(\rho) \right) - \theta \right)  \right) \right), \quad  f(\rho) = \sqrt{\frac{R_1}{R_2 + (R_2-R_1)\rho^2}}, \]
and set $\Lambda_0^{\varepsilon} = S_{\varepsilon}(I \times \R / 2\pi \Z)$.

\begin{prop}
For $\varepsilon = \pm 1$, the map $S_{\varepsilon}$ is continuous, and is a diffeomorphism from $(0,\zeta) \times \R / 2\pi\Z$ to $S_{\varepsilon}((0,\zeta) \times \R / 2\pi\Z)$. Moreover, $\Lambda_0 = \Lambda_0^{-1} \cup \Lambda_0^1$ and $\Lambda_0^{-1} \cap \Lambda_0^1 = \{(0,0)\} \cup \mathcal{C}$ where 
\[ \mathcal{C} = \left\{ \left( \sqrt{\frac{4R_2}{R_1} - 1} \exp(i\theta), \frac{\sqrt{R_1(4R_2-R_1)}}{2R_2-R_1} \exp(-i\theta) \right), \ \theta \in \R / 2 \pi \Z \right\}.\]
\end{prop}

This means that $\Lambda_0$ consists of two cylinders glued along $(0,0)$ on one end and along $\mathcal{C}$ on the other end. Therefore, $\Lambda_0$ is a torus with a pinch at $(0,0)$ (see Figure \ref{fig:pinched}); of course, we already knew it, since a focus-focus critical fiber is always a pinched torus (see for instance \cite[Proposition $6.2$]{VN}).

\begin{figure}
\begin{center}
\includegraphics[scale=0.7]{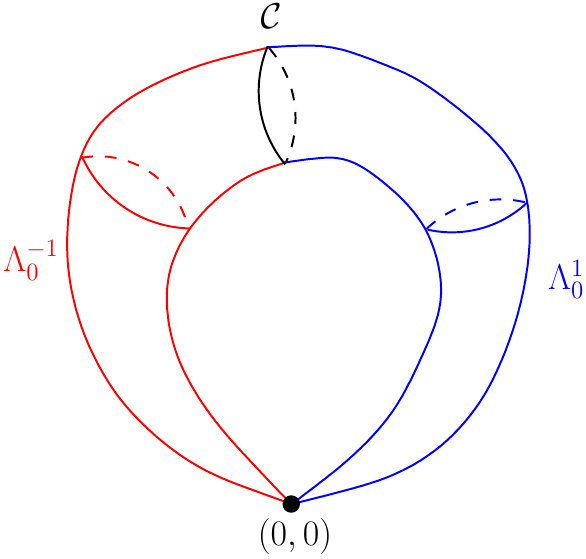}
\end{center}
\caption{The critical fiber $\Lambda_0$.}
\label{fig:pinched}
\end{figure}

\begin{proof}
The fact that $\Lambda_0 = \Lambda_0^{-1} \cup \Lambda_0^1$ comes from the above considerations. Let us prove that $\Lambda_0^{-1} \cap \Lambda_0^1 = \{(0,0)\} \cup \mathcal{C}$; let $(z,w) \in \C^2$ and assume that
\[ (z,w) = \begin{cases} (\rho_1 \exp(i\theta_1), \eta_1 \exp(i\varphi_1)) \in \Lambda_0^{-1}, \\ (\rho_2 \exp(i\theta_2), \eta_2 \exp(i\varphi_2)) \in \Lambda_0^1. \end{cases} \]
If $\rho_1 = 0$, then necessarily $\rho_2 = 0$, and we get that $\eta_1 = 0 = \eta_2$, so $(z,w) = (0,0)$. Otherwise, we have that $\rho_1 = \rho_2, \theta_1 = \theta_2, \eta_1 = \eta_2$ and $\varphi_1 = \varphi_2$. But then we have that 
\[ \varphi_1 = -\arccos\left( \frac{\rho_1^2-1}{2} f(\rho_1) \right) - \theta_1, \quad \varphi_2 = \arccos\left( \frac{\rho_1^2-1}{2} f(\rho_1) \right) - \theta_1, \]
therefore we obtain that $\arccos\left( \frac{\rho_1^2-1}{2} f(\rho_1) \right) = 0$. But we saw that this only happens when $\rho_1 = \zeta$, and in this case
\[ (z,w) = \left( \left(\sqrt{\frac{4R_2}{R_1} - 1}\right) \exp(i\theta_1), \frac{\sqrt{R_1(4R_2-R_1)}}{2R_2-R_1} \exp(-i\theta_1) \right), \]
i.e. $(z,w)$ belongs to $\mathcal{C}$. The other statements are easily checked.
\end{proof}

In what follows, we will also need to compute the Hamiltonian vector fields $X_J$ and $X_H$ of $J$ and $H$ on $F^{-1}(c_0) \setminus \{ m_0 \}$. But we will eventually want to use our parametrization and work on $\Lambda_0 \setminus \{(0,0)\}$; we will slightly abuse notation and use $J,H$ instead of $J \circ \Psi^{-1}, H \circ \Psi^{-1}$ and $X_J, X_H$ for the pushforwards of $X_J$ and $X_H$ by $\Psi$.

\begin{lm}
\label{lm:XJ_XH}
The Hamiltonian vector fields of $J$ and $H$ read
\[ X_J = i \left( -z \frac{\partial}{\partial z} + \bar{z} \frac{\partial}{\partial \bar{z}} + w \frac{\partial}{\partial w} - \bar{w} \frac{\partial}{\partial \bar{w}} \right), \ X_H = \frac{i}{2} \left( \lambda_1(z,w) \frac{\partial}{\partial z} - \overline{\lambda_1(z,w)} \frac{\partial}{\partial \bar{z}} + \lambda_2(z,w) \frac{\partial}{\partial w} - \overline{\lambda_2(z,w)} \frac{\partial}{\partial \bar{w}}  \right) \]
where the functions $\lambda_1, \lambda_2$ are defined as
\[ \lambda_1(z,w) = \frac{\bar{w}-2z|w|^2 - z^2 w}{R_1(1+|w|^2)} , \quad \lambda_2(z,w) = \frac{\bar{z} + w - w|z|^2 - z w^2}{R_2(1+|z|^2)}. \]
\end{lm}

We can also simplify this when $(z,w)$ belongs to $\Lambda_0 \setminus \{(0,0)\}$.

\begin{lm}
\label{lm:XH_Lambda}
On $\Lambda_0 \setminus \{(0,0)\}$, we have that $ X_H = \frac{i}{2} \left( -\frac{z(z + \bar{w})}{R_2 \bar{w}} \frac{\partial}{\partial z} + \frac{\bar{z}(\bar{z}+w)}{R_2 w} \frac{\partial}{\partial \bar{z}} - \frac{w^2}{R_1 \bar{z}} \frac{\partial}{\partial w} + \frac{\bar{w}^2}{R_1 z} \frac{\partial}{\partial \bar{w}}  \right)$.
\end{lm}

These two lemmas are proved in Appendix A.

\subsection{The Taylor series invariant}

The description of the Taylor series invariant can seem very complicated to work with; fortunately, one can use the following results to compute its first terms. Recall that there exist local symplectic coordinates $(\hat{u}_1,\hat{u}_2,\hat{\xi}_1,\hat{\xi}_2)$ (Eliasson coordinates) on a neighborhood of $m_0$ and a local diffeomorphism $g: (\R^2,0) \to (\R^2,0)$ such that $F = g \circ (q_1,q_2)$ where $q_1, q_2$ are as in Equation (\ref{eq:Eliasson_nf}). Let $\kappa_{1,0}$ and $\kappa_{2,0}$ be the differential forms defined near $m_0$ in $F^{-1}(c_0) \setminus \{ m_0 \}$ by the conditions
\[ \kappa_{1,0}(X_{q_1}) = -1, \quad \kappa_{1,0}(X_{q_2}) = 0, \quad \kappa_{2,0}(X_{q_1}) = 0, \quad \kappa_{2,0}(X_{q_2}) = -1. \]
Let $S^{\infty} = a_1 X + a_2 Y + \sum_{i+j > 1} b_{ij} X^i Y^j$ be the Taylor series invariant of $(M,\omega,F)$.

\begin{thm}[{\cite[Proposition $6.8$]{VN}, \cite[Theorem $2.5$]{PelVN}}]
Let $\gamma_0$ be a radial simple loop in $F^{-1}(c_0)$, i.e. a simple loop starting on the local unstable manifold at $m_0$ and coming back to $m_0$ via the local stable manifold. Then
\begin{equation} a_1 = \lim_{(s,t) \to (0,0)} \left( \int_{A_0 = \gamma_0(s)}^{B_0=\gamma_0(1-t)} \kappa_{1,0} + \mu_{B_0} - \nu_{A_0} \right), \quad  a_2 = \lim_{(s,t) \to (0,0)} \left( \int_{A_0 = \gamma_0(s)}^{B_0=\gamma_0(1-t)} \kappa_{2,0} + \ln(r_{A_0} \sigma_{B_0})  \right) \label{eq:radial_aj}\end{equation}
where for any point $C$ close to $m_0$ with Eliasson coordinates $(\hat{u}_{1,C},\hat{u}_{2,C},\hat{\xi}_{1,C},\hat{\xi}_{2,C})$, the coordinates $(r_C,\nu_C)$ (respectively $(\sigma_C,\mu_C)$) are the polar coordinates of $\hat{u}_{1,C} + i \hat{u}_{2,C}$ (respectively $\hat{\xi}_{1,C} + i \hat{\xi}_{2,C}$).
\end{thm}

At first glance, this seems easier to handle than the original definition, but these formulas involve Eliasson coordinates, which may be extremely difficult to compute. Nevertheless, the following lemma states that we only need to use a first order approximation of these coordinates.

\begin{lm}[{\cite[Lemma $2.13$]{PelVN}}]
\label{lm:linear_Elia}
The theorem remains true with linear Eliasson coordinates $(u_1,u_2,\xi_1,\xi_2)$ instead of $(\hat{u}_1,\hat{u}_2,\hat{\xi}_1,\hat{\xi}_2)$, i.e. local symplectic coordinates  such that the Hessian of $F$ at $m_0$ equals $\phi \circ (q_1, q_2)$ for some linear map $\phi = (\phi_1,\phi_2)$ with $\frac{\partial \phi_2}{\partial y} >0$, where $q_1 = u_1 \xi_2 - u_2 \xi_1$ and $q_2 = u_1 \xi_1 + u_2 \xi_2$.
\end{lm}

So we see that in order to derive the first terms in the Taylor series invariant, the first step is to compute linear Eliasson coordinates at the focus-focus critical point $m_0$. For this purpose, working with general parameters $R_1$ and $R_2$ leads to very complicated expressions, so it is better to fix their values once and for all. Here we choose $R_1 = 1, R_2 = 5/2$ because it simplifies the computations, but one could fix any other values and apply the following method. For our choice of parameters, the symplectic form and $J, H$ read
\[ \omega = -\left( \omega_{\S^2} \oplus \frac{5}{2} \omega_{\S^2}\right) , \quad J = z_1 + \frac{5}{2} z_2, \quad H = \frac{1}{2} \left( z_1 + x_1x_2 + y_1y_2 + z_1z_2 \right). \]

\paragraph{Linear Eliasson coordinates.} We look for symplectic coordinates $(u_1,u_2,\xi_1,\xi_2)$ on $\R^4$ satisfying the requirements of Lemma \ref{lm:linear_Elia}.

\begin{prop}
\label{prop:linear_eliasson}
Let $\phi: T_{m_0} M \to \R^4$ be the linear isomorphism given by the formula $\phi(x_1,x_2,y_1,y_2) = (u_1,u_2,\xi_1,\xi_2)$ with
\vspace{-5mm}
\[ \begin{cases} u_1 = x_1 + 3 y_1 + 5 x_2, \quad u_2 = 3 x_1 - y_1 - 5 y_2, \vspace{2mm} \\  \xi_1 = \frac{1}{6} x_1 + \frac{1}{12} x_2 + \frac{1}{4} y_2 , \quad \xi_2 = -\frac{1}{6} y_1 + \frac{1}{4} x_2 - \frac{1}{12} y_2. \end{cases} \]
Then $(\phi^{-1})^* \omega_{m_0} = du_1 \wedge d\xi_1 + du_2 \wedge d\xi_2$ and $\mathrm{Hess}(B \circ (F-F(m_0)) \circ \phi^{-1} = (q_1,q_2)$ where
\[ B = \begin{pmatrix} 1 & 0 \\ -\frac{1}{3} & \frac{10}{3} \end{pmatrix}. \]
\end{prop}

One could check these claims directly, using the explicit formulas
\[ q_1 = -\frac{1}{2} (x_1^2 + y_1^2) + \frac{5}{4} (x_2^2 + y_2^2), \quad q_2 = \frac{1}{6} (x_1^2 + y_1^2) + \frac{5}{12} (x_2^2 + y_2^2) + \frac{5}{3} (x_1 x_2 + y_1 y_2), \]
and the fact that for our choice of parameters, the quadratic parts of $J$ and $H$ are respectively
\[ - \frac{1}{2} (x_1^2 + y_1^2) + \frac{5}{4} (x_2^2 + y_2^2), \qquad \frac{1}{4} (x_2^2 + y_2^2) + \frac{1}{2}(x_1x_2 + y_1 y_2), \]
see the proof of Proposition \ref{prop:focus}. However, this would not give any insight on how to obtain these coordinates, hence we describe the general method. The idea is to find a basis of $T_{m_0} M$ in which the matrix $A = \Omega^{-1} d^2H(m_0)$ (again, see the proof of Proposition \ref{prop:focus} for notation) becomes
\begin{equation} \begin{pmatrix} -\alpha & \beta & 0 & 0 \\ -\beta & -\alpha & 0 & 0 \\ 0 & 0 & \alpha & \beta \\ 0 & 0 & -\beta & \alpha \end{pmatrix} \label{eq:mat_focus}\end{equation}
for some $\alpha, \beta \in \R$; then $(u_1,u_2,\xi_1,\xi_2)$ will be the coordinates associated with this basis, and we will have $\mathrm{Hess}(J) \circ \phi^{-1} = q_1$ and $\mathrm{Hess}(H) \circ \phi^{-1} = \beta q_1 + \alpha q_2$. Consequently, the matrix $B$ in the proposition will be given as the inverse of $\begin{pmatrix} 1 & 0 \\ \beta & \alpha \end{pmatrix}$. The details of the computation are available in Appendix A.

\paragraph{Construction of a radial simple loop in $\Lambda_0$.} The second ingredient that we will need is a radial simple loop $\gamma_0$ in $\Lambda_0$, that we will construct as an integral curve of the radial vector field $X_{q_2}$. It follows from the previous proposition that $X_{q_2} = (-X_J + 10 X_H)/3$, hence Lemmas \ref{lm:XJ_XH} and \ref{lm:XH_Lambda} yield
\[ X_{q_2} = -\frac{i}{3} \left( \frac{z(2z+\bar{w})}{\bar{w}} \frac{\partial}{\partial z} -  \frac{\bar{z}(2\bar{z} + w)}{w} \frac{\partial}{\partial \bar{z}} +  \frac{w(\bar{z} + 5w)}{\bar{z}} \frac{\partial}{\partial w} - \frac{\bar{w}(z + 5\bar{w})}{z} \frac{\partial}{\partial \bar{w}} \right) \]
on $\Lambda_0 \setminus \{(0,0)\}$. Since we want to use the parametrization of $\Lambda_0$ that we obtained in the previous section, we need to express $X_{q_2}$ in polar coordinates with $z = \rho \exp(i \theta)$ and $w = \eta \exp(i \varphi)$. One has
\[ \frac{\partial}{\partial z} = \frac{\exp(-i\theta)}{2} \left( \frac{\partial}{\partial \rho} - \frac{i}{\rho} \frac{\partial}{\partial \theta} \right), \quad \frac{\partial}{\partial \bar{z}} = \frac{\exp(i\theta)}{2} \left( \frac{\partial}{\partial \rho} + \frac{i}{\rho} \frac{\partial}{\partial \theta} \right), \]
and similarly for $\frac{\partial}{\partial w}, \frac{\partial}{\partial \bar{w}}$. A straightforward computation using these relations yields the following.

\begin{lm}
On $\Lambda_0 \setminus \{(0,0)\}$, we have that
\[X_{q_2} = \frac{1}{3} \left( \frac{2\rho^2 \sin(\theta + \varphi)}{\eta} \frac{\partial}{\partial \rho} -  \left( 1 + \frac{2\rho \cos(\theta + \varphi)}{\eta} \right) \frac{\partial}{\partial \theta} + \frac{5 \eta^2 \sin(\theta + \varphi)}{\rho} \frac{\partial}{\partial \eta} - \left( 1 + \frac{5 \eta \cos(\theta + \varphi)}{\rho} \right) \frac{\partial}{\partial \varphi} \right). \]
\end{lm}

With our choice of parameters, we parametrize $\Lambda_0$ by $S_{\varepsilon}: [0,3] \times \R / 2\pi\Z \to \C^2$, for $\varepsilon = \pm 1$ and set $\Lambda_0^{\varepsilon} = S_{\varepsilon}([0,3] \times \R / 2\pi\Z)$ with $S_{\varepsilon}(\rho,\theta) = \big( \rho \exp(i\theta), \eta(\rho,\theta) \exp\left( i \varphi(\rho,\theta) \right) \big)$ where $\eta$ and $\varphi$ satisfy
\[ \eta(\rho,\theta) = \rho \sqrt{\frac{2}{5+3\rho^2}}, \quad \varphi(\rho,\theta) = \varepsilon \arccos\left( \frac{\rho^2-1}{\sqrt{2(5+3\rho^2)}} \right) - \theta.  \]

\begin{lm}
\label{lm:Xq2_sing}
On $\Lambda_0^{\varepsilon} \setminus \{(0,0)\}$, we have that
\[ X_{q_2} = \frac{1}{3} \left( \varepsilon \rho \sqrt{(\rho^2+1)(9-\rho^2)} \frac{\partial}{\partial \rho} -  \rho^2 \frac{\partial}{\partial \theta}  + \frac{5\varepsilon \rho \sqrt{2(\rho^2+1)(9-\rho^2)}}{(5+3\rho^2)^{3/2}} \frac{\partial}{\partial \eta} - \frac{8 \rho^2}{5+3\rho^2}  \frac{\partial}{\partial \varphi} \right). \]
\end{lm}

\begin{proof}
Using the relation between $\theta$ and $\varphi$, we compute
\[ \cos(\theta + \varphi) = \frac{\rho^2-1}{\sqrt{2(5+3\rho^2)}}, \quad \sin^2(\theta + \varphi)  = 1 - \frac{(\rho^2-1)^2}{2(5+3\rho^2)} = \frac{(\rho^2+1)(9-\rho^2)}{2(5+3\rho^2)} \]
and we obtain, since $\theta + \varphi$ belongs to $[0,\pi]$ if $\varepsilon = 1$ and to $[-\pi,0]$ if $\varepsilon = -1$, that
\[ \sin(\theta + \varphi) = \varepsilon \sqrt{\frac{(\rho^2+1)(9-\rho^2)}{2(5+3\rho^2)}}. \]
We get the result by substituting these formulas in the expression obtained in the previous lemma.
\end{proof}

We define the loop $\gamma_0:[0,1] \to \Lambda_0$ as: 
\[ \gamma_0(t) = \begin{cases} \tilde{\gamma}_{-1}(6t) \qquad \mathrm{if} \ 0 \leq t \leq \frac{1}{2}, \\ \tilde{\gamma}_{1}(6(1-t)) \qquad \mathrm{if} \ \frac{1}{2} \leq t \leq 1. \end{cases} \]
where $\tilde{\gamma}_{\varepsilon}: I \to \Lambda_0^{\varepsilon}$ is given by $\tilde{\gamma}_{\varepsilon}(\rho) = S_{\varepsilon}\left( \rho, \theta_{\varepsilon}(\rho) \right) = \left(  \rho \exp(i\theta_{\varepsilon}(\rho)), \eta(\rho) \exp\left( i \varphi_{\varepsilon}(\rho) \right) \right)$ with
\begin{equation} \eta(\rho) = \rho \sqrt{\frac{2}{5+3\rho^2}}, \quad \theta_{\varepsilon}(\rho) = \frac{\varepsilon}{2} \arcsin\left(\frac{4-\rho^2}{5} \right) + \frac{\varepsilon \pi}{4}, \quad \varphi_{\varepsilon}(\rho) = \varepsilon \arccos\left( \frac{\rho^2-1}{\sqrt{2(5+3\rho^2)}} \right) - \theta_{\varepsilon}(\rho). \label{eq:etaphi}\end{equation} This means that $\gamma_0$ starts at $(0,0)$, goes through $\Lambda_0^{-1}$, then through $\Lambda_0^{1}$, and ends at $(0,0)$. Observe that this loop is well-defined since $\theta_{\varepsilon}(3) = 0$ and $S_{-1}(3,0) = S_1(3,0)$.

\begin{prop}
\label{prop:radial}
The curve $\gamma_0$ is an integral curve of $X_{q_2}$.
\end{prop}

\begin{proof}
It suffices to prove that the vector field $(\tilde{\gamma}_{\varepsilon})_* \frac{\partial}{\partial \rho}$, which by construction is tangent to the image of $\gamma_0$, is colinear to $X_{q_2}$ at every point of $\gamma_0$. A straightforward computation using the relation between $\theta_{\varepsilon}$ and $\varphi_{\varepsilon}$ yields
\[ \eta'(\rho) =  \frac{5\sqrt{2}}{(5+3\rho^2)^{3/2}}, \quad \theta_{\varepsilon}'(\rho) = \frac{-\varepsilon \rho}{\sqrt{(\rho^2 + 1)(9 - \rho^2)}}, \quad \varphi_{\varepsilon}'(\rho) = \frac{-8 \varepsilon \rho}{(5 + 3\rho^2) \sqrt{(\rho^2 + 1)(9 - \rho^2)}},\]
which means that at the point $\gamma_{\varepsilon}(\rho)$,
\[ (\tilde{\gamma}_{\varepsilon})_* \frac{\partial}{\partial \rho}  = \frac{\partial}{\partial \rho} - \frac{\varepsilon \rho}{\sqrt{(\rho^2 + 1)(9 - \rho^2)}} \frac{\partial}{\partial \theta} +  \frac{5\sqrt{2}}{(5+3\rho^2)^{3/2}} \frac{\partial}{\partial \eta} - \frac{8 \varepsilon \rho}{(5 + 3\rho^2) \sqrt{(\rho^2 + 1)(9 - \rho^2)}} \frac{\partial}{\partial \varphi}. \]
Comparing this with the expression for $X_{q_2}$ displayed in the previous lemma, we obtain that 
\begin{equation} X_{q_2} =\frac{\varepsilon}{3} \rho \sqrt{(\rho^2+1)(9-\rho^2)} \ (\tilde{\gamma}_{\varepsilon})_* \frac{\partial}{\partial \rho}. \label{eq:colinear} \end{equation}
\end{proof}

One final step before computing the invariants $a_1$, $a_2$ is to express the linear Eliasson coordinates of points on the image of $\gamma_0$ in polar coordinates.

\begin{lm}
\label{lm:polar_eliasson}
Let $\varepsilon = \pm 1$ and let $m = (\rho \exp(i\theta_{\varepsilon}),\eta \exp(i\varphi_{\varepsilon})) \in \Lambda_0^{\varepsilon} \cap \gamma_0$ be a point close to $m_0$, with $\eta$, $\theta_{\varepsilon}$ and $\varphi_{\varepsilon}$ as in Equation (\ref{eq:etaphi}). Then the linear Eliasson coordinates $(u_1, u_2, \xi_1, \xi_2)$ of $m$ satisfy
\[ u_1 + i u_2 = \frac{2\rho \left( (1 + 3i) \exp(i\theta_{\varepsilon}) + g(\rho) \exp(-i\varphi_{\varepsilon})  \right)}{1+\rho^2} , \quad \xi_1 + i \xi_2 = \frac{\rho\left( 10 \exp(i\theta_{\varepsilon})  +  (1 + 3i) g(\rho) \exp(-i\varphi_{\varepsilon}) \right)}{30(1+\rho^2)}   \]
where $g(\rho) = \sqrt{2(5+3\rho^2)}$. In particular, 
\[ |u_1 + i u_2 |^2 = \frac{8 \rho^2  \left( 9 + 4 \rho^2 - 3 \varepsilon \sqrt{(\rho^2 + 1)(9 - \rho^2)} \right)}{(1 + \rho^2)^2}, \quad |\xi_1 + i \xi_2|^2 =  \frac{\rho^2 \left( 9 + 4 \rho^2 + 3 \varepsilon \sqrt{(\rho^2 + 1)(9 - \rho^2)} \right)}{45(1 + \rho^2)^2}.  \]
\end{lm}

\begin{proof}
The first part follows from Proposition \ref{prop:linear_eliasson} and Equation (\ref{eq:invstereo}) giving $(x_j, y_j, z_j)$ in terms of $z, w$. The second part follows from the expressions of $\cos(\theta_{\varepsilon} + \varphi_{\varepsilon})$ and $\sin(\theta_{\varepsilon} + \varphi_{\varepsilon})$ in terms of $\rho$, see the proof of Lemma \ref{lm:Xq2_sing}.
\end{proof}

\paragraph{Computation of $a_1$.} We begin by computing $a_1$. In order to do so, we introduce two points $A = (\rho \exp(i\theta_{-1}),\eta \exp(i\theta_{-1})) \in \Lambda_0^{-1} \cap \gamma_0$ and $B = (\rho \exp(i\theta_{1}),\eta \exp(i\varphi_{1})) \in \Lambda_0^{1} \cap \gamma_0$ for $\rho > 0$ small enough (using the notation $\theta_{\pm 1}$ from Equation (\ref{eq:etaphi})), and write their linear Eliasson coordinates as
\[ u_{1,A} + i u_{2,A} =  r_A(\rho) \exp(i\nu_A(\rho)), \quad \xi_{1,B} + i \xi_{2,B} = \sigma_B(\rho) \exp(i \mu_B(\rho)). \]
Since $\kappa_{1,0}(X_{q_2}) = 0$, Equation (\ref{eq:radial_aj}) yields $a_1 = \lim_{\rho \to 0} (\mu_B(\rho) - \nu_A(\rho))$.

\begin{prop}
\label{prop:a1}
For our choice of parameters $t = 1/2, R_1 = 1, R_2 = 5/2$, the term $a_1$ in the Taylor series invariant satisfies $a_1 = \arctan(\frac{9}{13})$. 
\end{prop}

\begin{proof}

It follows from Lemma \ref{lm:polar_eliasson} that $ u_{1,A} + i u_{2,A} = C_A(\rho) z(\rho)$ and $\xi_{1,B} + i \xi_{2,B} = C_B(\rho) w(\rho) $
with $C_A(\rho), C_B(\rho) > 0$ for small enough $\rho > 0$ and, since $\theta_{-1} = -\theta_1$ and $\varphi_{-1} = - \varphi_1$, 
\[ z(\rho) = (1 + 3i) \exp(-i\theta_1(\rho)) + g(\rho) \exp(i\varphi_1(\rho)) , \quad w(\rho) =  10 \exp(i\theta_1(\rho))  +  (1 + 3i) g(\rho) \exp(-i\varphi_1(\rho)). \]
Using the relations $g(0) = \sqrt{10}$, $\sqrt{10} \exp(i\arccos(-\frac{1}{\sqrt{10}})) = - 1 + 3i$ and (\ref{eq:etaphi}), one readily checks that
\[ z(0) = 6 i \exp(-i\theta_1(0)), \quad w(0) = 6(3-i) \exp(i\theta_1(0)). \]
Consequently, $\nu_A(0) = \frac{\pi}{2} - \theta_1(0)$, $\mu_B(0) = -\arctan(\frac{1}{3}) + \theta_1(0)$ and
\[ a_1 = \mu_B(0) - \nu_A(0) = 2 \theta_1(0) - \frac{\pi}{2} -\arctan\left(\frac{1}{3}\right) \]
Since $2 \theta_1(0) = \arcsin(\frac{4}{5}) + \frac{\pi}{2} = \arctan(\frac{4}{3}) + \frac{\pi}{2}$, the $\arctan$ addition formula yields 
$a_1 = \arctan(\frac{9}{13})$.
\end{proof}

\paragraph{Computation of $a_2$.} In order to compute $a_2$, we start by computing the integral term in Equation (\ref{eq:radial_aj}); since $\kappa_{2,0}(X_{q_2}) = -1$, Equation (\ref{eq:colinear}) implies that
\[ \int_{A}^{B} \kappa_{2,0} = \int_{\rho}^{3} \frac{3 \ d\tau}{\tau \sqrt{(\tau^2+1)(9-\tau^2)}} - \int_3^{\rho} \frac{3 \ d\tau}{\tau \sqrt{(\tau^2+1)(9-\tau^2)}},  \]

\begin{lm}
We have that
\begin{equation} \int_{A}^{B} \kappa_{2,0} = - \ln 5 - 2 \ln \rho + \ln\left( 4 \rho_1^2 + 9 + 3 \sqrt{(\rho^2+1)(9-\rho^2)} \right). \label{eq:int_kappa}\end{equation}
\end{lm}

\begin{proof}
It suffices to prove that 
\[ I = \int_{\rho_1}^3 \frac{d\rho}{\rho\sqrt{(\rho^2+1)(9-\rho^2)}} = \frac{1}{6} \left( \ln\left( 4 \rho_1^2 + 9 + 3 \sqrt{(\rho_1^2+1)(9-\rho_1^2)} \right) - 2 \ln \rho_1 - \ln 5 \right). \]
 The successive changes of variables $u = \rho^2$ and $v=u-4$ yield
\[ I = \int_{\rho_1^2}^9 \frac{du}{2u \sqrt{-u^2+8u+9}} =  \int_{\rho_1^2}^9 \frac{du}{2u \sqrt{25-(u-4)^2}} = \int_{\rho_1^2-4}^5 \frac{dv}{2(v+4) \sqrt{25-v^2}}. \]
Using the change of variables $v = 5 \sin \theta$, we can rewrite this integral as
\[ I = \int_{\arcsin\left( \frac{\rho_1^2 - 4}{5} \right)}^{\frac{\pi}{2}} \frac{d\theta}{2(4 + 5\sin \theta)}. \]
This integral can be computed using the usual change of variables $t = \tan(\theta/2)$:
\[ I = \int_{f(\rho_1)}^1 \frac{dt}{4t^2 + 10t + 4} = \int_{f(\rho_1)}^1 \frac{dt}{3(2t+1)} - \int_{f(\rho_1)}^1 \frac{dt}{6(t+2)} = \frac{1}{6} \ln \left( \frac{f(\rho_1)+2}{2f(\rho_1)+1} \right), \]
where $f(\rho_1)$ reads, using the identity $x\tan(\arcsin(x)/2) = 1 - \sqrt{1-x^2}$,
\[ f(\rho_1) = \tan\left( \frac{\arcsin\left( \frac{\rho_1^2 - 4}{5} \right)}{2} \right) = \frac{5-\sqrt{(\rho_1^2+1)(9-\rho_1^2)}}{\rho_1^2-4}. \]
Consequently, we obtain that
\[ \frac{f(\rho_1)+2}{2f(\rho_1)+1} = \frac{-3 + 2\rho_1^2 - \sqrt{(\rho_1^2+1)(9-\rho_1^2)}}{6 + \rho_1^2 - 2 \sqrt{(\rho_1^2+1)(9-\rho_1^2)} } = \frac{4 \rho_1^2 + 9 + 3 \sqrt{(\rho_1^2+1)(9-\rho_1^2)}}{5\rho_1^2}.\]
\end{proof}

It remains to compute the logarithmic term in (\ref{eq:radial_aj}). 

\begin{prop}
\label{prop:a2}
For our choice of parameters $t = 1/2, R_1 = 1, R_2 = 5/2$, the term $a_2$ in the Taylor series invariant satisfies
\[ a_2 = \frac{7}{2} \ln 2 + 3 \ln 3 - \frac{3}{2} \ln 5. \]
\end{prop}

\begin{proof}
It follows from Lemma \ref{lm:polar_eliasson} that
\[ r_A^2 = \frac{8 \rho^2}{(1 + \rho^2)^2} \left( 9 + 4 \rho^2 + 3 \sqrt{(\rho^2 + 1)(9 - \rho^2)} \right), \quad \sigma_B^2 =  \frac{\rho^2}{45(1 + \rho^2)^2} \left( 9 + 4 \rho^2 + 3 \sqrt{(\rho^2 + 1)(9 - \rho^2)} \right). \]
Consequently, we obtain that
\[ \ln(r_A \sigma_B) = \frac{3}{2} \ln 2 - \ln 3 - \frac{1}{2} \ln 5 + 2 \ln \rho - 2 \ln(1 + \rho^2) + \ln\left(9 + 4 \rho^2 + 3 \sqrt{(\rho^2 + 1)(9 - \rho^2)}\right); \]
this equation together with Equation (\ref{eq:int_kappa}) gives
\[  \int_A^B \kappa_{2,0} + \ln(r_{A}\sigma_{B}) = \frac{3}{2} \ln 2 - \ln 3 - \frac{3}{2} \ln 5 + 2 \ln\left(9 + 4 \rho^2 + 3 \sqrt{(\rho^2 + 1)(9 - \rho^2)}\right) - 2 \ln(1 + \rho^2) \]
which implies, by taking the limit when $\rho$ goes to zero, that $a_2 = \frac{7}{2} \ln 2 + 3 \ln 3 - \frac{3}{2} \ln 5$, as announced.
\end{proof}

\subsection{The polygonal invariant}

Next we compute the polygonal invariant of the system (see Section \ref{sect:symp_invariant} for notation).

\begin{prop}
\label{prop:polygon}
For $t = 1/2$, the polygonal invariant is the $(\mathcal{G} \times \mathcal{T})$-orbit consisting of the two following rational convex polygons:
\begin{enumerate}
\item the parallelogram $\Delta_1$ with vertices at $(-(R_1 + R_2),0), (R_1-R_2,2R_1), (R_1 + R_2, 2R_1)$ and $(R_2 - R_1,0)$, with $\epsilon = 1$,
\item the trapezoid $\Delta_2$ with vertices at $(-(R_1 + R_2),0), (R_1 + R_2,0), (R_2 - R_1,-2R_1)$ and $(R_1-R_2,-2R_1)$, with $\epsilon = -1$.
\end{enumerate}
These polygons are depicted in Figure \ref{fig:polygons}.
\end{prop}

\begin{proof}
We use Theorem $5.3$ in \cite{VNpoly}, which allows us to compute the difference between the slopes of the top and bottom edges at a vertex $x$ of such a convex polygon in terms of the number of focus-focus critical points and the isotropy weights of $J$ at the elliptic-elliptic critical points. These isotropy weights are defined as follows; near a critical point $m$ of elliptic-elliptic type, there exists local symplectic coordinates $(q_1,q_2,p_1,p_2)$ in which $J$ can be written as
\[ J = J(m) + a \left( \frac{q_1^2+p_1^2}{2} \right) + b \left( \frac{q_2^2+p_2^2}{2} \right) + \mathcal{O}(3). \]
The numbers $a,b$ are the isotropy weights of $J$ at $m$: let us compute them in our case. One readily checks that
\[ J = -(R_1 + R_2) + R_1 \left( \frac{x_1^2 + y_1^2}{2} \right) + R_2 \left( \frac{x_2^2 + y_2^2}{2} \right) + O(3) \]
near $m_1$; and since $\omega = R_1 dx_1 \wedge dy_1 + R_2 dx_2 \wedge dy_2$ at $m_1$, the coordinates $q_i = \sqrt{R_i} y_i$, $p_i = \sqrt{R_i} x_i$ are symplectic. In these coordinates,
\[ J = -(R_1 + R_2) + \frac{q_1^2 + p_1^2}{2}  + \frac{q_2^2 + p_2^2}{2}+ O(3), \]
hence the isotropy weights are $(1,1)$. Similar computations show that the isotropy weights of $J$ at $m_2$ and $m_3$ are $(1,-1)$ and $(-1,-1)$ respectively.

We can now compute the polygonal invariant. Let us first choose $\epsilon = 1$, which means that we construct the polygon by introducing a corner above $c_0$. We may assume that $c_1$ is fixed, so the first vertex of the polygon is $A = (-(R_1+R_2),0)$. We may also assume that the bottom edge $e_1$ leaving $A$ is an horizontal segment going rightwards. By the discussion above, we know that the difference between the slope of the other edge $e_2$ leaving $A$ and the slope of $e_1$ (the latter being zero) is equal to one. Therefore the slope of $e_2$ is equal to one; the other vertex $B$ of $e_2$ is the intersection of $e_2$ with the vertical line defined by the abscissa of $c_0$. This means that $B = (R_1-R_2,2R_1)$. The difference between the slope of the other edge $e_3$ leaving $B$ and the slope of $e_2$ is equal to $-1$, hence $e_3$ is horizontal. Its other vertex $C$ has the same abscissa as $m_3$, so we must have $C = (R_1+R_2,2R_1)$. The difference between the slope of $e_3$ and the slope of the other edge $e_4$ leaving $C$ is equal to $-1$, thus $e_4$ has slope $1$. Finally, $e_4$ meets $e_1$ at $(R_2-R_1,0)$, and we obtain the parallelogram described above.

The case $\epsilon = -1$ follows the same lines and is left to the reader.
\end{proof}

\begin{figure}[H]
\begin{center}
\subfigure[$\Delta_1$, $\epsilon=1$.]{\includegraphics[scale=0.4]{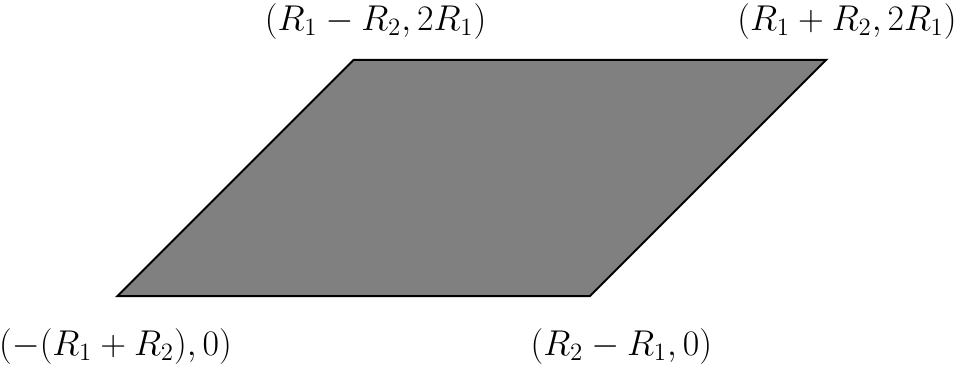} }
\hspace{1cm}
\subfigure[$\Delta_2$, $\epsilon=-1$.]{\includegraphics[scale=0.4]{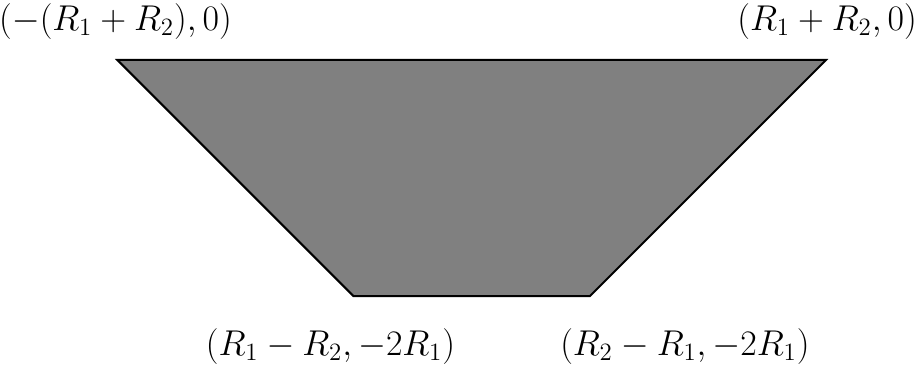} }
\end{center}
\caption{The polygonal invariants.}
\label{fig:polygons}
\end{figure}

We will not compute the twisting indices of these two polygons, but let us say a few words about them. Let $\ell_0$ be the vertical line passing through $c_0$. Recall that $T = \begin{pmatrix} 1 & 0 \\ 1 & 1 \end{pmatrix}$ and that for $n \in \Z$, $t_{\ell_0}^n$ is the piecewise integral-affine transformation equal to the identity on the left half-space defined by $\ell_0$ and to $T^n$ on the right half-space defined by $\ell_0$. We claim that $\Delta_2 = (t_{\ell_0}^1 \circ T^{-1})(\Delta_1)$. Therefore, the description of the behavior of the twisting index under the $(\mathcal{G} \times \mathcal{T})$-action in Section \ref{sect:symp_invariant} yields the following result.
\begin{lm}
\label{lm:twist}
The twisting indices $k_1, k_2$ of $\Delta_1, \Delta_2$ satisfy the equality $k_2 = k_1 - 1$.
\end{lm}

\subsection{The height invariant}
\label{sect:height}

Recall that the height invariant $h$ is the volume of $\{H < H(m_0)\} = \{H < 0\}$ in the reduced manifold $M^{\text{red}} = J^{-1}(J(m_0))/S^1$, with respect to $|\omega^{\text{red}}|/2\pi$, with $\omega^{\text{red}}$ the canonical symplectic form on $M^{\text{red}}$. The $S^1$-action is given by the Hamiltonian flow of $J$; in other words, in polar coordinates,
\[ \exp(is) \cdot (\rho \exp(i\theta),\eta \exp(i\varphi)) = (\rho \exp(i(\theta-s)),\eta \exp(i(\varphi+s))). \]
Moreover, on $J^{-1}(J(m_0))$, $\eta = \rho \sqrt{\frac{R_1}{R_2+(R_2-R_1)\rho^2}}$. Thus an element of $M^{\text{red}}$ can be described by two coordinates $(\rho,\alpha)$, as the equivalence class $$\left[\left(\rho, \rho \sqrt{\frac{R_1}{R_2+(R_2-R_1)\rho^2}} \exp(i\alpha)\right)\right].$$ We would like to express $\omega^{\text{red}}$ in these coordinates; by definition, for $[X],[Y] \in TM^{\text{red}}$, we have $\omega^{\text{red}}([X],[Y]) = \omega(X,Y)$ for any choice of representatives $X, Y \in T(J^{-1}(c_0))$. But
\[ X = \mu_1^X \frac{\partial}{\partial \rho} + \frac{\mu_2^X}{2} \frac{\partial}{\partial \theta} + \frac{R_2 \sqrt{R_1}}{(R_2+(R_2-R_1)\rho^2)^{3/2}} \frac{\partial}{\partial \eta} + \frac{\mu_2^X}{2} \frac{\partial}{\partial \varphi} \]
is a representative of $[X] = \mu_1^X \frac{\partial}{\partial \rho} + \mu_2^X \frac{\partial}{\partial \alpha}$. Therefore,
\[ \omega^{\text{red}}([X],[Y]) = \left( \frac{2R_1\rho}{(1+\rho^2)^2} + \frac{2R_2^2 \sqrt{R_1} \eta}{(1+\eta^2)^2 (R_2+(R_2-R_1)\rho^2)^{3/2}} \right) ( \mu_1^X \mu_2^Y - \mu_2^X \mu_1^Y ). \]
Writing $\eta$ in terms of $\rho$, a straightforward computation yields
\[ \omega^{\text{red}}([X],[Y]) = \frac{4R_1\rho}{(1+\rho^2)^2} ( \mu_1^X \mu_2^Y - \mu_2^X \mu_1^Y ), \quad i.e. \quad \omega^{\text{red}} = \frac{4R_1\rho}{(1+\rho^2)^2}  d\rho \wedge d\alpha.  \]
This is consistent with the fact that the volume of $M^{\text{red}}$ with respect to $|\omega^{\text{red}}|/2\pi$ is equal to the length of the vertical segment containing the image of $c_0$ in any polygon constructed in the previous section, that is $2R_1$. Indeed,
\[ \int_0^{2\pi} \int_0^{+\infty} \frac{4R_1\rho}{(1+\rho^2)^2}  \frac{d\rho \wedge d\alpha}{2\pi} = \int_0^{+\infty} \frac{4R_1\rho}{(1+\rho^2)^2} d\rho = 2R_1. \]
Using previous considerations, we get that the height invariant is the volume of the submanifold consisting of all elements $\left[\left(\rho, \rho \sqrt{\frac{R_1}{R_2+(R_2-R_1)\rho^2}} \exp(i\alpha)\right)\right] \in M^{\text{red}}$ satisfying
\[ \sqrt{\frac{R_1}{R_2+(R_2-R_1)\rho^2}}(1-\rho^2) + 2 \cos \alpha < 0  . \]
This inequality is possible if and only if $0 < \rho < \sqrt{\frac{4R_2}{R_1}-1}$, and it means that
\[ \arccos\left( \frac{\rho^2-1}{2} \sqrt{\frac{R_1}{R_2+(R_2-R_1)\rho^2}} \right) < \alpha < 2\pi - \arccos\left( \frac{\rho^2-1}{2} \sqrt{\frac{R_1}{R_2+(R_2-R_1)\rho^2}} \right). \]
Consequently, we have that
\[ h = \int_0^{\sqrt{4\Theta-1}} \frac{4R_1\rho}{2\pi(1+\rho^2)^2} \left( 2\pi-2\arccos\left( \frac{\rho^2-1}{2} \sqrt{\frac{R_1}{R_2+(R_2-R_1)\rho^2}} \right) \right) = 2R_1\left(1-\frac{2I}{\pi}\right), \]
where $\Theta = R_2/R_1$ and $I$ is the integral
\[ I = \int_0^{\sqrt{4\Theta-1}} \frac{\rho}{(1+\rho^2)^2} \arccos\left( \frac{\rho^2-1}{2} \sqrt{\frac{R_1}{R_2+(R_2-R_1)\rho^2}}\right) d\rho. \]
One can compute the latter as follows.

\begin{lm}
\label{lm:compute_int}
The integral $I$ satisfies
\[ I = \frac{\pi}{2} - \frac{1}{2} \arccos\left(\frac{1}{2\sqrt{\Theta}}\right)  - \frac{\sqrt{4\Theta-1}}{4}  + \left(\frac{\Theta-1}{2} \right) \arctan\left( \frac{4 \Theta - 1 - \sqrt{4\Theta-1}}{(2 \Theta - 1)(\sqrt{4\Theta - 1} - 1)} \right). \]
\end{lm}

This can be checked using any computer algebra software, but we give a proof in Appendix A. Using this lemma and the considerations before it, we finally obtain the following formula for the height invariant.

\begin{prop}
\label{prop:height}
As above, let $\Theta = R_2/R_1$. The height invariant $h$ is given by the formula
\[ h = \frac{R_1}{\pi} \left(2 \arccos\left(\frac{1}{2\sqrt{\Theta}}\right) + \sqrt{4\Theta - 1} - 2(\Theta-1) \arctan\left( \frac{4 \Theta - 1 - \sqrt{4\Theta-1}}{(2 \Theta - 1)(\sqrt{4\Theta - 1} - 1)} \right) \right). \]
\end{prop}

Coming back to our example where $R_1 = 1, R_2 = 5/2$, this formula yields
\[ h = \frac{1}{\pi} \left( 2 \arccos\left(\frac{1}{\sqrt{10}}\right) + 3 - 3 \arctan\left(\frac{3}{4}\right) \right) \approx 1.136.  \]

\section{Quantization of $F$}
\label{sect:quantum}

We now want to study a quantum version of this integrable system. Namely, we want to find two commuting self-adjoint operators $\hat{J}, \hat{H}$ acting on some Hilbert space, and quantizing $J,H$ in some sense. To be more precise, we will be working with a semiclassical parameter, and we will ask the principal symbols of $\hat{J}, \hat{H}$ to be $J,H$. Note that such operators might not exist in general, see \cite{GvS}, but we will see that in our example such a pair can be constructed. Since the phase space for coupled angular momenta is compact, the quantization procedure relies on the so-called geometric quantization \cite{Kos,Sou}, and the quantum observables are Berezin-Toeplitz operators \cite{BouGui,Cha1,Mama,Schli}. We start by briefly reviewing these notions for the sake of completeness.

\subsection{Geometric quantization and Berezin-Toeplitz operators}

Let $(M,\omega,j)$ be a compact, connected, K\"ahler manifold. Assume that $M$ is quantizable, i.e. that there exists a holomorphic, Hermitian line bundle $L \to M$ whose Chern connection $\nabla$ has curvature $-i\omega$ (a prequantum line bundle); this is equivalent to asking the cohomology class $[\omega/2\pi]$ to be integral. Let $K \to M$ be another holomorphic, Hermitian line bundle; typically, we would like to take $K = \delta$ a half-form bundle, that is a square root of the canonical line bundle $\Lambda^{n,0} T^*M$, but such a line bundle might not exist globally. For any positive integer $k$, we consider the Hilbert space
\[ \Hil_k = H^0(M,L^{\otimes k} \otimes K), \quad \scal{\phi}{\psi} = \int_M h_k(\phi,\psi) \mu \]
of holomorphic sections of the line bundle $L^{\otimes k} \otimes K \to M$. Here $h_k$ is the Hermitian metric induced on $L^{\otimes k} \otimes K$ by the ones on $L$ and $K$, and $\mu$ is the Liouville measure associated with $\omega$. Since $M$ is compact, this space is finite dimensional. The integer $k$ is a semiclassical parameter, representing the inverse of Planck's constant $\hbar$, and the semiclassical limit is $k \to +\infty$.

The quantum observables, Berezin-Toeplitz operators, are sequences of operators defined as follows. Let $L^2(M, L^{\otimes k} \otimes K)$ be the completion of the space $\classe{\infty}(L^{\otimes k} \otimes K)$ of smooth sections of $L^{\otimes k} \otimes K \to M$ with respect to the scalar product $\scal{\cdot}{\cdot}$, and let $\Pi_k$ be the orthogonal projector from $L^2(M, L^{\otimes k} \otimes K)$ to the quantum space $\Hil_k$.

\begin{dfn}
A \strong{Berezin-Toeplitz operator} is a sequence of operators $(T_k:\Hil_k \to \Hil_k)_{k \geq 1}$ of the form $T_k = \Pi_k M_{f(\cdot,k)} + R_k$ for some sequence of smooth functions $f(\cdot,k)$ with an asymptotic expansion of the form $f(\cdot,k) = \sum_{\ell \geq 0} k^{-\ell} f_{\ell}$ in the $\classe{\infty}{}$-topology, and some sequence of operators $R_k$ whose operator norm $\| R_k \|$ is a $O(k^{-\infty})$, that is a $O(k^{-N})$ for every $N \geq 1$.
\end{dfn}

Here $M_f$ stands for the operator of multiplication by $f$. The first term $f_0$ in the above asymptotic expansion is the \strong{principal symbol} of $T_k$.

\subsection{Quantization of the sphere}

In order to use this recipe to quantize the sphere $\S^2$, we start by working on $\C\P^1$, and consider the tautological line bundle $\bigO{-1} = \left\{ ([u],v) \in \C\P^1 \times \C^{2} \ | \ v \in \C u  \right\} \subset \C\P^1 \times \C^{2}, $
endowed with its natural holomorphic and Hermitian structures. One can check that the associated Chern connection has curvature $i \omega_{FS}$, where $\omega_{FS}$ is the Fubini-Study form on $\C\P^1$ (normalized so that the area of $\C\P^1$ is equal to $2\pi$). Therefore, the dual line bundle $L = \bigO{1}$ is a prequantum line bundle for $(\C\P^1,\omega_{FS})$. Moreover, it is well-known that the canonical bundle can be identified with $\bigO{-2}$; thus $\delta = \bigO{-1}$ is a half-form bundle. So the Hilbert spaces that we consider are $H^0(\C\P^1,L^{\otimes k} \otimes \delta) = H^0(\C\P^1,\bigO{k-1})$. Now, given an integer $p \geq 1$, it is standard that $H^0(\C\P^1,\bigO{p})$ can be identified with the space $\C_{p}[z_1,z_2]$ of homogeneous polynomials of degree $p$ in two complex variables. In this isomorphism, the scalar product becomes
\[ \scal{P}{Q} = \int_{\C} \frac{P(1,z) \overline{Q(1,z)}}{(1+|z|^2)^{p+2}} \ | dz \wedge d\bar{z} |. \]

To come back to $\S^2$, we use the stereographic projection (for instance from the north pole to the equatorial plane); one can check that the pullback of $\omega_{\text{FS}}$ by the latter is $-\tfrac{1}{2}\omega_{\S^2}$ (actually, we already used this result in a previous section). There exist very explicit formulas for the quantization of the coordinates $(x_0,y_0,z_0)$ on $\S^2$. In what follows, we identify $\C_p[z_1,z_2]$ with the space of polynomials of one complex variable $z$ of degree less than or equal to $p$; the following result holds in this identification.

\begin{lm}
\label{lm:coord_BTO}
The self-adjoint operators
\[ T_p(x_0) = \frac{1}{p+2} \left( (1-z^2) \frac{d}{dz} + pz \right), \ T_p(y_0) = \frac{i}{p+2} \left( (1+z^2) \frac{d}{dz} - pz \right), \ T_p(z_0) = \frac{1}{p+2} \left( 2z \frac{d}{dz} - p \ \mathrm{Id} \right)  \]
are Berezin-Toeplitz operators acting on $H^0(\C\P^1,\mathcal{O}(p))$ with respective principal symbols $x_0$, $y_0$, $z_0$.
\end{lm}

For a proof, see for instance \cite[Lemma $3.4$]{BloGol}. One readily checks that these operators satisfy the commutation relation
\begin{equation} [T_p(x_0),T_p(y_0)] =  \frac{2i}{p+2} T_p(z_0)   \label{eq:commutation_BTO} \end{equation}
and its cyclic permutations. Now, one can check that the family
\[ \phi_{\ell} = \sqrt{\frac{(p+1) \binom{p}{\ell}}{2\pi}} z^{p-\ell}, \quad 0 \leq \ell \leq p \]
is an orthonormal basis of $H^0(\C\P^1,\mathcal{O}(p))$. A direct computation shows the following.

\begin{lm}
\label{lm:action_BTO}
The action of $T_p(x_0), T_p(y_0)$ and $T_p(z_0)$ in the basis $(\phi_{\ell})_{0 \leq \ell \leq p}$ is given by the formulas
\begin{enumerate}
\item $T_p(x_0) \phi_{\ell} = \frac{1}{p+2} \left( \sqrt{\ell(p-\ell+1)} \ \phi_{\ell-1} + \sqrt{(\ell+1)(p-\ell)} \ \phi_{\ell+1}  \right)$,
\item $T_p(y_0) \phi_{\ell} = \frac{-i}{p+2} \left( \sqrt{\ell(p-\ell+1)} \ \phi_{\ell-1} - \sqrt{(\ell+1)(p-\ell)} \ \phi_{\ell+1}  \right)$,
\item $T_p(z_0) \phi_{\ell} = \left(\frac{p-2\ell}{p+2}\right) \phi_{\ell}$.
\end{enumerate}
Here we have used the convention $\phi_{-1} = 0 = \phi_{p+1}$.
\end{lm}

\subsection{Quantum coupled angular momenta}

Now we want to quantize $M = \S^2 \times \S^2$ with symplectic form $-(R_1 \omega_{\S^2} \otimes R_2 \omega_{\S^2})$. This is possible if and only if $R_1$ and $R_2$ are positive half-integers. In this case, the external tensor product
\[ L = \bigO{2R_1} \boxtimes \bigO{2R_2} = p_1^* \bigO{2R_1} \otimes p_2^* \bigO{2R_2} \]
is a prequantum line bundle over $\C\P^1 \times \C\P^1$, where $p_i: \C\P^1 \times \C\P^1 \to \C\P^1$, $i=1,2$, are the natural projections to the first and second factor. Moreover, the line bundle $\delta = \bigO{-1} \boxtimes \bigO{-1}$ is a half-form bundle over $\C\P^1 \times \C\P^1$, hence the quantum spaces are $ \Hil_k = H^0\left(\C\P^1 \times \C\P^1, \bigO{2kR_1-1} \boxtimes \bigO{2kR_2-1} \right)$
for $k \geq 1$ integer. By a version of the K\"unneth formula for the Dolbeault cohomology \cite{SamWas}, this yields
\[ \Hil_k =  H^0\left(\C\P^1, \bigO{2kR_1-1}\right) \otimes H^0\left(\C\P^1, \bigO{2kR_2-1}\right). \]
Let us now turn to quantum observables. For $i=1,2$, we consider the operators
\[ X_i = \left( 1 + \frac{1}{2kR_i} \right) T_{2kR_i-1}(x_i), \quad Y_i = \left( 1 + \frac{1}{2kR_i} \right) T_{2kR_i-1}(y_i), \quad
 Z_i = \left( 1 + \frac{1}{2kR_i} \right) T_{2kR_i-1}(z_i) \]
acting on $H^0\left(\C\P^1, \bigO{2kR_i-1}\right)$. It follows from Lemma \ref{lm:coord_BTO} that the operators
\begin{equation} \hat{J}_k = R_1 Z_1 + R_2 Z_2, \qquad \hat{H}_k = (1-t) Z_1 \otimes \mathrm{Id} + t \left( X_1 \otimes X_2 + Y_1 \otimes Y_2 + Z_1 \otimes Z_2 \right) \label{eq:quantum_system}\end{equation}
acting on $\Hil_k$ are Berezin-Toeplitz operators with principal symbols $J$ and $H$ respectively. Indeed, multiplying by a scalar of the form $1 + \bigO{k^{-1}}$ does not change the principal symbol.

\begin{lm}
The operators $\hat{J}_k$ and $\hat{H}_k$ commute.
\end{lm}

\begin{proof}
We have that
\[[\hat{J}_k,\hat{H}_k]  =  t ( R_1 [Z_1 \otimes \mathrm{Id}, X_1 \otimes X_2] + R_1 [Z_1 \otimes \mathrm{Id}, Y_1 \otimes Y_2] + R_2 [\mathrm{Id} \otimes Z_2, X_1 \otimes X_2]  +  R_2 [\mathrm{Id} \otimes Z_2, Y_1 \otimes Y_2]). \]
Using the commutation relations (\ref{eq:commutation_BTO}), we obtain that
\[  [Z_1 \otimes \mathrm{Id}, X_1 \otimes X_2] = [Z_1,X_1] \otimes X_2 = \frac{i}{k R_1} Y_1 \otimes X_2. \]
Similarly, we get
\[ [Z_1 \otimes \mathrm{Id}, Y_1 \otimes Y_2] = -\frac{i}{k R_1} X_1 \otimes Y_2, \quad [\mathrm{Id} \otimes Z_2, X_1 \otimes X_2] = \frac{i}{k R_2} X_1 \otimes Y_2, \quad [\mathrm{Id} \otimes Z_2, Y_1 \otimes Y_2] = -\frac{i}{k R_2} Y_1 \otimes X_2.  \]
Using these relations, we obtain that $[\hat{J}_k,\hat{H}_k] = 0$.
\end{proof}

\subsection{Joint spectrum}

We want to compute the joint spectrum of $(\hat{J}_k,\hat{H}_k)$, which is the set of elements $(\lambda_1,\lambda_2) \in \R^2$ such that there exists a common eigenvector $v \neq 0 \in \Hil_k$ such that $\hat{J}_k v = \lambda_1 v$ and $\hat{H}_k v = \lambda_2 v$. In order to do so, we start by finding the eigenvalues of $\hat{J}_k$. We start by endowing $\Hil_k$ with the orthonormal basis
\[ g_{\ell,m} = e_{\ell} \otimes f_m, \qquad 0 \leq \ell \leq 2kR_1-1, \quad 0 \leq m \leq 2kR_2 - 1 \]
where $e_{\ell}$, $f_m$ are defined by
\[ e_{\ell} = \sqrt{\frac{2kR_1 \binom{2kR_1-1}{\ell}}{2\pi}} z^{2kR_1-1-\ell}, \qquad f_{m} = \sqrt{\frac{2kR_2 \binom{2kR_2-1}{m}}{2\pi}} w^{2kR_2-1-m} \]
in the identification of $H^0\left(\C\P^1, \bigO{2kR_1-1}\right)$ with the space of polynomials of degree at most $2kR_1-1$ in the variable $z$ and of $H^0\left(\C\P^1, \bigO{2kR_2-1}\right)$ with the space of polynomials of degree at most $2kR_2 - 1$ in the variable $w$.

\begin{lm}
\label{lm:ev_J}
The eigenvalues of $\hat{J}_k$ are the numbers $\lambda_j^{(k)}$ for $0 \leq j \leq 2(k(R_1 + R_2) - 1)$, where
\[ \lambda_j^{(k)} = R_1 + R_2 - \frac{j+1}{k}. \]
\end{lm}

\begin{proof}
It follows from Lemma \ref{lm:action_BTO} that
\[ Z_1 e_{\ell} = \frac{2kR_1-1-2\ell}{2kR_1} e_{\ell}, \quad Z_2 f_m = \frac{2kR_2-1-2m}{2kR_2} f_{m}. \]
Since $\hat{J}_k g_{\ell,m} = R_1 Z_1 e_{\ell} \otimes f_m + R_2 e_{\ell} \otimes Z_2 f_m$, this yields
\[ \hat{J}_k g_{\ell,m} =  \frac{2k(R_1+R_2) - 2 - 2(\ell + m)}{2k} g_{\ell,m} = \left( R_1 + R_2 - \frac{\ell + m + 1}{k} \right) g_{\ell,m}. \]
\end{proof}
In order to compute the joint spectrum of $\hat{J}_k$ and $\hat{H}_k$, we need to find the eigenvalues of the restriction of $\hat{H}_k$ to a given eigenspace of $\hat{J}_k$. The eigenspace associated with the eigenvalue $\lambda_j^{(k)}$ is
\[ E_{\lambda_j^{(k)}} = \mathrm{Span} \left\{ g_{\ell,m} | \ \ \ell + m = j, \ 0 \leq \ell \leq 2kR_1-1, \ 0 \leq m \leq 2kR_2 - 1 \right\}. \]
In order to compute an orthonormal basis for $E_{\lambda_j^{(k)}}$, we need to separate the three following cases:
\begin{enumerate}
\item if $0 \leq j \leq 2kR_1 - 1$, then $E_{\lambda_j^{(k)}} = \mathrm{Span}\left( g_{0,j}, g_{1,j-1}, \ldots, g_{j,0} \right)$, which has dimension $j+1$,
\item if $2kR_1 \leq j \leq 2kR_2 - 1$, then $E_{\lambda_j^{(k)}} = \mathrm{Span}\left( g_{0,j}, g_{1,j-1}, \ldots, g_{2kR_1-1,j-2kR_1+1} \right)$, which has dimension $2kR_1$,
\item if $2kR_2 \leq j \leq 2(k(R_1+R_2) - 1)$, then $E_{\lambda_j^{(k)}} = \mathrm{Span}\left( g_{j-2kR_2+1,2kR_2-1}, \ldots, g_{2kR_1-1,j-2kR_1+1} \right)$, which has dimension $2k(R_1 + R_2)-(j+1)$.
\end{enumerate}
A direct computation using these formulas shows that the sum of the dimensions of the eigenspaces of $\hat{J}_k$ is indeed equal to $4k^2 R_1R_2 = \dim \Hil_k$. Now that we have this very explicit description, it only remains to understand how $\hat{H}_k$ acts on a basis element $g_{\ell,m}$.

\begin{lm}
Using the convention $e_{-1} = 0 = e_{2kR_1}$ and $f_{-1} = 0 = f_{2kR_2}$, one has
\[ \begin{split} & \hat{H}_k g_{\ell,m} =  \frac{1}{4k^2R_1R_2} \left( 2t \sqrt{\ell(2kR_1-\ell)(m+1)(2kR_2-1-m)} \ g_{\ell-1,m+1} \right. \\
 \\
& + \left. \left(2(kR_1-\ell)-1\right)(2kR_2 - (2m+1)t) g_{\ell,m} + 2t \sqrt{(\ell+1)(2kR_1-1-\ell)m(2kR_2-m)} \ g_{\ell+1,m-1} \right). \end{split} \]
\end{lm}

\begin{rmk}
This formula is consistent with the fact that $\hat{H}_k$ preserves the eigenspaces of $\hat{J}_k$; indeed, if $(\ell,m)$ is such that $\ell + m = j$, the same holds for $(\ell-1,m+1)$ and $(\ell+1,m-1)$.
\end{rmk}

\begin{proof}
Again, it directly follows from Lemma \ref{lm:action_BTO} that
\[ (Z_1 \otimes \mathrm{Id}) g_{\ell,m} = \frac{2(kR_1-\ell)-1}{2kR_1} g_{\ell,m}, \quad (Z_1 \otimes Z_2) g_{\ell,m} = \frac{(2(kR_1-\ell)-1)(2(kR_2-m)-1)}{4k^2 R_1 R_2} g_{\ell,m}. \]
Therefore, we obtain that
\[ \left((1-t)(Z_1 \otimes \mathrm{Id}) + t Z_1 \otimes Z_2 \right)  g_{\ell,m} = \frac{\left(2(kR_1-\ell)-1\right)(2kR_2 - (2m+1)t)}{4k^2 R_1 R_2}  g_{\ell,m}.  \]
It remains to understand how the operator $X_1 \otimes X_2 + Y_1 \otimes Y_2$ acts on $g_{\ell,m}$. Using Lemma \ref{lm:action_BTO}, we obtain that
\[ (X_1 \otimes X_2) g_{\ell,m} = \frac{1}{4k^2 R_1 R_2} \left( a_{\ell,m} g_{\ell-1,m-1} + b_{\ell,m} g_{\ell-1,m+1} + c_{\ell,m} g_{\ell+1,m-1} + d_{\ell,m} g_{\ell+1,m+1} \right) \]
where $a_{\ell,m} = \sqrt{\ell(2kR_1-\ell)m(2kR_2-m)}$, $\ b_{\ell,m} = \sqrt{\ell(2kR_1-\ell)(m+1)(2kR_2-1-m)}$, \vspace{1mm} \\ $c_{\ell,m} = \sqrt{(\ell+1)(2kR_1-1-\ell)m(2kR_2-m)}$ and $d_{\ell,m} = \sqrt{(\ell+1)(2kR_1-1-\ell)(m+1)(2kR_2-1-m)}$. Similarly, we obtain that
\[ (Y_1 \otimes Y_2) g_{\ell,m} = \frac{1}{4k^2 R_1 R_2} \left( -a_{\ell,m} g_{\ell-1,m-1} + b_{\ell,m} g_{\ell-1,m+1} + c_{\ell,m} g_{\ell+1,m-1} - d_{\ell,m} g_{\ell+1,m+1} \right). \]
Consequently, $ (X_1 \otimes X_2 + Y_1 \otimes Y_2) g_{\ell,m} = \frac{1}{4k^2 R_1 R_2} \left( 2b_{\ell,m} g_{\ell-1,m+1} + 2c_{\ell,m} g_{\ell+1,m-1}\right)$.
\end{proof}

Now we have everything that we need in order to compute the matrices of the restrictions $\hat{H}_k^{(j)}$ of $\hat{H}_k$ to the eigenspaces $E_{\lambda_j^{(k)}}$ in the bases introduced above, hence the joint spectrum
\[ \mathrm{JSp}\left(\hat{J}_k,\hat{H}_k\right) = \left\{ \left( \lambda_j^{(k)}, \lambda \right) | \ 0 \leq j \leq 2(k(R_1 + R_2) - 1), \ \lambda \in \mathrm{Sp}\left(\hat{H}_k^{(j)}\right) \right\} . \]
We display this joint spectrum for certain values of the parameters in Figure \ref{fig:spectrum}.

\begin{figure}[H]
\begin{center}
\subfigure[$t=0$]{\includegraphics[scale=0.22]{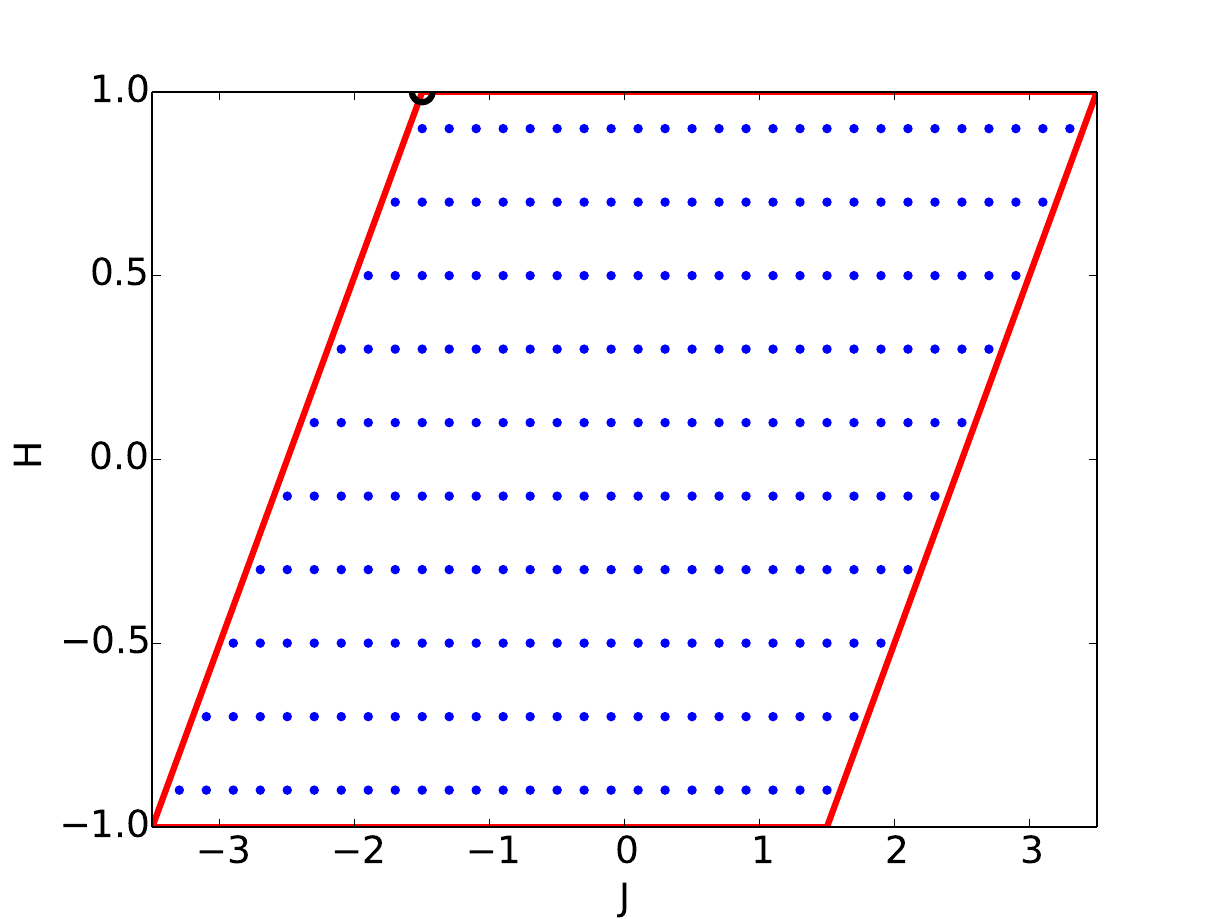} }
\subfigure[$t=0.2$]{\includegraphics[scale=0.22]{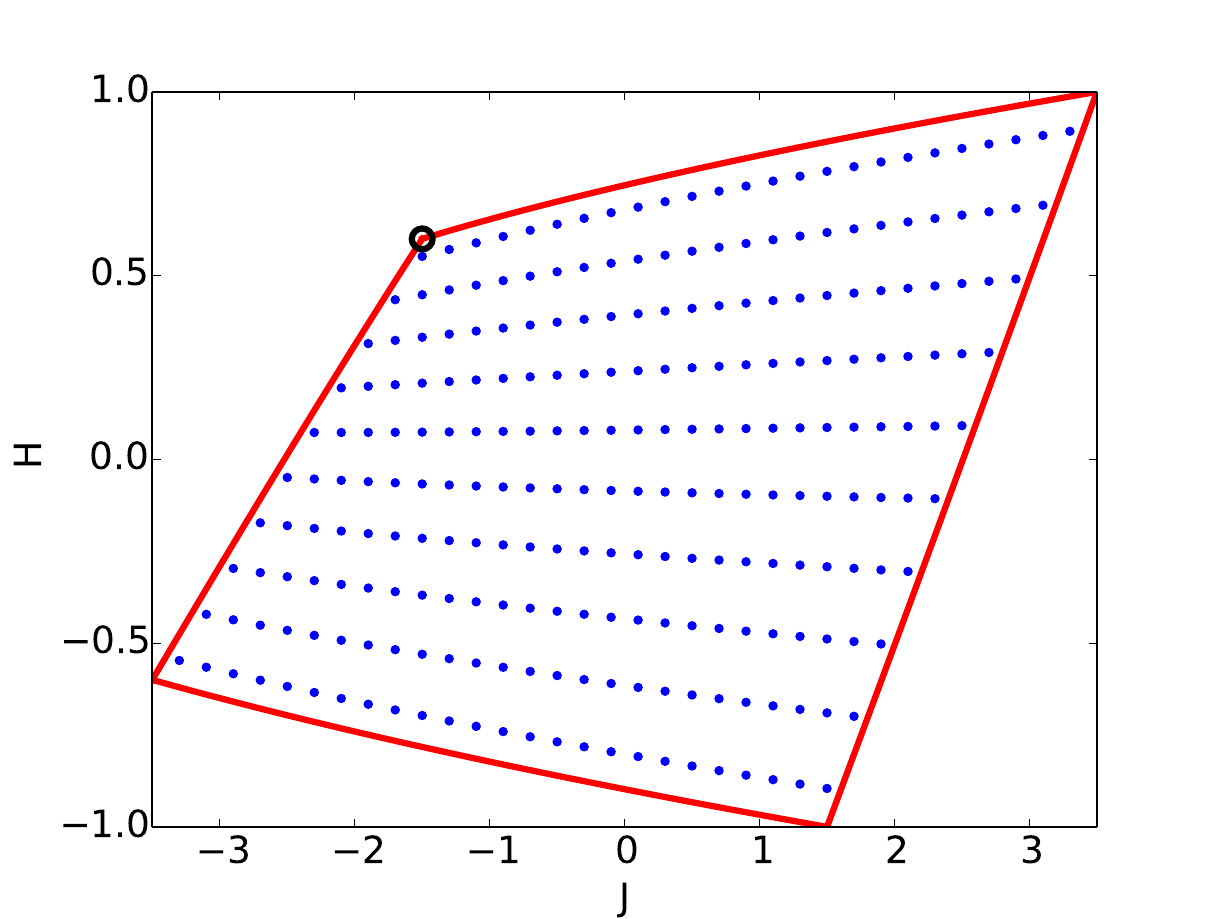} }
\subfigure[$t=t^- = \frac{5}{2(6+\sqrt{10})}$]{\includegraphics[scale=0.22]{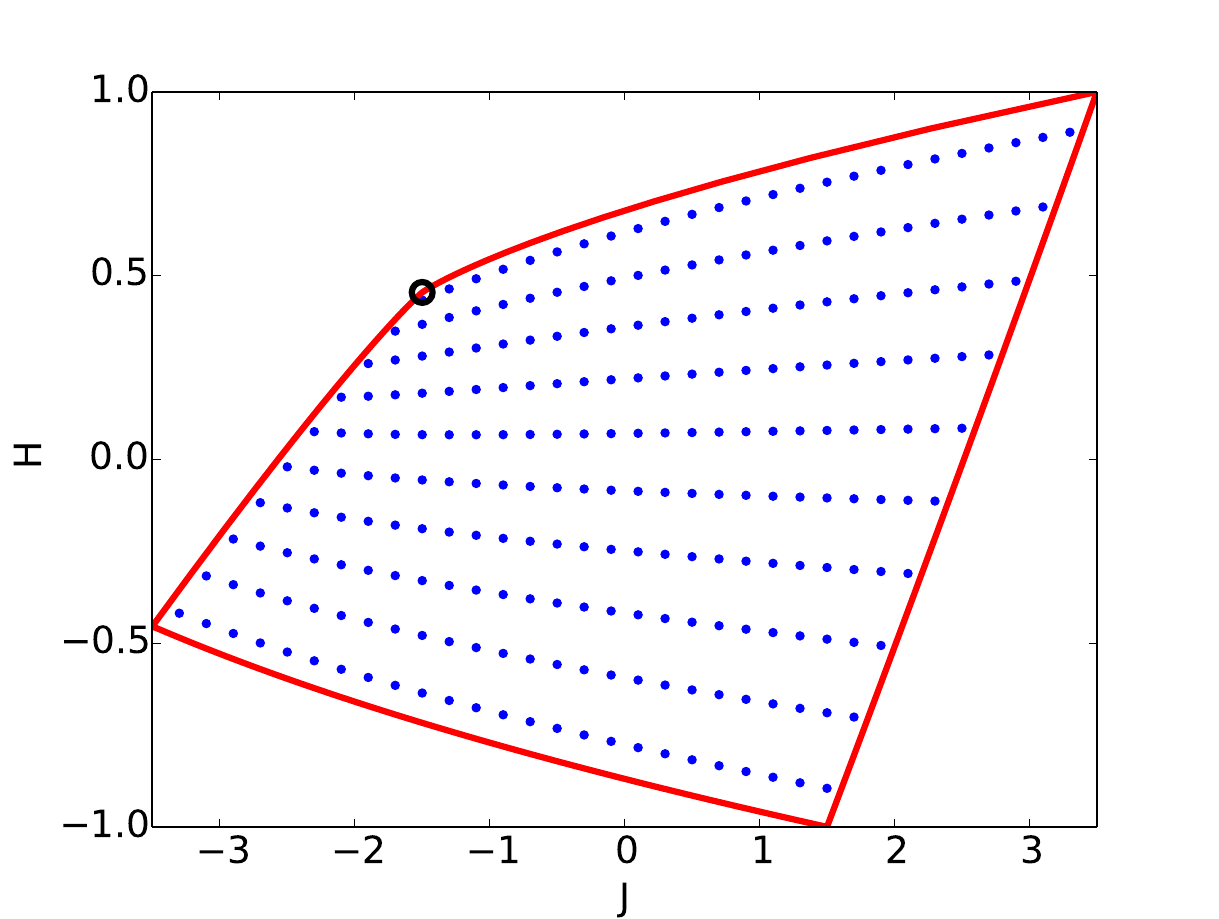} }
\subfigure[$t=0.4$]{\includegraphics[scale=0.22]{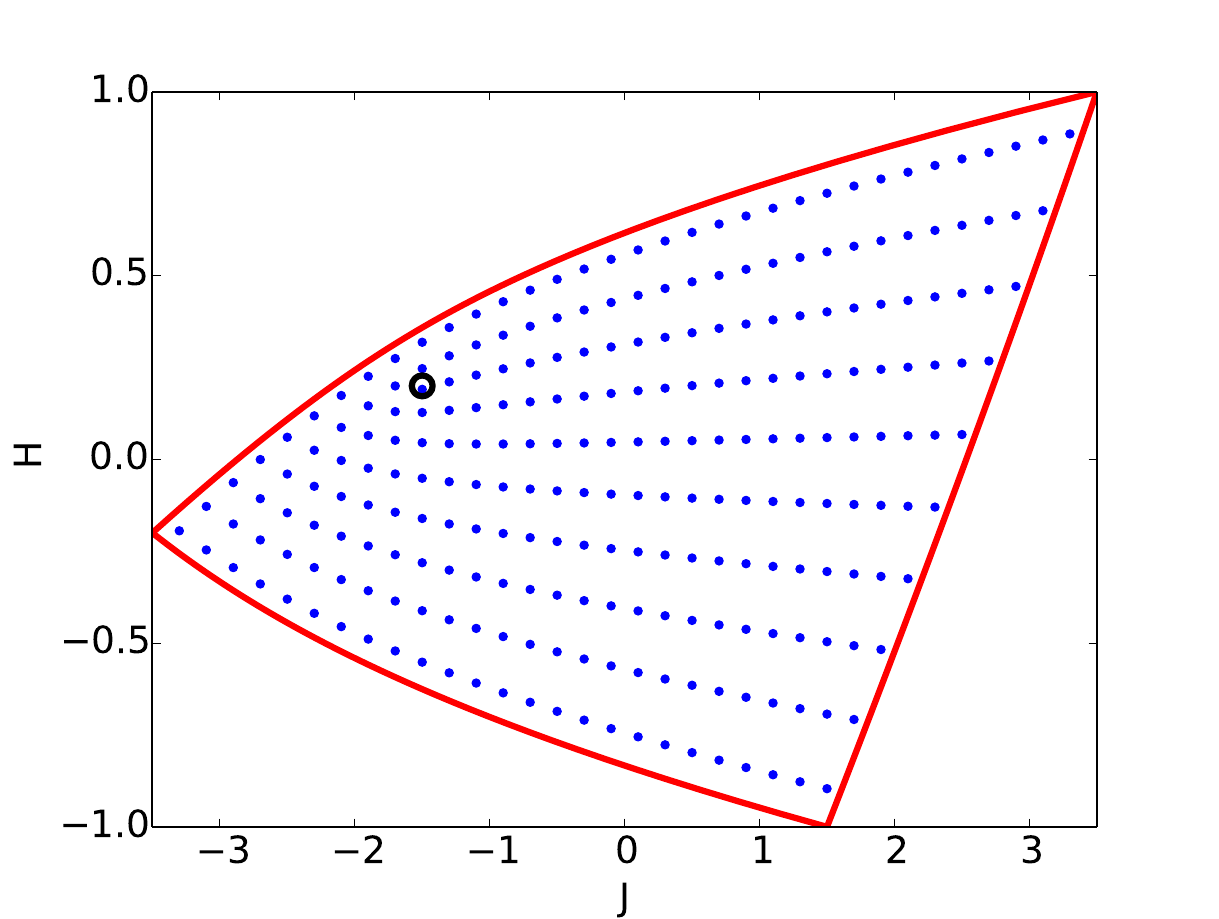} }
\subfigure[$t=0.5$]{\includegraphics[scale=0.22]{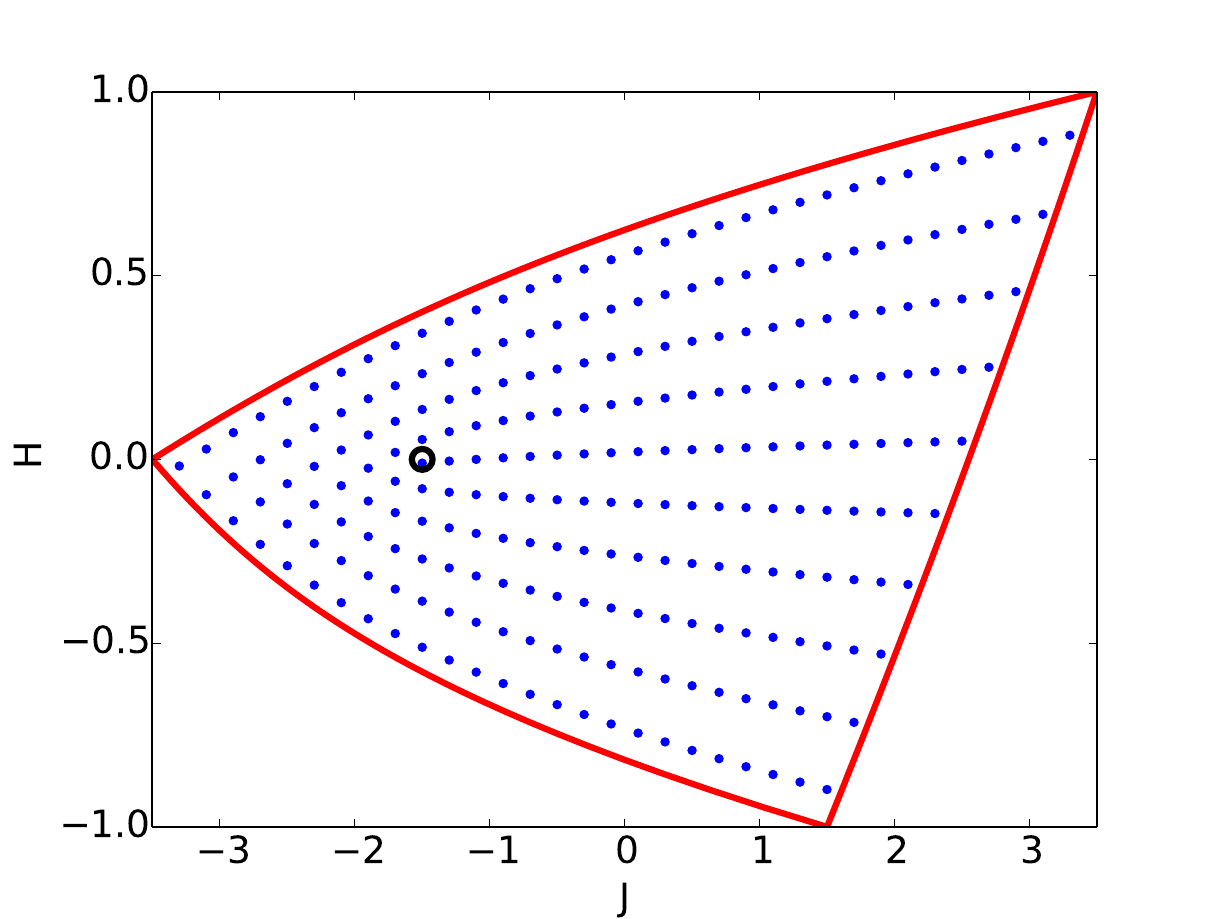} }
\subfigure[$t=0.7$]{\includegraphics[scale=0.22]{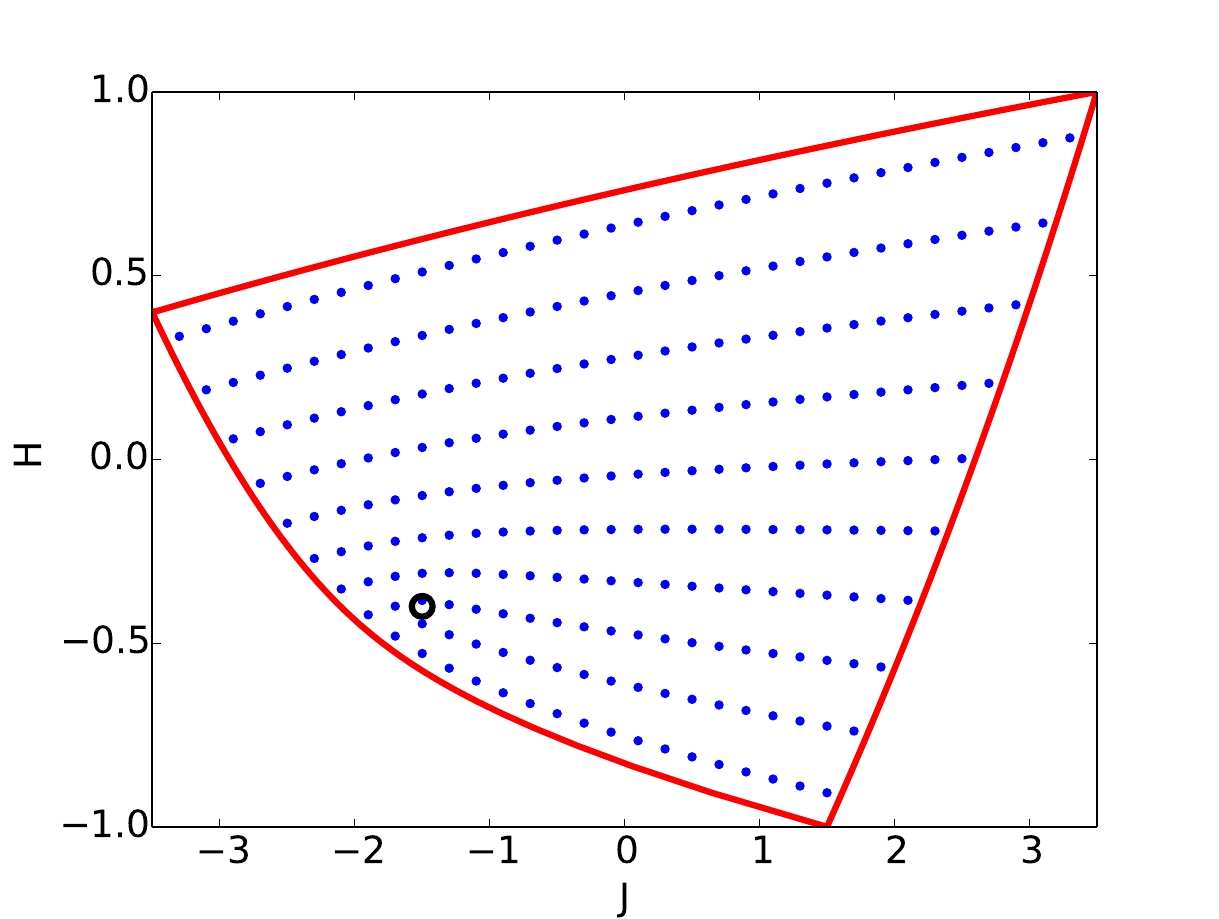} }
\subfigure[$t = t^+ = \frac{5}{2(6-\sqrt{10})}$]{\includegraphics[scale=0.22]{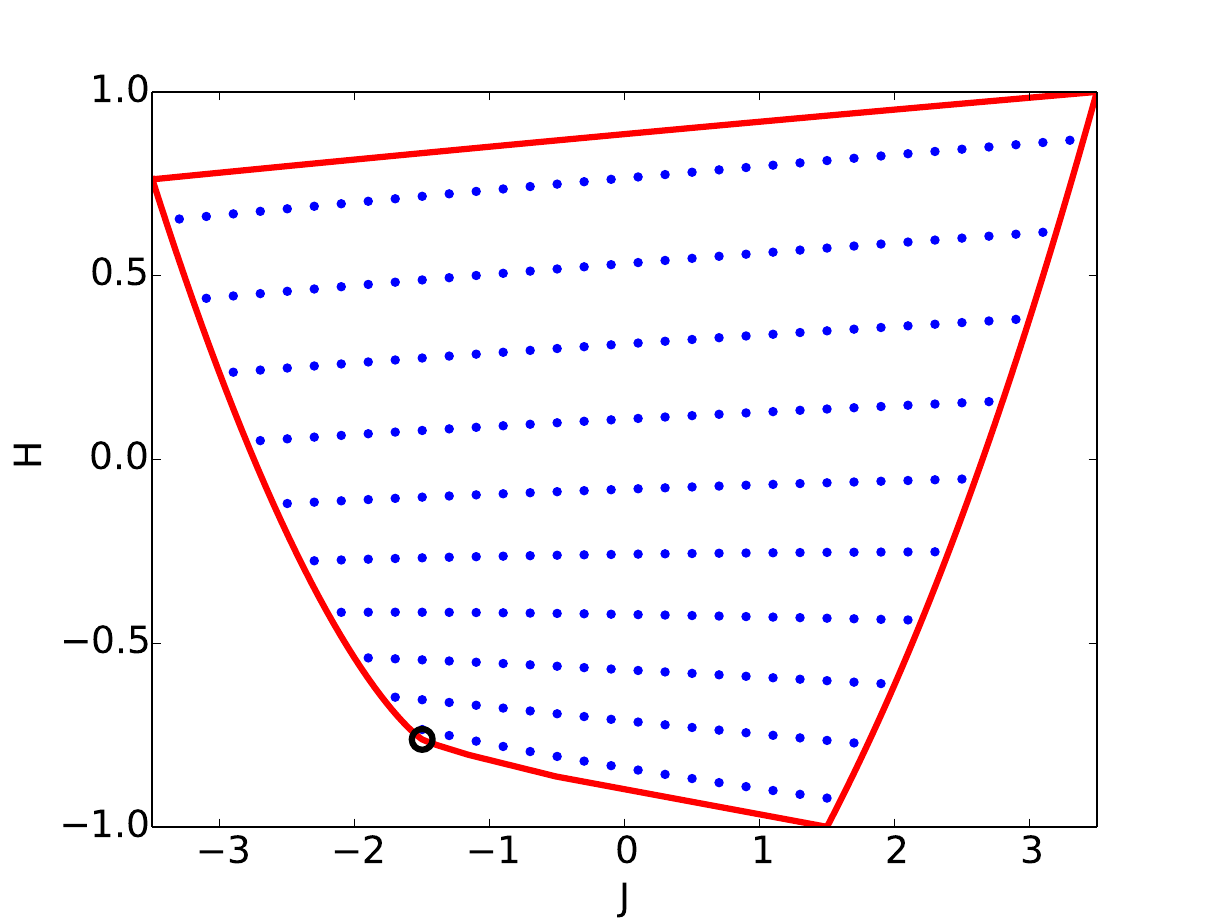} }
\subfigure[$t=0.9$]{\includegraphics[scale=0.22]{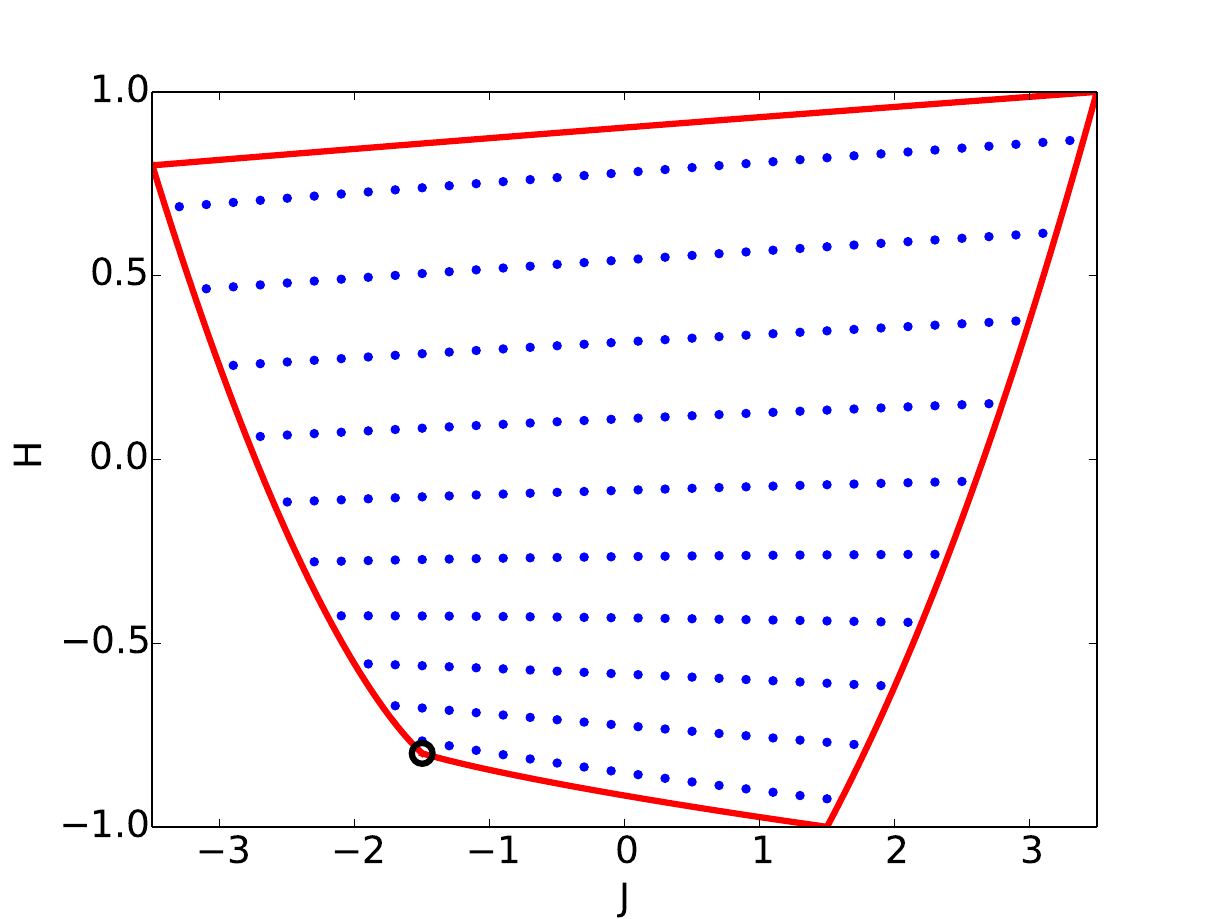} }
\subfigure[$t=1$]{\includegraphics[scale=0.22]{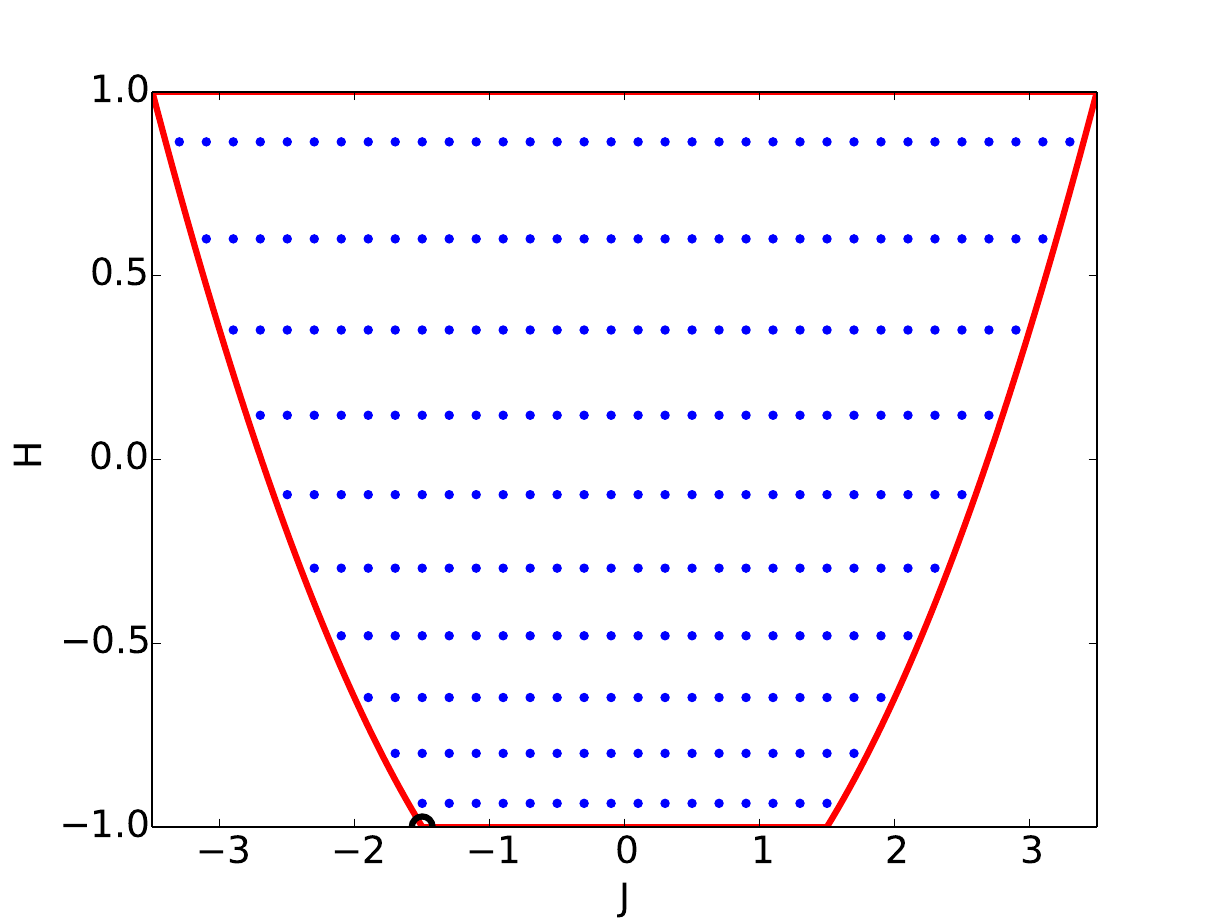} }
\end{center}
\caption{The blue dots constitute the joint spectrum of the quantized coupled angular momenta $(\hat{J}_k,\hat{H}_k)$ for different values of $t$ for $R_1 = 1$, $R_2 = 5/2$ and $k=5$. The red line indicates the boundary of $F(M)$; the center of the black circle corresponds to the critical value $c_0 = (-\tfrac{3}{2},1-2t)$.}
\label{fig:spectrum}
\end{figure}

\subsection{Conjectural Bohr-Sommerfeld rules}

In this section, we state a conjecture regarding the description of the joint spectrum of a pair of commuting self-adjoint Berezin-Toeplitz operators near a focus-focus critical value of the map formed by their principal symbols. We also give numerical evidence for this conjecture using the operators constructed above (or, if one believes in the conjecture, numerical methods to recover the coefficient $a_2$ and the height invariant). Note that a similar conjecture was stated in \cite[Conjecture $5.2$]{PelVN}, but we give here a more precise version taking into account the specificities of Berezin-Toeplitz operators. 

Let $(M,\omega)$ be a compact, connected four dimensional K\"ahler manifold, endowed with a prequantum line bundle $(L,\nabla) \to M$. In order to simplify the discussion, let us assume that there exists a half-form bundle $\delta \to M$, with a given line bundle isomorphism $\varphi: \delta^{\otimes 2} \to  \Lambda^{n,0} T^*M$; we could still state a conjecture without this assumption, but it would be more complicated (observe also that in the example above, such a half-form bundle exists). Let $A_k, B_k$ be two commuting self-adjoint Berezin-Toeplitz operators acting on $\Hil_k = H^0(M,L^{\otimes k} \otimes \delta)$, with respective principal symbols $f$ and $g$ and subprincipal symbols $r$ and $s$ (for a definition of the subprincipal symbol in this context, see \cite{ChaHalf}). Assume also that $(0,0)$ is a critical value of focus-focus type for $F = (f,g)$, and that there is a unique critical point $m_0$ on $F^{-1}(0,0)$. We want to describe the joint spectrum $\Sigma_k$ of $(A_k,B_k)$ near $(0,0)$.

Let $\gamma \subset M$ be an embedded closed curve. The \strong{principal action} of $\gamma$ is the number $c_0(\gamma) \in \Z/2\pi\Z$ such that the parallel transport along $\gamma$ in $(L,\nabla)$ is the multiplication by $\exp(ic_0(\gamma))$. We also define an index $\varepsilon(\gamma) \in \{0,1\}$ in the following way. Let $\delta_{\gamma}$ be the restriction of $\delta$ to $\gamma$, and consider the map $\varphi_{\gamma}: \delta_{\gamma} \to T^*\gamma \otimes \C, \ u \mapsto \iota^*\varphi(u)$ where $\iota$ is the embedding of $\gamma$ into $M$. It is an isomorphism of line bundles, and the set $ \{ u \in \delta_{\gamma} \ | \quad \varphi_{\gamma}(u^{\otimes 2}) > 0 \}$ has either one connected component, in which case we set $\varepsilon(\gamma) = 1$, or two connected components, in which case we set $\varepsilon(\gamma) = 0$. Finally, let $\gamma_0$ be a radial simple loop on $F^{-1}(0,0)$.

\begin{conj}
\label{conj:BS}
There exist sequences of functions $\lambda(t,k)$, $e^{(1)}(t,k)$, $e^{(2)}(t,k) \in \classe{\infty}{(\R^2,\R)}$ having asymptotic expansions of the form
\[ \lambda(t,k) = \sum_{\ell \geq -1} k^{-\ell} \lambda_{\ell}, \quad e^{(i)}(t,k) = \sum_{\ell \geq 0} k^{-\ell} {e_{\ell}^{(i)}} \]
for $i=1,2$, for the $\classe{\infty}$-topology, such that for every compact neighborhood $K$ of the origin in $\R^2$ and for every family $t = (t_1,t_2) \in K$, the pair $(k^{-1}t_1,k^{-1}t_2)$ belongs to $\Sigma_k + O(k^{-\infty})$ if and only if $e^{(1)}(t,k) = n + O(k^{-\infty})$ for some $n \in \Z$ and
\[ \lambda(t,k) + |n| \frac{\pi}{2} - e^{(2)}(t,k) \ln(2k^{-1}) - 2 \arg\left(\Gamma\left( \frac{i e^{(2)}(t,k) + |n| + 1}{2} \right) \right) \in 2\pi\Z + O(k^{-\infty}) . \]
Here, $\Gamma$ is the Gamma function. Furthermore,
\begin{itemize}
\item $\lambda_{-1}(t) = c_0(\gamma_0)$, \ $\lambda_0(t) = I_{\gamma_0}(t) + \varepsilon(\gamma_0) \pi$, where $I_{\gamma_0}$ is a regularized integral (the same as in the pseudodifferential case, see \cite[Corollary 6.10]{VN});
\item $\begin{pmatrix} e^{(1)}_0(t) \vspace{2mm} \\ e^{(2)}_0(t) \end{pmatrix} = B \begin{pmatrix} t_1 - r(m_0) \\ t_2 - s(m_0) \end{pmatrix}$ where $B$ stands for a $2 \times 2$ matrix such that $B \circ (d^2 f,d^2 g) = (q_1,q_2)$ near the focus-focus critical point (see Equation (\ref{eq:Eliasson_nf}) and the discussion preceding it).
\end{itemize}
\end{conj}

This conjecture seems very plausible because a similar result exists for pseudodifferential operators \cite[Theorem $7.5$]{VN}, the so-called Bohr-Sommerfeld conditions. Furthermore, Berezin-Toeplitz operators are microlocally equivalent to pseudodifferential operators, and Bohr-Sommerfeld conditions, similar to the ones for pseudodifferential operators, have been obtained for these operators near elliptic \cite{LFell} and hyperbolic points \cite{LFhyp} in one degree of freedom. Thus, there is very little doubt that the conjecture could be proved, in a straightforward but tedious way, by adapting the methods in \cite{VN}.

Now, assume for simplicity that for every integer $k \geq 1$, zero is an eigenvalue of $A_k$. Let $E_0^{(k)} \leq E_1^{(k)} \leq \ldots$ be the eigenvalues of the restriction of $B_k$ to $\ker A_k$. Following \cite[Theorem $7.6$]{VN} (see also \cite[Section $5.3$]{PelVN}), if true, the above conjecture would imply in particular that
\[ \min_p \left( k \left(E_{p+1}^{(k)} - E_p^{(k)}\right) \right) = \frac{2\pi \alpha}{\ln k + a_2 + \ln 2 + \gamma} + O(k^{-1}) \]
where $\alpha = \left\| B^{-1} \begin{pmatrix} 0 \\ 1 \end{pmatrix} \right\|$, $\gamma$ is the Euler-Mascheroni constant and $S^{\infty} = a_1 X + a_2 Y + \sum_{i+j > 1} b_{ij} X^i Y^j$ is the Taylor series defined in Section \ref{sect:symp_invariant}. Coming back to our system $(J,H)$ of coupled angular momenta with $t=1/2$ and its quantum counterpart $(\hat{J}_k,\hat{H}_k)$ (taking into account the fact that the focus-focus value is not $(0,0)$ but $(R_1 - R_2,0)$), it is clear from Lemma \ref{lm:ev_J} that $R_1 - R_2$ is always in the spectrum of $\hat{J_k}$, and $B$ is of the form
\[ B = \begin{pmatrix} 1 & 0 \\ a_{21} & a_{22} \end{pmatrix},  \]
so the above formula gives
\[ \min_p \left( k \left(E_{p+1}^{(k)} - E_p^{(k)}\right) \right) = \frac{2\pi}{|a_{22}|(\ln k + a_2 + \ln 2 + \gamma)} + O(k^{-1}). \]
Let us test this numerically on our example with $R_1 = 1$ and $R_2 = 5/2$; in this case, we know from Propositions \ref{prop:linear_eliasson} and \ref{prop:a2} that
\[  B = \begin{pmatrix} 1 & 0 \\ -\frac{1}{3} & \frac{10}{3} \end{pmatrix}, \quad a_2 = \frac{7}{2} \ln 2 + 3 \ln 3 - \frac{3}{2} \ln 5. \]
Consequently, the asymptotics
\[ \min_p \left( k \left(E_{p+1}^{(k)} - E_p^{(k)}\right) \right) = \frac{3\pi}{5(\ln k + \frac{9}{2} \ln 2 + 3 \ln 3 - \frac{3}{2} \ln 5 + \gamma)} + O(k^{-1}) \]
should hold; this is checked in Figure \ref{fig:gap}.

\begin{figure}
\begin{center}
\includegraphics[scale=0.4]{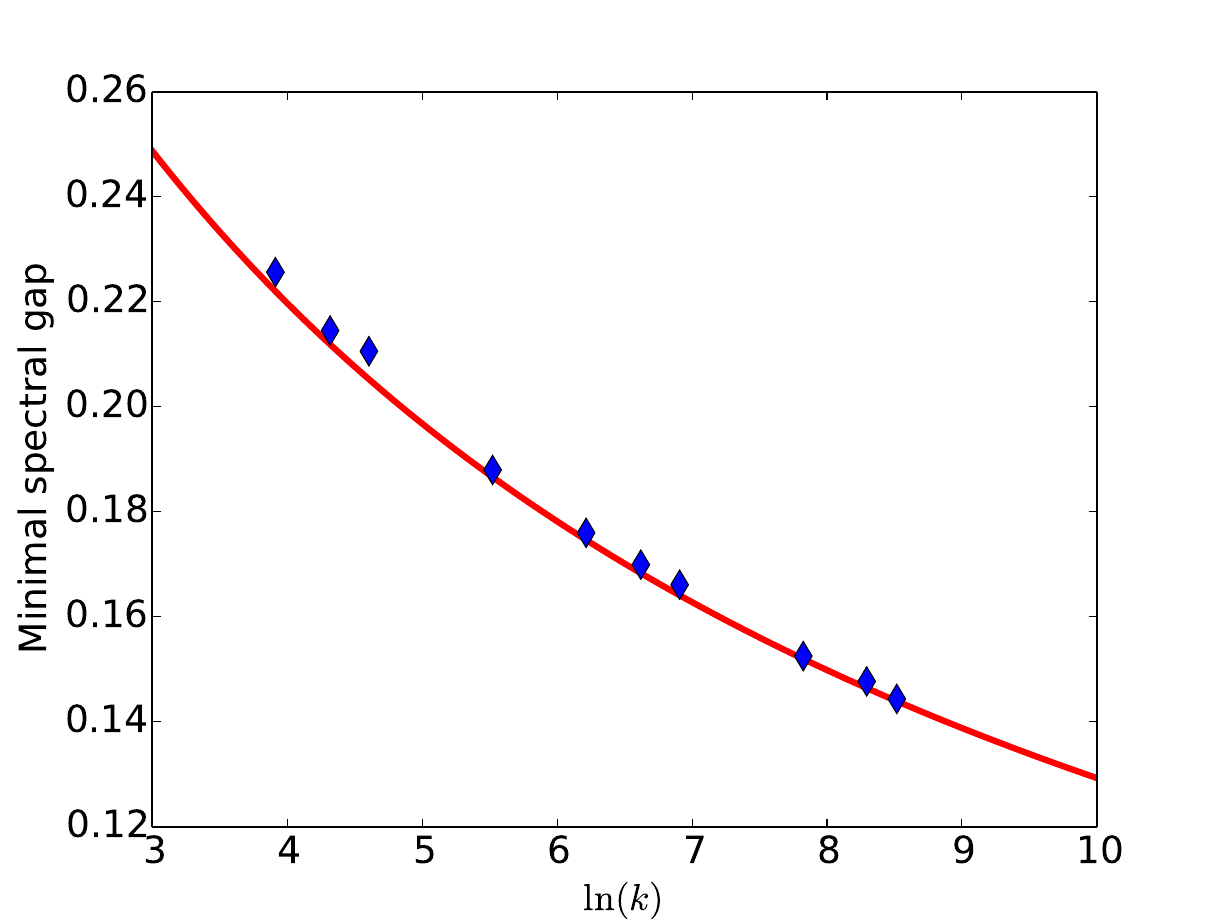}
\end{center}
\caption{The blue diamonds correspond to numerical computations of the minimal spectral gap $\min\left( k \left(E_{p+1}^{(k)} - E_p^{(k)}\right) \right)$ as a function of $\ln k$, for $t=1/2$, $R_1 = 1$ and $R_2 = 5/2$. The red line corresponds to the curve $\ln x \mapsto \frac{3\pi}{5\left(\ln x + \frac{9}{2} \ln 2 + 3 \ln 3 - \frac{3}{2} \ln 5 + \gamma\right)}$, giving the theoretical equivalent of this minimal spectral gap.}
\label{fig:gap}
\end{figure}

\subsection{Recovering the height invariant from the joint spectrum}

To conclude, let us briefly indicate how we can also use the joint spectrum to check that we have found the right formula for the height invariant $h$ in Proposition \ref{prop:height}. Indeed, assuming once again that Conjecture \ref{conj:BS} holds, one can use a  Weyl law to relate $h$ to the number of negative eigenvalues of the restriction of $\hat{H}_k$ to the kernel of $\hat{J}_k-(R_1 - R_2)\mathrm{Id}$. More precisely,
\begin{equation} h = \lim_{k \to +\infty} k^{-1} \#\left( \mathrm{Sp}\left( \hat{H}_{k  |\ker(\hat{J}_k-(R_1 - R_2)\mathrm{Id})} \right) \cap (-\infty,0) \right). \label{eq:weyl}\end{equation}
We check this formula in Figure \ref{fig:height}.

\begin{figure}
\begin{center}
\includegraphics[scale=0.4]{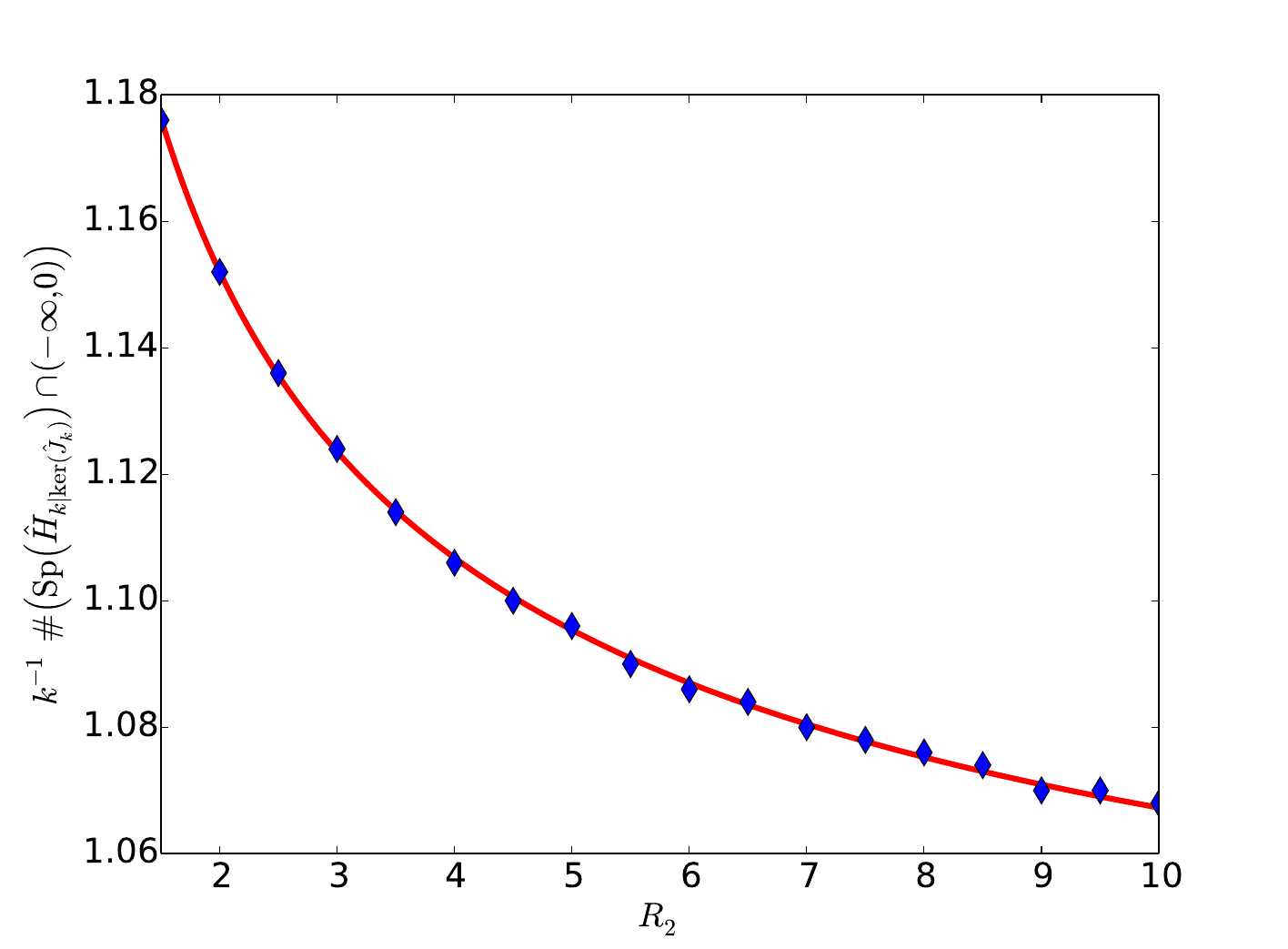}
\end{center}
\caption{Comparison of the two terms in Equation (\ref{eq:weyl}) for $k=500$, $t=1/2$, $R_1 = 1$ and $R_2 = 3/2,2,5/2, \ldots, 10$. The red line corresponds to the theoretical value of $h$ obtained in Proposition \ref{prop:height}, while the blue diamonds correspond to the numerical computation of the number of negative eigenvalues of the restriction of $\hat{H}_k$ to $\ker \hat{J}_k$, divided by $k$.}
\label{fig:height}
\end{figure}

\appendix

\section{Appendix: proofs of technical results}

\begin{proof}[Proof of Lemma \ref{lm:XJ_XH}]
A standard computation shows that the symplectic form becomes
\[ (\Psi^{-1})^* \omega = \frac{2 i R_1 dz \wedge d\bar{z}}{(1+|z|^2)^2} + \frac{2 i R_2 dw \wedge d\bar{w}}{(1+|w|^2)^2}. \]
Moreover, we deduce from Equation (\ref{eq:JH_comp}) that
\[ \frac{\partial J}{\partial z} = \frac{-2R_1 \bar{z}}{(1+|z|^2)^2}, \qquad \frac{\partial J}{\partial w} = \frac{2R_2 \bar{w}}{(1+|w|^2)^2}. \]
Since $J$ is real-valued, $\frac{\partial J}{\partial \bar{z}} = \overline{\frac{\partial J}{\partial z}}$ and $\frac{\partial J}{\partial \bar{w}} = \overline{\frac{\partial J}{\partial w}}$, and we obtain the desired result. Furthermore,
\[ \frac{\partial H}{\partial z} = \frac{1}{2} \left( \frac{-2\bar{z}}{(1+|z|^2)^2} + \frac{(2w-\bar{z}|w|^2+\bar{z})(1+|z|^2) - \bar{z}(2(zw + \bar{z}\bar{w}) + (1-|z|^2)(|w|^2-1))}{(1+|z|^2)^2(1+|w|^2)} \right), \]
which yields, after simplification
\[ \frac{\partial H}{\partial z} = \frac{w-2\bar{z}|w|^2-\bar{z}^2\bar{w}}{(1+|z|^2)^2(1+|w|^2)}. \]
A similar computation shows that
\[ \frac{\partial H}{\partial w} = \frac{z + \bar{w}-\bar{w}|z|^2-\bar{z}\bar{w}^2}{(1+|z|^2)(1+|w|^2)^2}, \]
and we conclude by using the same argument as above.
\end{proof}

\begin{proof}[Proof of Lemma \ref{lm:XH_Lambda}]
Let $\lambda_1(z,w)$, $\lambda_2(z,w)$ be as in the previous lemma, and let $(z,w) \in \Lambda_0 \setminus \{(0,0)\}$. Since $z \neq 0$, we get, by multiplying both the numerator and the denominator in $\lambda_1$ by $\bar{z}$, that
\[ \lambda_1(z,w) = \frac{\bar{z} \bar{w} - 2 |z|^2 |w|^2 - zw |z|^2}{R_1 \bar{z}(1+|w|^2)}. \]
But the second equation in (\ref{eq:lambda_0}) yields $\bar{z} \bar{w} = - zw - |w|^2 + |z|^2 |w|^2$, hence
\[ \lambda_1(z,w) = \frac{-zw - |w|^2 -  |z|^2 |w|^2 - zw |z|^2}{R_1 \bar{z}(1+|w|^2)} = -\frac{w(z+\bar{w})(1+|z|^2)}{R_1 \bar{z}(1+|w|^2)}. \]
The first equation in (\ref{eq:lambda_0}) allows us to further simplify this expression and to obtain
\[ \lambda_1(z,w) = -\frac{zw(z+\bar{w})(1+|z|^2)}{R_2 |w|^2 (1+|z|^2)} = -\frac{z(z+\bar{w})}{R_2 \bar{w}}. \]
Now, since $w \neq 0$, we have that
\[ \lambda_2(z,w) = \frac{\bar{z} |w|^2 + w |w|^2 - w |z|^2 |w|^2- z w^2 |w|^2}{R_2 |w|^2 (1+|z|^2)} = \frac{\bar{z} |w|^2 - z w^2 |w|^2 + w |w|^2 (1 -|z|^2)}{R_2 |w|^2 (1+|z|^2)}. \]
The second equation in (\ref{eq:lambda_0}) gives
\[ \bar{z} |w|^2 - z w^2 |w|^2 + w |w|^2 (1 -|z|^2) = \bar{z} |w|^2 - z w^2 |w|^2 - z w^2 - \bar{z} |w|^2 = -zw^2(1+|w|^2),  \]
thus, using the first one again, we finally obtain that $ \lambda_2(z,w) = -\frac{w^2}{R_1 \bar{z}}$.
\end{proof}

\begin{proof}[Proof of Proposition \ref{prop:linear_eliasson}]
We start by diagonalizing $A$ over $\C$. With our choice of parameters, it follows from Equation (\ref{eq:hessians}) that
\[ A = \begin{pmatrix} 0 & 0 & 0 & \frac{1}{2} \\ 0 & 0 & -\frac{1}{2} & 0 \\ 0 & -\frac{1}{5} & 0 & -\frac{1}{5} \\ \frac{1}{5} & 0 & \frac{1}{5} & 0  \end{pmatrix}. \]
$A$ has eigenvalues $\lambda_1 = \frac{3+i}{10}, \bar{\lambda}_1, \lambda_2 = - \lambda_1, \bar{\lambda}_2$ with respective eigenvectors $X_1, \overline{X_1}, X_2, \overline{X_2}$ where
\[ X_1 = \frac{1}{2} \begin{pmatrix} 3 - i \\ -1 - 3i \\ 2i \\  2 \end{pmatrix} , \qquad X_2 = \frac{1}{2} \begin{pmatrix} -3 + i \\ - 1 - 3i \\ -2i \\ 2 \end{pmatrix}. \]
Let $E_{\lambda}$ be the eigenspace associated with $\lambda$, and let $F = E_{\lambda_1} \oplus E_{\bar{\lambda}_1}$ and $G = E_{\lambda_2} \oplus E_{\bar{\lambda}_2}$; then $T_{m_0}M = F \oplus G$ and a real basis of $F$ is given by
\[ Y_1 = X_1 + \overline{X_1} = \begin{pmatrix} 3 \\ -1 \\ 0 \\ 2 \end{pmatrix} , \qquad Y_2 = -i(X_1 - \overline{X_1}) = \begin{pmatrix} -1 \\ -3 \\ 2 \\ 0 \end{pmatrix}.  \]
There exists a unique basis $(Z_1,Z_2)$ of $G$ such that $(Y_1,Y_2,Z_1,Z_2)$ is a symplectic basis of $T_{m_0}M$ (see for instance \cite[Lemma $3.2.3$]{Meyer}); in the latter, $A$ will take the form displayed in Equation (\ref{eq:mat_focus}). Since
\[ G = \mathrm{Span}\left( \begin{pmatrix} -3 \\ -1 \\ 0 \\ 2 \end{pmatrix}, \begin{pmatrix} 1 \\ -3 \\ -2 \\ 0 \end{pmatrix} \right), \]
there exists $a,b,c,d \in \R$ such that
\[ Z_1 = \begin{pmatrix} -3a + b \\ -a - 3b \\ -2b \\ 2a \end{pmatrix} , \qquad Z_2 = \begin{pmatrix} -3c + d \\ -c - 3d \\ -2d \\ 2c \end{pmatrix}.  \]
Recall that $\omega_{m_0} = - dx_1 \wedge dy_1 + \frac{5}{2} dx_2 \wedge dy_2$. Hence we need to solve the system
\[ \begin{cases} 1 = \omega_{m_0}(Y_1,Z_1) = 6a + 18b \\ 0 = \omega_{m_0}(Y_1,Z_2) = 6c + 18d \\ 0 = \omega_{m_0}(Y_2,Z_1) = 18a - 6b \\ 1 = \omega_{m_0}(Y_2,Z_2) = 18c - 6d \end{cases}. \]
We find $a = \frac{1}{60}, b = \frac{1}{20}, c = \frac{1}{20}$ and $d = -\frac{1}{60}$, hence
\[ Z_1 = \frac{1}{30} \begin{pmatrix} 0 \\ -5 \\ -3 \\ 1 \end{pmatrix} , \qquad Z_2 = \frac{1}{30} \begin{pmatrix} -5 \\ 0 \\ 1 \\ 3 \end{pmatrix}.  \]
Consequently, the matrix $P$ of change of basis from the basis $\mathcal{B}$ associated with $(x_1,y_1,x_2,y_2)$ to the basis $(Y_1,Y_2,Z_1,Z_2)$ satisfies
\[ P = \begin{pmatrix} 3 & -1 & 0 & -\frac{1}{6} \vspace{2mm}\\ -1 & -3 & - \frac{1}{6} & 0 \vspace{2mm}\\ 0 & 2 & -\frac{1}{10} & \frac{1}{30} \vspace{2mm}\\ 2 & 0 & \frac{1}{30} & \frac{1}{10} \end{pmatrix}, \quad P^{-1} = \begin{pmatrix} \frac{1}{6} & 0 & \frac{1}{12} & \frac{1}{4} \vspace{2mm}\\ 0 & -\frac{1}{6} & \frac{1}{4} & -\frac{1}{12} \vspace{2mm}\\ -1 & -3 & -5 & 0 \vspace{2mm}\\ -3 & 1 & 0 & 5 \end{pmatrix}, \]
which yields the desired expression for the coordinates $u_1,u_2,\xi_1,\xi_2$. Moreover, $P^{-1}AP$ is as in Equation $(\ref{eq:mat_focus})$, with $\alpha = \frac{-3}{10}$ and $\beta =\frac{1}{10}$; this gives 
\[ B = \begin{pmatrix} 1 & 0 \\ \frac{1}{10} & -\frac{3}{10} \end{pmatrix}^{-1} = \begin{pmatrix} 1 & 0 \\ \frac{1}{3} & -\frac{10}{3} \end{pmatrix}. \]
This is not satisfactory since we want the lower right coefficient in this matrix to be positive. In order to obtain a $B$ satisfying this requirement, it suffices to perform the symplectic change of coordinates $(u_1,u_2,\xi_1,\xi_2) \mapsto (-\xi_1,-\xi_2, u_1, u_2)$.
\end{proof}

\begin{proof}[Proof of Lemma \ref{lm:compute_int}]
Let $g(\rho) = \arccos\left( \frac{\rho^2-1}{2} \sqrt{\frac{R_1}{R_2+(R_2-R_1)\rho^2}}\right) = \arccos\left( \frac{\rho^2-1}{2} \sqrt{\frac{1}{\Theta+(\Theta-1)\rho^2}}\right)$; a tedious but straightforward computation yields
\[ g'(\rho) =  \frac{-((\Theta-1)\rho^3 + (3\Theta-1)\rho)}{(\Theta+(\Theta-1)\rho^2)\sqrt{(1+\rho^2)(4\Theta-1-\rho^2)}}, \]
and therefore an integration by parts leads to
\[ I = \frac{1}{2} \arccos\left(\frac{-1}{2\sqrt{\Theta}}\right) - \frac{J}{2} = \frac{\pi}{2} - \frac{1}{2} \arccos\left(\frac{1}{2\sqrt{\Theta}}\right) - \frac{J}{2}  \]
where $J$ is the integral
\[ J = \int_0^{\sqrt{4\Theta-1}} \frac{(\Theta-1)\rho^3 + (3\Theta-1)\rho}{(1+\rho^2)(\Theta+(\Theta-1)\rho^2) \sqrt{(1+\rho^2)(4\Theta-1-\rho^2)}} \ d\rho. \]
Now, one can check that
\[ \frac{(\Theta-1)\rho^2 + 3\Theta-1}{(1+\rho^2)(\Theta+(\Theta-1)\rho^2)} = \frac{2\Theta}{1+\rho^2} - \frac{(2\Theta-1)(\Theta-1)}{\Theta+(\Theta-1)\rho^2}, \]
and consequently $J = K - L$ with
\[ K = \int_0^{\sqrt{4\Theta-1}} \frac{2\Theta \rho \ d\rho}{(1+\rho^2)^{3/2} \sqrt{4\Theta-1-\rho^2}}, \ L = \int_0^{\sqrt{4\Theta-1}} \frac{(2\Theta-1)(\Theta-1)\rho \ d\rho}{(\Theta+(\Theta-1)\rho^2)\sqrt{(1+\rho^2)(4\Theta-1-\rho^2)}} .  \]
One readily checks that
\[ K = - \frac{1}{2} \int_0^{\sqrt{4\Theta-1}} h'(\rho) \ d\rho, \quad h(\rho) = \sqrt{\frac{4\Theta-1-\rho^2}{1+\rho^2}},  \]
and therefore $K = \frac{1}{2} \sqrt{4\Theta-1}$. Regarding $L$, we start by making the change of variables $u = \rho^2$:
\[ L = \frac{(2\Theta-1)(\Theta-1)}{2} \int_0^{4\Theta-1} \frac{du}{(\Theta+(\Theta-1)u)\sqrt{4\Theta^2-(u+1-2\Theta)^2}}. \]
Now, we perform the changes of variables $v = u+1-2\Theta$ and $v = 2\Theta \sin \theta$ to get
\[ L = \frac{(2\Theta-1)(\Theta-1)}{2} \int_{\arcsin\left(\frac{1}{2\Theta}-1\right)}^{\frac{\pi}{2}} \frac{d\theta}{2\Theta(\Theta-1) \sin \theta + 2 \Theta^2 - 2\Theta + 1}. \]
Setting $w = \tan(\theta/2)$, and using the fact that $\tan\left(\frac{1}{2}\arcsin\left(\frac{1}{2\Theta}-1\right)\right) = \frac{2\Theta - \sqrt{4\Theta-1}}{1-2\Theta}$ yields
\[ L = (2\Theta-1)(\Theta-1) \int_{ \frac{2\Theta - \sqrt{4\Theta-1}}{1-2\Theta}}^{1} \frac{dw}{4\Theta(\Theta-1) w + (1+w^2)(2 \Theta^2 - 2\Theta + 1)}.   \]
This can be rewritten as
\[ L = (2\Theta-1)(\Theta-1) \int_{ \frac{2\Theta - \sqrt{4\Theta-1}}{1-2\Theta}}^{1} \frac{dw}{\left( w \sqrt{2\Theta^2-2\Theta+1} + \frac{2\Theta(\Theta-1)}{\sqrt{2\Theta^2-2\Theta+1} } \right)^2 + \frac{4\Theta^2-4\Theta+1}{2\Theta^2-2\Theta+1} }. \]
Setting $z = w \sqrt{2\Theta^2-2\Theta+1} + \frac{2\Theta(\Theta-1)}{\sqrt{2\Theta^2-2\Theta+1} }$, we finally arrive at
\[ L = (\Theta-1) \int_{ \frac{2\Theta^2 - (2\Theta^2-2\Theta+1) \sqrt{4\Theta-1}}{1+2\Theta}}^{2\Theta-1} \frac{dw}{1+w^2}. \]
This means that
\[ L = (\Theta-1) \left( \arctan(2\Theta-1)  - \arctan\left( \frac{(2\Theta^2-2\Theta+1)\sqrt{4\Theta-1}-2\Theta^2}{(2\Theta-1)^2} \right)\right), \]
and we obtain that 
\[ I = \frac{\pi}{2} - \frac{1}{2} \arccos\left(\frac{1}{2\sqrt{\Theta}}\right)  - \frac{\sqrt{4\Theta-1}}{4}  + \left(\frac{\Theta-1}{2} \right) \left( \arctan(2\Theta-1)   -  \arctan\left( \frac{(2\Theta^2-2\Theta+1)\sqrt{4\Theta-1}-2\Theta^2}{(2\Theta-1)^2} \right) \right). \]
We conclude the proof by using the arctan addition formula.
\end{proof}

\section{Critical points of corank one}

We explain how to prove the claim about the critical points of corank one of $F=(J,H_t)$ in the proof of Corollary \ref{cor:semitor}, namely that they are non-degenerate for every $t \in ]0,1]$ (the case $t=0$ is clear since the system is toric up to vertical scaling). It suffices to prove that for every $E$ in the image of $J$, except the ones corresponding to critical points of corank two, the critical points of the restriction of $H_t$ to the symplectic quotient $M^{\text{red}}_E = J^{-1}(E) / S^1$ with respect to the action generated by $J$ are non-degenerate. Although this is a folk result, it seems that a proof only appeared very recently in the literature \cite[Corollary 2.5]{HohPal}. Coming back to our particular case, let $E \in (-(R_1 + R_2), R_1 + R_2) \setminus \{ R_1 - R_2, R_2 - R_1 \}$; since the poles do not give rise to critical points on $J^{-1}(E)$, we may work with cylindrical coordinates as in Section \ref{sect:height} (or as in \cite[Section 3.3]{HohPal} where a similar computation is performed)
\[ (x_j, y_j, z_j) = \left( \sqrt{1-z_j^2} \cos \theta_j,  \sqrt{1-z_j^2} \sin \theta_j, z_j \right), \quad j = 1,2. \] 
In these coordinates, $H_t$ reads
\[ H_t(\theta_1, z_1, \theta_2, z_2) = (1-t) z_1 + t \left( \sqrt{(1-z_1^2)(1-z_2^2)} \cos(\theta_1 - \theta_2) + z_1 z_2 \right).  \]
Since $z_2$ can be deduced from $z_1$ on $J^{-1}(E)$, namely 
\[ z_2 = \frac{E-R_1 z_1}{R_2}, \qquad \max\left(-1, \frac{E - R_2}{R_1} \right) < z_1 < \min\left(1, \frac{E + R_2}{R_1} \right) \]
and since the action of $J$ preserves the angle $\theta = \theta_1 - \theta_2$, we can use $(z_1, \theta)$ as coordinates on $M^{\text{red}}_E$:
\[ H_t(z_1,\theta) = (1-t) z_1 + \frac{t z_1(E-R_1 z_1)}{R_2} + \frac{t \cos \theta}{R_2} \sqrt{P_E(z_1)}, \quad P_E(z_1) = (1-z_1^2)\left(R_2^2 - (E - R_1 z_1)^2 \right).  \]
The first derivatives of $H_t$ read
\[ \frac{\partial H_t}{\partial \theta}(z_1,\theta) = \frac{-t \sqrt{P_E(z_1)}\sin \theta}{R_2}, \quad  \frac{\partial H_t}{\partial z_1}(z_1,\theta) = 1-t + \frac{tE}{R_2} - \frac{2 t R_1 z_1}{R_2} + \frac{t P_E'(z_1) \cos \theta}{2 R_2 \sqrt{P_E(z_1)}}. \]
Hence if $(z_1^*,\theta^*)$ is a critical point, then necessarily $\theta^* \in \{0, \pi\}$, and $\frac{\partial^2 H_t}{\partial \theta \partial z_1}(z_1^*,\theta^*) = 0$. Let $\varepsilon = \cos \theta^* \in \{-1,1\}$; then one readily checks that 
\[ \frac{\partial^2 H_t}{\partial \theta^2}(z_1^*,\theta^*) = \frac{-\varepsilon t \sqrt{P_E(z_1^*)}}{R_2}, \quad \frac{\partial^2 H_t}{\partial z_1^2}(z_1^*,\theta^*) = \frac{t}{4 R_2} \left( - 8 R_1 + \varepsilon \left( \frac{2 P_E''(z_1^*) P_E(z_1^*) - P_E'(z_1^*)^2}{P_E(z_1^*)^{3/2}} \right) \right) \] 
We claim that this last quantity has the sign of $-\varepsilon$; this follows from the fact that
\[ Q(E,z_1) = \frac{2 P_E''(z_1) P_E(z_1) - P_E'(z_1)^2}{P_E(z_1)^{3/2}} < - 8 R_1 \]  
for any $E, z_1$ satisfying the above bounds. In order to prove this, one may check that $Q$ is minimal at $(E,z_1) = (0,0)$; since $Q(0,0) = -4 \left( R_2 + \frac{R_1^2}{R_2} \right)$ and since the function $x > 0 \mapsto x + \frac{R_1^2}{x}$ is minimal at $x = R_1$ with value $2 R_1$, we obtain the desired result because $R_2 > R_1$.

In fact, this analysis gives us the sign of the determinant of the Hessian of $H_t$ at a critical point, so we can deduce from it that the corank one critical points are of elliptic-transverse type. Hence if one is only interested in proving this, and not in finding a parametrization of the boundary of the image of the momentum map, this appendix constitutes a faster way to obtain Corollary \ref{cor:semitor}.

\bibliographystyle{abbrv}
\bibliography{angular_momenta}

\vspace{1cm}

\begin{minipage}{0.50\linewidth}
     {\bf Yohann Le Floch} \\
Institut de Recherche Math\'ematique Avanc\'ee\\
UMR 7501, Universit\'e de Strasbourg et CNRS\\
67000 Strasbourg, France\\
{\em E\--mail:} \texttt{ylefloch@unistra.fr}\\
      
   \end{minipage}\hfill
   \begin{minipage}{0.50\linewidth}   
{\bf {{\'A}lvaro Pelayo}}\\
Department of Mathematics\\
University of California, San Diego\\
9500 Gilman Drive  \# 0112\\
La Jolla, CA 92093-0112, USA\\
{\em E\--mail}: \texttt{alpelayo@math.ucsd.edu}
\end{minipage}

\end{document}